\documentclass[11pt]{article}

\usepackage[margin=1in]{geometry}
\usepackage[T1]{fontenc}
\usepackage[utf8]{inputenc}
\usepackage{amsmath,amssymb,amsthm,mathtools,mathrsfs}
\usepackage{array,booktabs,longtable,tabularx}
\usepackage{enumitem}
\usepackage[expansion=false]{microtype}
\usepackage{xcolor}
\usepackage{tikz}
\usetikzlibrary{arrows.meta}
\usepackage{hyperref}
\usepackage{ifthen}
\usepackage{aliascnt}
\usepackage[capitalize,noabbrev]{cleveref}
\usepackage[backend=biber,style=alphabetic,maxbibnames=6]{biblatex}
\hypersetup{
  colorlinks=true,
  linkcolor=blue!55!black,
  citecolor=blue!55!black,
  urlcolor=blue!55!black
}

\newtheorem{theorem}{Theorem}[section]
\newaliascnt{lemma}{theorem}
\newtheorem{lemma}[lemma]{Lemma}
\aliascntresetthe{lemma}
\newaliascnt{proposition}{theorem}
\newtheorem{proposition}[proposition]{Proposition}
\aliascntresetthe{proposition}
\newaliascnt{corollary}{theorem}
\newtheorem{corollary}[corollary]{Corollary}
\aliascntresetthe{corollary}
\newaliascnt{external}{theorem}

\aliascntresetthe{external}
\newaliascnt{construction}{theorem}
\newtheorem{construction}[construction]{Construction}
\aliascntresetthe{construction}
\theoremstyle{definition}
\newaliascnt{definition}{theorem}
\newtheorem{definition}[definition]{Definition}
\aliascntresetthe{definition}
\newaliascnt{convention}{theorem}
\newtheorem{convention}[convention]{Convention}
\aliascntresetthe{convention}
\newaliascnt{requirement}{theorem}
\newtheorem{requirement}[requirement]{Verification requirement}
\aliascntresetthe{requirement}
\newaliascnt{example}{theorem}

\aliascntresetthe{example}
\theoremstyle{remark}
\newaliascnt{remark}{theorem}
\newtheorem{remark}[remark]{Remark}
\aliascntresetthe{remark}

\crefname{lemma}{lemma}{lemmas}
\Crefname{lemma}{Lemma}{Lemmas}
\crefname{proposition}{proposition}{propositions}
\Crefname{proposition}{Proposition}{Propositions}
\crefname{corollary}{corollary}{corollaries}
\Crefname{corollary}{Corollary}{Corollaries}
\crefname{external}{external input}{external inputs}
\Crefname{external}{External input}{External inputs}
\crefname{construction}{construction}{constructions}
\Crefname{construction}{Construction}{Constructions}
\crefname{definition}{definition}{definitions}
\Crefname{definition}{Definition}{Definitions}
\crefname{convention}{convention}{conventions}
\Crefname{convention}{Convention}{Conventions}
\crefname{requirement}{verification requirement}{verification requirements}
\Crefname{requirement}{Verification requirement}{Verification requirements}
\crefname{example}{example}{examples}
\Crefname{example}{Example}{Examples}
\crefname{remark}{remark}{remarks}
\Crefname{remark}{Remark}{Remarks}

\newcommand{\CC}{\mathbb C}
\newcommand{\RR}{\mathbb R}
\newcommand{\ZZ}{\mathbb Z}
\newcommand{\QQ}{\mathbb Q}
\newcommand{\FF}{\mathbb F}

\newcommand{\Cl}{\operatorname{Cl}}
\newcommand{\Sk}{\operatorname{Sk}}
\newcommand{\Cone}{\operatorname{Cone}}
\newcommand{\Cyl}{\operatorname{Cyl}}
\newcommand{\im}{\operatorname{im}}
\newcommand{\rank}{\operatorname{rank}}
\newcommand{\Span}{\operatorname{span}}
\newcommand{\dep}{\operatorname{dep}}
\newcommand{\spec}{\operatorname{spec}}
\newcommand{\Proj}{\operatorname{Proj}}
\newcommand{\Dbase}{D_{\mathrm{base}}}
\newcommand{\Dpriv}{D_{\mathrm{priv}}}
\newcommand{\Qenc}{\mathscr Q}
\newcommand{\Eenc}{\mathscr E}
\newcommand{\Bbulk}{\mathscr B}
\newcommand{\Odown}{\mathscr O^{\downarrow}}
\newcommand{\Oup}{\mathscr O^{\uparrow}}

\title{A Complete Proof of $\mathsf{QMA}_1$-Hardness for Weighted Gapped Clique Homology}
\author{Tal Barak\thanks{The author acknowledges support by ERC Advanced Grant CCnoisyQC (grant number 101201649).}}
\date{October 2, 2026}

\begin{document}
\maketitle

\begin{abstract}
The computational complexity of simplicial homology has been studied since
the question was raised by Kaibel and Pfetsch~\cite{KaibelPfetsch}, and
clique homology was later shown to be $\mathsf{QMA}_1$-hard by Crichigno
and Kohler~\cite{CrichignoKohler2024}. King and Kohler subsequently
introduced the weighted, spectrally gapped version considered here and
developed a reduction intended to establish the same hardness result in
this setting~\cite{KingKohler2026}.

We give a complete proof of the hard direction.  The earlier King-Kohler
argument contains a gap in the analysis of the nontrivially folded local
gadgets, so the required topological and spectral properties have to be
re-established on the clique complexes actually produced by the reduction.
This is not a local repair of the published proof: the argument requires a
new analysis of the local topology, the filtered cochain complexes, and the
interaction of many overlapping gadgets.

Our proof uses exact finite certificates for the local combinatorial data
and relative homology to determine the effect of each gadget on the encoded
qubit space.  We then prove the finite-dimensional filtered-Hodge statement
needed to recover the corresponding rank-one low-energy penalty and derive
a quantitative global estimate that transfers the source-Hamiltonian
promise gap to the final weighted Hodge Laplacian.  Consequently, for fixed
efficiently computable $k(n)$ and inverse-polynomial $\gamma(n)$, weighted
gapped clique homology is $\mathsf{QMA}_1(\mathcal G)$-hard.  The previously
known $\mathsf{QMA}$-containment argument is logically independent of the
gap and remains unaffected.
\end{abstract}

\tableofcontents

\section{Introduction}
\label{sec:introduction}

\subsection{From the complexity of holes to quantum complexity}

Homology is among the most basic invariants in topology.  Given a space,
it detects patterns with no boundary, such as loops in dimension one or
closed surfaces in dimension two, and asks which of them can be filled by
higher-dimensional pieces.  Such boundaryless patterns are called cycles;
those that can be filled are boundaries.  Homology records, in each
dimension, the cycles that are not boundaries.

For a finite simplicial complex whose simplices are explicitly listed,
computing homology is ultimately a problem in finite linear algebra: one
writes down the boundary matrices and computes their kernels and images.
The computational picture becomes much less transparent when the complex
is described indirectly by a much shorter object, which may implicitly
specify exponentially many simplices.  Such a representation is called
\emph{succinct}.  This distinction led Kaibel and Pfetsch to ask, more than
two decades ago, for the complexity of computing the homology of a
simplicial complex given by a succinct description
\cite{KaibelPfetsch}.

A particularly natural succinct representation is provided by
\emph{clique complexes}.  A graph \(G\) determines a simplicial complex
\(\Cl(G)\) by declaring every clique to be a simplex.  The graph itself has
a polynomial-size description, while \(\Cl(G)\) may contain exponentially
many simplices.  Clique complexes are also central objects in topological
data analysis, where one seeks global topological information from
combinatorial data.

A connection with quantum complexity emerged from a different direction.
The correspondence between supersymmetric zero-energy states and
cohomology goes back to Witten~\cite{Witten1982}.  Cade and Crichigno used
this connection in the study of the computational complexity of
(co)homological problems~\cite{CadeCrichigno}.  Write \(H_k(K;\CC)\) for
the \(k\)-dimensional homology of a simplicial complex \(K\), and
\(\Delta_k\) for the corresponding Hodge Laplacian.  For a finite complex,
\[
   H_k(K;\CC)\neq0
   \qquad\Longleftrightarrow\qquad
   0\in\spec\Delta_k .
\]
Thus the existence of a topological hole can be viewed as the existence of
a zero-energy state of a particular Hamiltonian.  This makes it possible
to bring tools from Hamiltonian complexity to problems that appear, on
their face, to be entirely classical.

Crichigno and Kohler subsequently proved that deciding whether a clique
complex has nontrivial homology is \(\mathsf{QMA}_1\)-hard
\cite{CrichignoKohler2024}.  Their result showed that even the apparently
classical task of detecting a hole in a graph-defined simplicial complex
can encode the full difficulty of quantum verification with perfect
completeness.

There is, however, an important distinction between merely deciding
whether zero lies in the spectrum and requiring a quantitative separation
from zero when it does not.  A NO instance of the ordinary homology problem
may have trivial homology while the smallest positive eigenvalue of its
Hodge Laplacian is arbitrarily small.  Such an instance is topologically
hole-free, but spectrally almost indistinguishable from one containing a
hole.

King and Kohler introduced the weighted \emph{gapped clique-homology}
problem considered here to impose precisely this quantitative distinction
\cite{KingKohler2026}.  In the YES case the prescribed homology group is
nonzero.  In the NO case it vanishes and the corresponding weighted Hodge
Laplacian is promised to have an inverse-polynomial gap above zero.  They
developed a hardness reduction based on a homological encoding of qubits,
weighted topological gadgets for local Hamiltonian constraints, and
spectral sequences for controlling the resulting low-energy Hodge
spectrum.

The present article revisits the hard direction of that reduction.  Our
relation to King and Kohler and the issue in the published argument are
discussed below.

\subsection{The basic idea of the reduction}

We first describe the mechanism at an intuitive level.  No knowledge of
homology or spectral sequences is needed for this discussion.

The basic qubit gadget is a small graph whose clique complex has two
independent one-dimensional holes, which represent the computational basis
states \(|0\rangle\) and \(|1\rangle\).  The topological join combines
these building blocks in the same way that a tensor product combines
independent quantum systems.  Taking joins of \(N\) copies therefore
produces \(2^N\) distinguished homology classes, one for each computational
basis state of
\[
   (\CC^2)^{\otimes N}.
\]
An exponentially large quantum state space is thereby encoded in the
homology of a graph of polynomial size.

A local Hamiltonian term penalizes one direction in this state space.  If
\[
   P=|\phi\rangle\langle\phi|
\]
is a rank-one projector, then \(P\) assigns energy to the component of a
state along \(|\phi\rangle\), while every state orthogonal to
\(|\phi\rangle\) lies in its kernel.  Topologically, the corresponding
gadget should fill exactly the encoded hole representing
\(|\phi\rangle\).  If it fills that class and no other encoded class, its
effect on the encoded homology is the analogue of the projector
\(P\).

For a Hamiltonian with many local terms, many such gadgets are attached to
the same encoded base.  The classes that survive all of the attachments
should then correspond precisely to states satisfying all of the source
constraints.

For the gapped problem this topological statement is not enough.  In a NO
instance one must also prove an inverse-polynomial spectral separation from
zero.  The private vertices of a gadget are therefore assigned a small
weight \(\lambda\).  A class that can be removed only through several
layers of private vertices gives rise to a positive mode at a correspondingly
high power of \(\lambda\).  Filtering simplices by the number of private
vertices organizes these successive orders, and the associated spectral
sequence records the first order at which an encoded class encounters an
obstruction.  In this sense the filtration plays a role analogous to
perturbation theory.

The local gadgets therefore carry two responsibilities at once: they must
have exactly the intended effect on encoded homology, and their low-energy
Hodge spectrum must reproduce the corresponding local Hamiltonian
penalty.  Both properties have to survive when polynomially many gadgets
are attached to the same base.

\subsection{Relation to King and Kohler and the gap in the hardness proof}
\label{subsec:introduction-failure}
\label{subsec:relation-king-kohler}

\paragraph{Relation to King and Kohler.}
The present proof builds very substantially on the framework developed by
King and Kohler~\cite{KingKohler2026}.  They introduced the weighted gapped
clique-homology problem and developed the central architecture of the
hardness reduction: the bow-tie encoding of qubits in homology, weighted
local gadgets for Hamiltonian constraints, the thickening, coning, and
folding constructions underlying those gadgets, the use of small vertex
weights together with a filtration and spectral sequence to control the
low-energy Hodge spectrum, and the local-to-global strategy for combining
many constraints.  Many of the concrete gadget constructions used here
also originate in their work and its companion files.  The connection they
developed between quantum constraints, topological fillings, and small
Hodge-Laplacian eigenvalues is an elegant and powerful aspect of their
construction, and it remains the conceptual backbone of the proof below.

\paragraph{The gap in the hard direction.}
The difficulty lies in the topology of the nontrivially folded local
gadgets.  The published proof uses preservation claims for the folding step
that do not hold in the required generality, leaving the topology of the
resulting clique complexes insufficiently controlled.  Since these local
topological conclusions are then used in the filtered and spectral
analysis, this creates a fundamental gap in the hard-direction proof.
This concerns the proof, not the truth of the hardness theorem, and King
and Kohler's independent $\mathsf{QMA}$-containment argument is unaffected.

The proof below addresses this issue by analyzing the actual clique
complexes produced by the reduction.  Exact signed chain identities,
relative homology, and a finite set of certified local conditions establish
the required topology, after which the spectral analysis is carried out on
those complexes using the same broad filtration and spectral-sequence
framework.  \Cref{sec:king-kohler-gap} gives the detailed diagnosis,
identifies the specific lemmas involved, and states precisely what must be
re-established.

\subsection{Main result}

The main result is the following.

\begin{theorem}[Main theorem]
\label{thm:introduction-main}
Fix the perfect-completeness gate set
\[
   \mathcal G=\{\mathrm{CNOT},U_{\mathrm{Pyth}}\}
\]
used below.  There exist functions
\[
   k:\ZZ_{\ge2}\longrightarrow\ZZ_{\ge1},
   \qquad
   \gamma:\ZZ_{\ge2}\longrightarrow\QQ_{>0},
\]
computable in deterministic polynomial time, with
\(\gamma(n)\ge1/\operatorname{poly}(n)\), such that weighted gapped clique
homology with degree \(k(n)\) and promise gap \(\gamma(n)\) is
\(\mathsf{QMA}_1(\mathcal G)\)-hard.
\end{theorem}

Combined with the independent \(\mathsf{QMA}\)-containment argument, this
recovers the intended complexity statement for the weighted gapped
problem.

\subsection{Outline of the proof}
\label{subsec:introduction-contributions}
\label{subsec:introduction-roadmap}

We briefly describe the proof, using notation that is introduced formally
later.

Starting from a \(\mathsf{QMA}_1(\mathcal G)\) verifier, we obtain a
positive semidefinite Hamiltonian
\[
   H_x=\sum_{j=1}^{t_x}P_j
\]
whose terms, up to ordered placement, belong to a fixed catalogue of
thirteen rank-one projectors.  In a YES instance \(H_x\) has nonzero
kernel, while in a NO instance
\[
   H_x\succeq g_xI
\]
for an inverse-polynomial \(g_x\).

The \(N\)-qubit space is encoded in the distinguished homology of an
\(N\)-fold join of bow-tie complexes:
\[
   H_{q_N}(X_N;\CC)\cong(\CC^2)^{\otimes N},
   \qquad
   q_N=2N-1.
\]
For each catalogue projector we attach a local graph to the corresponding
bow-tie factors and work with the full clique complex of the resulting
graph.  An explicit signed filling makes the prescribed encoded cycle a
boundary.  The relative homology of
\[
   (Y_a^{\mathrm{loc}},X_{m_a})
\]
then determines the full effect of the attachment on the encoded base.
The finite certificate package verifies the local combinatorial
properties needed for this argument and for the spectral analysis.

After placement and padding, the global topology gives
\[
   \Qenc_N\Big/\sum_j\operatorname{im}P_j
\]
as the surviving encoded quotient.  Since the \(P_j\) are orthogonal
projections,
\[
   \left(\sum_j\operatorname{im}P_j\right)^\perp
   =
   \bigcap_j\ker P_j
   =
   \ker H_x,
\]
which gives the exact YES-case correspondence.

For the spectral analysis, base vertices have weight \(1\), while private
gadget vertices have weight \(\lambda\).  In weighted orthonormal
coordinates the local coboundary splits as
\[
   D_\lambda
   =
   D_{\mathrm{base}}
   +
   \lambda D_{\mathrm{priv}}.
\]
Filtering by private depth produces a finite spectral sequence.  The
certified local data determine its relevant page structure, and the
filtered-Hodge theorem converts the surviving higher-page differential
into a rank-one positive eigenvalue in the projector direction while
separating the remaining local modes.

The global step is not obtained by simply summing the single-gadget
energies.  All gadget occurrences share the same encoded base, and their
down-boundary contributions can meet on common base simplices.  The
many-gadget estimate controls these shared-base terms and compares the
resulting low-energy form with the source Hamiltonian \(H_x\).  Choosing
\[
   \lambda
   =
   c_0\frac{g_x}{t_x}
\]
then gives an inverse-polynomial lower bound for the final Hodge
Laplacian in every NO instance.

Finally, isolated-vertex padding makes both the relevant homological
degree and the promise threshold functions only of the number of vertices
in the output graph.

The proof is organized accordingly.  After the preliminaries, source
Hamiltonian, and bow-tie encoding, \cref{sec:king-kohler-gap} gives the
detailed diagnosis of the earlier hardness argument.
Sections~\ref{sec:local-gadget} and~\ref{sec:global-graph} define the local
and global constructions, \cref{sec:certificates} records the finite
certificate conditions, and \cref{sec:topological-correctness} determines
their effect on homology.  The local and global spectral estimates are
proved in \cref{sec:single-gadget-spectrum,sec:global-spectrum}, and
\cref{sec:reduction} completes the reduction.  The appendices contain the
filtered-Hodge proof, the certificate protocols and reproducibility
information, and the remaining finite construction details.

\paragraph{Acknowledgments.}
The author is deeply grateful to Thomas Vidick, who first introduced him to
this problem and whose guidance, teaching, encouragement, and support were
invaluable.  Part of this work began during the author's visit to Vidick's
group at EPFL in 2025.  The author is also very grateful to Dorit Aharonov
for encouraging him to write this work up, and for her continued support,
guidance, teaching, and helpful comments on the manuscript.  He would also
like to thank Robbie King for generously discussing the original
construction with him, for being open about the issue, and for encouraging
him to try to resolve it.

\section{Promise problems, quantum verification, and hardness}
\label{sec:quantum-complexity}

We use the standard notions of promise problems, polynomial time, and
many-one reductions.  A promise problem is a pair
$L=(L_{\mathrm{YES}},L_{\mathrm{NO}})$ of disjoint subsets of
$\{0,1\}^*$.  For $N\ge0$, write
$\Qenc_N=(\CC^2)^{\otimes N}$, with its computational basis
$\{|b_1\cdots b_N\rangle:b_i\in\{0,1\}\}$ and the convention
$\Qenc_0=\CC$.

The fixed gate choice used below follows Crichigno and
Kohler~\cite[Section~3.1]{CrichignoKohler}.

\begin{definition}[Fixed quantum gate set]
\label{def:gate-set}
A fixed quantum gate set is a finite labelled collection
$\mathcal G=\{G_1,\ldots,G_s\}$ of unitary matrices of bounded arity,
independent of the input.  Throughout the article,
\[
  \mathcal G=\{\mathrm{CNOT},U_{\mathrm{Pyth}}\},
  \qquad
  U_{\mathrm{Pyth}}=\frac15
  \begin{pmatrix}3&4\\-4&3\end{pmatrix}.
\]
The gate $U_{\mathrm{Pyth}}$ acts on one qubit, while
$\mathrm{CNOT}|a,b\rangle=|a,a\oplus b\rangle$ on two qubits.  Both
matrices have rational entries.  Computational-basis initialization and
measurement are allowed separately from the gates in $\mathcal G$.
\end{definition}

\begin{definition}[Tensor-factor placement]
\label{def:tensor-placement}
Let $S=(s_1,\ldots,s_k)$ be an ordered tuple of distinct elements of
$\{1,\ldots,N\}$, and write its complement in increasing order as
$S^c=(r_1,\ldots,r_{N-k})$.  The placement unitary
$W_S:\Qenc_k\otimes\Qenc_{N-k}\to\Qenc_N$ is determined on computational
basis vectors by placing the first tensor factor on the qubits in $S$, in
the displayed order, and the second on the qubits in $S^c$.  Explicitly,
\[
 W_S\bigl(|b_1\cdots b_k\rangle\otimes|c_1\cdots c_{N-k}\rangle\bigr)
 =|x_1\cdots x_N\rangle,
\]
where $x_{s_i}=b_i$ and $x_{r_j}=c_j$.  For an operator $A$ on $\Qenc_k$,
its placement on $S$ is
\[
  A_S^{(N)}:=W_S\bigl(A\otimes I_{\Qenc_{N-k}}\bigr)W_S^*.
\]
\end{definition}

A quantum circuit over $\mathcal G$ on $N$ qubits is an ordered list
$C=((g_1,S_1),\ldots,(g_L,S_L))$ of gate instructions, with each $S_j$
matching the arity of $g_j$.  In the placement notation above it implements
\[
  U_C=(g_L)_{S_L}^{(N)}\cdots(g_1)_{S_1}^{(N)}.
\]
Its size is $L$.  Since $\mathcal G$ is fixed and has bounded arity, the
usual binary encoding of the gate list has length polynomial in $L$ and
$\log N$.  A circuit family $\{V_x\}$ is polynomial-time uniform if a
deterministic polynomial-time algorithm outputs this description from $x$,
and has polynomial size if its numbers of qubits and gates are bounded by a
polynomial in $|x|$.

\begin{definition}[$\mathsf{QMA}_1(\mathcal G)$]
\label{def:qma1}
A promise problem $L=(L_{\mathrm{YES}},L_{\mathrm{NO}})$ belongs to
$\mathsf{QMA}_1(\mathcal G)$ if there is a polynomial-time uniform family
of polynomial-size circuits $\{V_x\}$ over $\mathcal G$ with the following
properties.  The circuit $V_x$ acts on a witness register, supplied in an
arbitrary unit state $|\psi\rangle$, and an auxiliary register initialized
in a computational-basis state $|b_x\rangle$ computable from $x$ in
polynomial time.  One output qubit is measured in the computational basis,
with outcome $0$ interpreted as acceptance, and
\begin{enumerate}[label=(\roman*)]
\item if $x\in L_{\mathrm{YES}}$, some unit witness is accepted with
probability $1$;
\item if $x\in L_{\mathrm{NO}}$, every unit witness is accepted with
probability at most $1/3$.
\end{enumerate}
Allowing mixed witnesses is equivalent, since the acceptance probability is
affine in the witness state.
\end{definition}

A polynomial-time many-one reduction from a promise problem $L$ to a
promise problem $P$ is a deterministic polynomial-time computable map $f$
with $f(L_{\mathrm{YES}})\subseteq P_{\mathrm{YES}}$ and
$f(L_{\mathrm{NO}})\subseteq P_{\mathrm{NO}}$.  A problem is
$\mathsf{QMA}_1(\mathcal G)$-hard if every problem in
$\mathsf{QMA}_1(\mathcal G)$ admits such a reduction to it.  No condition
is imposed outside the respective promises.

\section{Clique complexes, simplicial homology, and Hodge theory}
\label{sec:clique-hodge}

We fix the simplicial and Hodge-theoretic conventions used below.  The
underlying topological material is standard; see, for example,
\cite[Chapter~2]{Hatcher}.  The weighted conventions are stated explicitly
because the reduction depends on the corresponding matrix coefficients.

\subsection{Simplicial complexes and clique complexes}
\label{subsec:simplicial-clique}

\begin{definition}[Finite abstract simplicial complex]
\label{def:simplicial-complex}
A finite abstract simplicial complex $K$ on a finite vertex set $V$ is a
collection of subsets of $V$ closed under taking subsets.  Its $r$-simplices
are \(K_r:=\{\sigma\in K:|\sigma|=r+1\}\), and for $W\subseteq V$ the induced subcomplex is
$K[W]:=\{\sigma\in K:\sigma\subseteq W\}$.
\end{definition}

\begin{definition}[Geometric realization]
\label{def:geometric-realization}
For $v\in V(K)$ let $e_v$ be the corresponding standard basis vector of
$\RR^{V(K)}$.  The geometric realization of $K$ is
\[
  |K|:=\bigcup_{\sigma\in K}\operatorname{conv}\{e_v:v\in\sigma\},
\]
where $\operatorname{conv}$ denotes the convex hull.
\end{definition}

\begin{center}
\begin{minipage}[t]{0.47\textwidth}
\centering
\begin{tikzpicture}[
  vertex/.style={circle,fill=black,inner sep=1.7pt},
  edge/.style={draw=black,line width=1.05pt},
  every node/.style={font=\small}
]
\coordinate (v0) at (-1.7,0);
\coordinate (v1) at ( 1.7,0);
\coordinate (v2) at ( 0,2.8);

\draw[edge] (v0)--(v1)--(v2)--cycle;

\foreach \p in {v0,v1,v2}
  \node[vertex] at (\p) {};

\node[below left]  at (v0) {$v_0$};
\node[below right] at (v1) {$v_1$};
\node[above]       at (v2) {$v_2$};

\node at (0,-0.65) {$|K|\cong S^1$};
\end{tikzpicture}

\refstepcounter{figure}
\label{fig:triangle-boundary-realization}
\smallskip

{\small
\textbf{Figure~\thefigure.}
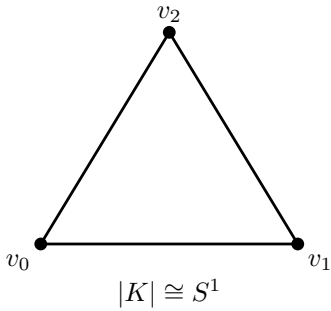
The complex \(K\) contains the three vertices and the three edges joining
them, but not the two-simplex \(\{v_0,v_1,v_2\}\).  Its geometric
realization therefore consists only of the boundary of the triangle and is
homeomorphic to \(S^1\). }
\end{minipage}\hfill
\begin{minipage}[t]{0.47\textwidth}
\centering
\begin{tikzpicture}[
  vertex/.style={circle,fill=black,inner sep=1.7pt},
  edge/.style={draw=black,line width=1.05pt},
  every node/.style={font=\small}
]
\coordinate (v0) at (-1.7,0);
\coordinate (v1) at ( 1.7,0);
\coordinate (v2) at ( 0,2.8);

\fill[blue!13] (v0)--(v1)--(v2)--cycle;
\draw[edge] (v0)--(v1)--(v2)--cycle;

\foreach \p in {v0,v1,v2}
  \node[vertex] at (\p) {};

\node[below left]  at (v0) {$v_0$};
\node[below right] at (v1) {$v_1$};
\node[above]       at (v2) {$v_2$};

\node at (0,1.1) {$\{v_0,v_1,v_2\}$};
\node at (0,-0.65) {$|K'|\cong D^2$};
\end{tikzpicture}

\refstepcounter{figure}
\label{fig:filled-triangle-realization}
\smallskip

{\small
\textbf{Figure~\thefigure.}
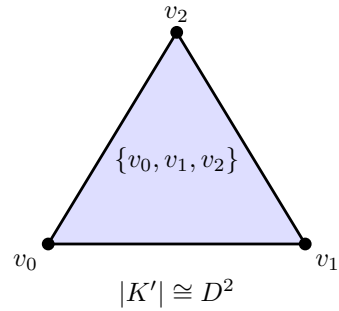
The complex \(K'\) also contains the two-simplex \(\{v_0,v_1,v_2\}\).  Its geometric realization therefore contains the
interior of the triangle and is homeomorphic to the closed disk \(D^2\). }
\end{minipage}
\end{center}

\begin{definition}[One-skeleton]
\label{def:one-skeleton}
The one-skeleton $\Sk_1(K)$ is the graph with vertex set $V(K)$ and edge
set $K_1$.
\end{definition}

\begin{definition}[Clique complex]
\label{def:clique-complex}
For a finite simple graph $G=(V,E)$, its clique complex (or flag complex) is
\[
  \Cl(G):=\{\sigma\subseteq V:\text{ every two distinct vertices of
  $\sigma$ are adjacent in $G$}\}.
\]
Thus the simplices of $\Cl(G)$ are the cliques of $G$, and
$\Sk_1(\Cl(G))=G$.
\end{definition}

\subsection{Chains, boundaries, and homology}
\label{subsec:chains-homology}

Choose an orientation of every simplex.  We write
$[v_0,\ldots,v_r]$ for the orientation represented by the displayed order;
a permutation changes this vector by its sign.  The chain space
$C_r(K;\CC)$ is the complex vector space with these oriented $r$-simplices
as an orthonormal basis.  Thus a chain is
$c=\sum_{\sigma\in K_r}c_\sigma[\sigma]$ and
$\langle c,d\rangle=\sum_\sigma\overline{c_\sigma}d_\sigma$.

\begin{definition}[Boundary operator]
\label{def:boundary-operator}
For $r\ge1$, the boundary operator
$\partial_r:C_r(K;\CC)\to C_{r-1}(K;\CC)$ is defined by
\[
  \partial_r[v_0,\ldots,v_r]
  :=\sum_{i=0}^r(-1)^i[v_0,\ldots,\widehat v_i,\ldots,v_r],
\]
and $\partial_0:=0$.  The alternating signs give
$\partial_{r-1}\partial_r=0$.
\end{definition}

Write
\[
  Z_r(K;\CC):=\ker\partial_r,
  \qquad
  B_r(K;\CC):=\operatorname{im}\partial_{r+1}.
\]
Elements of $Z_r$ are cycles, elements of $B_r$ are boundaries, and
$B_r\subseteq Z_r$.  The $r$-th simplicial homology is
\[
  H_r(K;\CC):=Z_r(K;\CC)/B_r(K;\CC),
\]
with the class of a cycle $z$ denoted by $[z]$.  Unless stated otherwise,
homology in this article has coefficients in $\CC$.

\subsection{Cochains and cohomology}
\label{subsec:cochains-cohomology}

The cochain space is the dual
$C^r(K;\CC):=\operatorname{Hom}_{\CC}(C_r(K;\CC),\CC)$.  We denote by
$[\sigma]^*$ the basis dual to the oriented simplex basis.  The coboundary
$D_r:C^r(K;\CC)\to C^{r+1}(K;\CC)$ is dual to the boundary:
\[
  (D_r\alpha)(c):=\alpha(\partial_{r+1}c).
\]
Hence $D_{r+1}D_r=0$.  We write
\[
  Z^r(K;\CC):=\ker D_r,
  \qquad
  B^r(K;\CC):=\operatorname{im}D_{r-1},
  \qquad
  H^r(K;\CC):=Z^r(K;\CC)/B^r(K;\CC).
\]
Unless stated otherwise, cohomology also has coefficients in $\CC$.

\subsection{Hodge Laplacians and harmonic cochains}
\label{subsec:hodge-theory}

Equip $C^r(K;\CC)$ with the inner product for which the dual simplex basis
is orthonormal, and let $D_r^*$ be the Hilbert-space adjoint of $D_r$.
Under the unweighted simplex-basis identification of chains and cochains,
$D_r^*=\partial_{r+1}$.

The down-, up-, and full degree-$r$ Hodge Laplacians are
\[
  \Delta_r^{\downarrow}:=D_{r-1}D_{r-1}^*,
  \qquad
  \Delta_r^{\uparrow}:=D_r^*D_r,
  \qquad
  \Delta_r:=\Delta_r^{\downarrow}+\Delta_r^{\uparrow},
\]
with the missing map in degree zero interpreted as zero.

\begin{lemma}[Hodge energy identity]
\label{lem:hodge-energy}
For every $u\in C^r(K;\CC)$,
\[
  \langle u,\Delta_ru\rangle
  =\|D_{r-1}^*u\|^2+\|D_ru\|^2.
\]
Consequently $\Delta_r$ is self-adjoint and positive semidefinite, and
\[
  \ker\Delta_r=\ker D_r\cap\ker D_{r-1}^*.
\]
\end{lemma}

\begin{proof}
Apply the defining property of the adjoint to the two summands of
$\Delta_r$.  The displayed quadratic form is a sum of two squared norms,
which vanishes exactly when both terms vanish.
\end{proof}

A cochain in $\ker\Delta_r$ is called harmonic; write
$\mathscr H^r(K):=\ker\Delta_r$.

The following finite-dimensional Hodge theorem is due to
Eckmann~\cite{Eckmann1944}.

\begin{theorem}[Combinatorial Hodge theorem]
\label{thm:combinatorial-hodge}
For every finite simplicial complex $K$ and $r\ge0$,
\[
  C^r(K;\CC)
  =\operatorname{im}D_{r-1}\oplus\mathscr H^r(K)\oplus\operatorname{im}D_r^*
\]
orthogonally, and the map $u\mapsto[u]$ induces an isomorphism
$\mathscr H^r(K)\cong H^r(K;\CC)$.
\end{theorem}

\begin{proof}
Since $D_rD_{r-1}=0$, one has
$\operatorname{im}D_{r-1}\subseteq\ker D_r$.  Finite-dimensional
orthogonal decomposition gives
\[
 C^r=\ker D_r\oplus\operatorname{im}D_r^*,
 \qquad
 \ker D_r=\operatorname{im}D_{r-1}\oplus
 (\ker D_r\cap\ker D_{r-1}^*).
\]
The last intersection is $\mathscr H^r(K)$ by
\cref{lem:hodge-energy}, and quotienting $\ker D_r$ by
$\operatorname{im}D_{r-1}$ gives the cohomology statement.
\end{proof}

\begin{corollary}[Homology and zero eigenvalues]
\label{cor:homology-zero-eigenvalues}
For every finite simplicial complex $K$ and $r\ge0$,
\[
  \dim\ker\Delta_r
  =\dim H^r(K;\CC)
  =\dim H_r(K;\CC).
\]
In particular, $H_r(K;\CC)\ne0$ if and only if $0\in\spec\Delta_r$, with
multiplicity $\dim H_r(K;\CC)$.
\end{corollary}

\begin{proof}
The first equality is \cref{thm:combinatorial-hodge}; the second follows
from the universal coefficient theorem over the field $\CC$.
\end{proof}

\subsection{Vertex weights and weighted Hodge Laplacians}
\label{subsec:weighted-hodge}

Let $G=(V,E)$ be a finite simple graph and $K=\Cl(G)$.  A positive vertex
weighting is a function $w:V\to\RR_{>0}$, extended multiplicatively to
simplices by
\[
  w(\sigma):=\prod_{v\in\sigma}w(v),
  \qquad w(\varnothing):=1.
\]
On $C^r(K;\CC)$ we use the weighted inner product for which distinct
oriented simplices are orthogonal and
$\langle\sigma,\sigma\rangle_w=w(\sigma)^2$.  Thus \(|\widehat\sigma\rangle:=w(\sigma)^{-1}|\sigma\rangle\) is the weighted orthonormal simplex basis.

The algebraic coboundary $D_r$ is independent of the weights.  Its matrix
in the weighted orthonormal basis is not.

\begin{proposition}[Coboundary in weighted orthonormal coordinates]
\label{prop:weighted-coboundary-matrix}
For an oriented $r$-simplex $\sigma$,
\[
  D_r|\widehat\sigma\rangle
  =\sum_{\substack{v\notin\sigma\\ \sigma\cup\{v\}\in K}}
  \epsilon(\sigma,\sigma\cup\{v\})\,w(v)\,
  |\widehat{\sigma\cup\{v\}}\rangle,
\]
where $\epsilon(\sigma,\tau)\in\{\pm1\}$ is the incidence sign of
$\sigma$ in $\partial\tau$.
\end{proposition}

\begin{proof}
In the ordinary simplex basis,
$D_r|\sigma\rangle=\sum_{\tau\supset\sigma}\epsilon(\sigma,\tau)|\tau\rangle$.
Since $|\widehat\sigma\rangle=w(\sigma)^{-1}|\sigma\rangle$ and every
coface $\tau=\sigma\cup\{v\}$ satisfies
$w(\tau)/w(\sigma)=w(v)$, the displayed formula follows.
\end{proof}

Equivalently, if $W_r=\operatorname{diag}(w(\sigma):\sigma\in K_r)$, then
\[
  [D_r]_w=W_{r+1}[D_r]_{\mathrm{simp}}W_r^{-1}.
\]
Let $D_r^{*,w}$ denote the adjoint with respect to the weighted inner
products.  The weighted Hodge Laplacian is
\[
  \Delta_{r,w}:=D_{r-1}D_{r-1}^{*,w}+D_r^{*,w}D_r,
\]
with down- and up-parts given by the two displayed summands.  In the
weighted orthonormal basis,
\[
  D_r^{*,w}|\widehat\tau\rangle
  =\sum_{v\in\tau}
  \epsilon(\tau\setminus\{v\},\tau)\,w(v)\,
  |\widehat{\tau\setminus\{v\}}\rangle.
\]

\begin{lemma}[Weighted Hodge energy identity]
\label{lem:weighted-hodge-energy}
For every $u\in C^r(K;\CC)$,
\[
  \langle u,\Delta_{r,w}u\rangle_w
  =\|D_{r-1}^{*,w}u\|_w^2+\|D_ru\|_w^2.
\]
Consequently $\Delta_{r,w}$ is self-adjoint and positive semidefinite, and
\[
  \ker\Delta_{r,w}=\ker D_r\cap\ker D_{r-1}^{*,w}.
\]
\end{lemma}

\begin{proof}
This is the same adjoint calculation as in \cref{lem:hodge-energy}, now
with the weighted inner products.
\end{proof}

\begin{proposition}[Weights preserve cohomology and zero modes]
\label{prop:weights-preserve-cohomology}
The map $u\mapsto[u]$ induces an isomorphism
\[
  \ker\Delta_{r,w}\cong H^r(K;\CC).
\]
Consequently
\[
  \dim\ker\Delta_{r,w}=\dim H^r(K;\CC)=\dim H_r(K;\CC),
\]
so positive vertex weights do not change the multiplicity of the zero
eigenvalue.
\end{proposition}

\begin{proof}
The maps $D_r$ do not depend on the weights.  With the weighted inner
product,
\[
  \ker D_r
  =\operatorname{im}D_{r-1}\oplus
   (\ker D_r\cap\ker D_{r-1}^{*,w}).
\]
The second summand is $\ker\Delta_{r,w}$ by
\cref{lem:weighted-hodge-energy}, so quotienting by
$\operatorname{im}D_{r-1}$ gives $H^r(K;\CC)$.  The equality of homology
and cohomology dimensions again follows from the universal coefficient
theorem over $\CC$.
\end{proof}

\section{The target problem and main result}
\label{sec:target-problem}

\subsection{Weighted gapped clique homology}
\label{subsec:target-definition}

Fix functions
\[
  k:\ZZ_{\geq2}\longrightarrow\ZZ_{\geq1},
  \qquad
  \gamma:\ZZ_{\geq2}\longrightarrow\QQ_{>0}
\]
that are exactly computable in polynomial time and satisfy, for some
polynomial $p$,
\[
  1\leq k(n)<n,
  \qquad
  \frac1{p(n)}\leq\gamma(n)\leq1.
\]
We also require the binary representation of $\gamma(n)$ to have length
polynomial in $n$.

\begin{definition}[Weighted gapped clique homology]
\label{def:weighted-gapped-clique-homology}
An input is the standard binary encoding $\langle G,w\rangle$ of a finite
simple graph $G$ on $n\geq2$ labelled vertices and rational vertex weights
$0<w(v)\leq1$.  The encoding records the graph and the numerator and
denominator of each weight in binary; the input size is
$|\langle G,w\rangle|$.

Put $K:=\Cl(G)$ and let $\Delta_{k(n),w}$ be the degree-$k(n)$ weighted
Hodge Laplacian.  For a positive semidefinite operator $A$ on a nonzero
finite-dimensional space, write $\lambda_{\min}(A)$ for its smallest
eigenvalue; on the zero-dimensional space set $\lambda_{\min}=+\infty$.
The input is promised to satisfy exactly one of
\begin{enumerate}[label=(\roman*)]
\item \emph{YES:} $H_{k(n)}(K;\CC)\neq0$;
\item \emph{NO:} $H_{k(n)}(K;\CC)=0$ and
\[
  \lambda_{\min}(\Delta_{k(n),w})\geq\gamma(n).
\]
\end{enumerate}
The problem is to decide which case holds.
\end{definition}

By \cref{cor:homology-zero-eigenvalues,prop:weights-preserve-cohomology},
$H_{k(n)}(K;\CC)\neq0$ exactly when $\Delta_{k(n),w}$ has a zero mode.  The
NO promise therefore imposes an inverse-polynomial separation from zero
when the relevant homology group vanishes.

\subsection{Main result}
\label{subsec:main-result}

The reduction below uses one fixed pair $(k,\gamma)$ for every source
problem and every input; only the output graph varies.

\begin{theorem}[Main hardness theorem]
\label{thm:main-hardness}
There exist fixed functions $k$ and $\gamma$ satisfying
\cref{def:weighted-gapped-clique-homology} such that weighted gapped clique
homology is $\mathsf{QMA}_1(\mathcal G)$-hard.  More explicitly, for every \(L=(L_{\mathrm{YES}},L_{\mathrm{NO}}) \in\mathsf{QMA}_1(\mathcal G)\) there is a deterministic polynomial-time reduction sending an input $x$ to
a finite simple graph $G_x$ with rational weights
$0<w_x(v)\leq1$ such that
\[
  x\in L_{\mathrm{YES}}
  \quad\Longrightarrow\quad
  H_{k(|V(G_x)|)}(\Cl(G_x);\CC)\neq0,
\]
whereas
\[
  x\in L_{\mathrm{NO}}
  \quad\Longrightarrow\quad
  H_{k(|V(G_x)|)}(\Cl(G_x);\CC)=0
\]
and
\[
  \lambda_{\min}\!\left(\Delta_{k(|V(G_x)|),w_x}\right)
  \geq\gamma(|V(G_x)|).
\]
\end{theorem}

\subsection{\texorpdfstring{$\mathsf{QMA}$}{QMA} containment}
\label{subsec:qma-containment}

The containment result is not new.  The sparse-Hamiltonian and
phase-estimation method for gapped homological problems was developed by
Cade and Crichigno~\cite{CadeCrichigno}, applied to clique homology by
Crichigno and Kohler~\cite{CrichignoKohler2024}, and proved for the weighted
gapped problem by King and Kohler~\cite[Theorem~11.2]{KingKohler2026}.  We
record the argument in the present notation.

\begin{proposition}[$\mathsf{QMA}$ containment]
\label{prop:qma-containment}
For any fixed $k$ and $\gamma$ satisfying
\cref{def:weighted-gapped-clique-homology}, weighted gapped clique homology
belongs to $\mathsf{QMA}$.
\end{proposition}

\begin{proof}
Let $(G,w)$ be a promised input on $n$ vertices, set $K=\Cl(G)$ and
$k=k(n)$, and let $\Delta_{k,w}$ be the weighted Hodge Laplacian.  In the
weighted orthonormal simplex basis, embed
\[
  C^k(K;\CC)\hookrightarrow(\CC^2)^{\otimes n},
  \qquad
  |\widehat\sigma\rangle\longmapsto|z_\sigma\rangle,
\]
where $z_\sigma$ is the indicator string of the $k$-simplex $\sigma$.
Extend $\Delta_{k,w}$ by the identity on the orthogonal complement of this
simplex space, obtaining $\widetilde\Delta_{k,w}$.

The operator $\widetilde\Delta_{k,w}$ has efficient polynomially sparse
access: clique membership and every nonzero incidence matrix element are
computable in polynomial time from $(G,w)$.  Since $0<w(v)\leq1$, each
weighted coboundary row and column has at most $n$ nonzero entries of
absolute value at most $1$, and hence
\[
  \|\Delta_{k,w}\|\leq2n^2,
  \qquad
  \|\widetilde\Delta_{k,w}\|\leq B_n:=2n^2+1.
\]
Thus $H_G:=\widetilde\Delta_{k,w}/B_n$ is an efficiently simulable sparse
Hamiltonian of norm at most $1$~\cite{BerryChildsKothari2015}.  The promised
separation $\delta_n:=\gamma(n)/B_n$ is inverse-polynomial, so phase
estimation to error below $\delta_n/4$, accepting below $\delta_n/2$,
separates the two cases.

In a YES instance, Hodge theory gives a zero-energy eigenstate that the
prover can supply.  In a NO instance, $\Delta_{k,w}\succeq\gamma(n)I$ on
the simplex space, while the orthogonal complement has eigenvalue $1$;
therefore $H_G\succeq\delta_n I$.  Standard phase-estimation error
reduction gives the required soundness.
\end{proof}

The containment proof applies to an arbitrary promised weighted graph and
uses no part of the hardness construction.

\section{From quantum verifiers to a finite catalogue of local projections}
\label{sec:source-side}

The source Hamiltonians have the form \(H=\sum_{j=1}^{t}P_j\), where each term $P_j$ is obtained by placing, in the sense of
\cref{def:tensor-placement}, one member of a fixed catalogue of thirteen
rank-one projections on an ordered tuple of qubits.

\subsection{The finite catalogue}
\label{subsec:projector-catalogue}

The following catalogue is the rank-one projector list of
\cite[Table~1]{CrichignoKohler}, with a fixed ordering and the
normalizations specified below.

\begin{definition}[Canonical catalogue]
\label{def:catalogue}
Let \(v^{(1)},\ldots,v^{(13)}\) be the following thirteen nonzero vectors.

The five four-qubit vectors are
\[
\begin{gathered}
  |1011\rangle-|1000\rangle,
  \qquad
  |0110\rangle-|0101\rangle,
  \qquad
  |0010\rangle-|0001\rangle,\\
  |1101\rangle-|1010\rangle,
  \qquad
  |1011\rangle-|1100\rangle.
\end{gathered}
\]

The two three-qubit vectors are
\[
  -5|011\rangle+4|100\rangle+3|101\rangle
\]
and
\[
  -5|010\rangle+3|100\rangle-4|101\rangle.
\]

The remaining six computational-basis vectors are
\[
  |00\rangle,
  \qquad
  |11\rangle,
  \qquad
  |011\rangle,
  \qquad
  |1100\rangle,
  \qquad
  |0111\rangle,
  \qquad
  |0001\rangle.
\]

For each $a\in\{1,\ldots,13\}$, let $m_a\in\{2,3,4\}$ be the number of
qubits of $v^{(a)}$, and define
\[
  \mathcal Z_a:=\sqrt{\langle v^{(a)}|v^{(a)}\rangle},
  \qquad
  |\phi^{(a)}\rangle:=\mathcal Z_a^{-1}|v^{(a)}\rangle,
  \qquad
  \Pi_a:=|\phi^{(a)}\rangle\langle\phi^{(a)}|
  \quad\text{on }\Qenc_{m_a}.
\]
Then $\Pi_a$ is the orthogonal projection onto
$\CC\phi^{(a)}$.
\end{definition}

All coordinates of the vectors $v^{(a)}$ are integers.  Consequently,
all matrix entries of the projections $\Pi_a$ are rational and can be
represented exactly.

\paragraph{Placing a catalogue vector.}
Let $S=(s_1,\ldots,s_m)$ be an ordered tuple of distinct qubits.  By
\cref{def:tensor-placement}, the unitary $W_S$ sends
$|c_1\cdots c_m\rangle\otimes|e\rangle$ to the computational basis vector
whose bit at qubit $s_i$ is $c_i$ for every $i$, with the bits of $e$ on
the remaining qubits.  Hence $(\Pi_a)_S^{(N)}$ acts on the qubits
$s_1,\ldots,s_m$ as the projection onto the vector obtained from
$\phi^{(a)}$ by writing, in each of its basis strings, the $i$-th bit on
qubit $s_i$, and as the identity on every other qubit.

We use this in the following form.  Let $R$ be an ordered tuple of $m$
qubits and let $v$ be a nonzero vector on those qubits, written in the
order of $R$.  If, for some reordering $S$ of $R$, writing the bits of
$\phi^{(a)}$ onto $S$ produces a nonzero multiple of $v$, then the
projection onto $\CC v$, placed on $R$, equals $(\Pi_a)_S^{(N)}$.  In particular, a computational-basis projection $|e\rangle\langle e|$
is a placement of a catalogue basis projection whenever the catalogue
contains a basis string of the same length and Hamming weight, since two
such strings differ only in the positions of their ones.  All basis
projections required by the source construction satisfy this condition, as
verified in \cref{app:catalogue-closure}.

\subsection{Finite-catalogue source presentations}
\label{subsec:source-presentations}

\begin{definition}[Finite-catalogue source presentation]
\label{def:source-family}
Let $N,t\geq1$.  For every $j\in\{1,\ldots,t\}$, choose a catalogue label \(a_j\in\{1,\ldots,13\}\) and an ordered tuple \(S_j=(s_{j,1},\ldots,s_{j,m_{a_j}})\) of distinct elements of $\{1,\ldots,N\}$.

The corresponding \emph{finite-catalogue source Hamiltonian} is
\[
  H
  :=
  \sum_{j=1}^{t}P_j,
  \qquad
  P_j
  :=
  \bigl(\Pi_{a_j}\bigr)_{S_j}^{(N)}.
\]
The data \(\bigl(N,t,(a_j,S_j)_{j=1}^{t}\bigr)\) are called a \emph{finite-catalogue source presentation} of $H$.

The integer $t$ is the length of the displayed list, counted with
multiplicity.  In particular, repeated occurrences of the same placed
projection are counted separately.
\end{definition}

Every $P_j$ is an orthogonal projection, and hence \(H\succeq0\). Moreover, for every $\psi\in\Qenc_N$,
\[
  \langle\psi,H\psi\rangle
  =
  \sum_{j=1}^{t}\|P_j\psi\|^2.
\]
It follows that
\[
  \ker H
  =
  \bigcap_{j=1}^{t}\ker P_j.
\]

\subsection{The source theorem}
\label{subsec:finite-catalogue-source-theorem}

We state Bravyi's circuit-to-$4$-SAT construction in the form used here,
record its decomposition into the catalogue, and then give the modification
needed to accommodate the prescribed auxiliary bits allowed by
\cref{def:qma1}, without enlarging the catalogue.

\paragraph{Clock notation.}
A \emph{clock particle} is a pair of qubits.  Its four computational basis
states are named
\[
  |u\rangle:=|00\rangle,\qquad
  |a_1\rangle:=|01\rangle,\qquad
  |a_2\rangle:=|10\rangle,\qquad
  |d\rangle:=|11\rangle,
\]
read as unborn, active in phase one, active in phase two, and dead.  A
subscript $j$ on an operator indicates that it acts on the $j$-th clock
particle; a subscript $i$ on a one-qubit operator indicates the data qubit
$i$.

\begin{theorem}[
  Circuit-to-$4$-SAT construction of Bravyi
  {\cite[Section~III]{Bravyi}};
  explicit terms as in {\cite[Appendix~A]{CrichignoKohler}}
]
\label{thm:cited-source-reduction}
Let $C=\bigl((g_1,T_1),\ldots,(g_L,T_L)\bigr)$ be a circuit with $L\geq1$ instructions on a data
register of $N_d$ qubits, in which each $g_t$ is an arbitrary unitary
acting on the $|T_t|\in\{1,2\}$ qubits listed in $T_t$.  The data qubits
are partitioned into a witness part $\mathsf W$ and an auxiliary part
$\mathsf A$; the auxiliary qubits are initialized in $|0\rangle$; one data
qubit $o$ is the output qubit, and outcome $0$ means acceptance.

Adjoin $L$ clock particles.  There is a Hamiltonian
\[
  H(C)
  =
  H_{\mathrm{clock}}
  +H_{\mathrm{in}}
  +H_{\mathrm{out}}
  +\sum_{t=1}^{L}H_{\mathrm{prop},t}
  +\sum_{t=1}^{L-1}H'_{\mathrm{prop},t},
\]
each summand a sum of orthogonal projections acting on at most four
qubits, of the following form:
\begin{enumerate}[label=(\alph*)]
\item $H_{\mathrm{clock}}$ is a sum of computational-basis projections,
each acting on one clock particle or on two clock particles; these are
listed explicitly in \cref{app:catalogue-closure};
\item $H_{\mathrm{in}}=\sum_{i\in\mathsf A}
|a_1\rangle\langle a_1|_1\otimes|1\rangle\langle1|_i$;
\item $H_{\mathrm{out}}=|a_2\rangle\langle a_2|_L\otimes|1\rangle\langle1|_o$;
\item for each $t$, with $g_t$ acting on the qubits $T_t$,
\[
  H_{\mathrm{prop},t}
  =\tfrac12\Bigl[
     \bigl(|a_1\rangle\langle a_1|+|a_2\rangle\langle a_2|\bigr)_t\otimes I
     -|a_2\rangle\langle a_1|_t\otimes g_t
     -|a_1\rangle\langle a_2|_t\otimes g_t^*
   \Bigr];
\]
\item $H'_{\mathrm{prop},t}$ is the projection onto
$\CC\bigl(|a_2u\rangle-|da_1\rangle\bigr)$ on clock particles $t,t+1$.
\end{enumerate}
Moreover, there is a fixed polynomial $q\ge1$ with nonnegative
coefficients, depending only on the construction and the soundness
constant $1/3$, such that:
\begin{enumerate}[label=(\roman*)]
\item if some witness $\psi$ satisfies
$\Pr[C\text{ accepts }\psi\otimes|0\rangle_{\mathsf A}]=1$, then
$\ker H(C)\neq\{0\}$;
\item if every witness $\psi$ satisfies
$\Pr[C\text{ accepts }\psi\otimes|0\rangle_{\mathsf A}]\le\frac13$, then
$H(C)\succeq \frac{1}{q(L+N_d)}I$.
\end{enumerate}
\end{theorem}

The polynomial $q$ is universal for this construction: it is independent
of the circuit $C$, its length, its gates, and the source promise problem.

\begin{lemma}[Catalogue closure]
\label{lem:catalogue-closure}
Let $C$ be a circuit whose gates lie in the gate set $\mathcal G$ of
\cref{def:gate-set}.  Each summand of $H(C)$, and each operator
\[
  |a_1\rangle\langle a_1|_1\otimes|0\rangle\langle0|_i\otimes I_s
\]
for a data qubit $i$ and an additional qubit $s$, is a sum of at most four
pairwise orthogonal rank-one projections, each equal to
$(\Pi_a)_S^{(N)}$ for some catalogue label $a$ and some reordering $S$ of
the qubits on which the summand acts.
\end{lemma}

\begin{proof}
For the summands of $H(C)$ this is the decomposition tabulated in
\cite[Table~1]{CrichignoKohler}.  A complete proof, including the
additional operator, is given in \cref{app:catalogue-closure}.
\end{proof}

\begin{proposition}[Reduction to the thirteen-projector catalogue]
\label{prop:source-catalogue-identification}
Let $V$ be a circuit over the gate set $\mathcal G$ of
\cref{def:gate-set} with $L\ge0$ gate instructions, acting on a data
register of $N_d$ qubits, partitioned into a witness part $\mathsf W$ and
an auxiliary part $\mathsf A$.  Let $b\in\{0,1\}^{\mathsf A}$ assign a
prescribed bit $b_i$ to each auxiliary qubit $i\in\mathsf A$, and let $o$
be the output qubit, with outcome $0$ meaning acceptance, as in
\cref{def:qma1}.

From the description of $(V,b,o)$ one can compute, in time polynomial in
$L+N_d$, a finite-catalogue source presentation of a Hamiltonian $H^+$ on \(N^+:=2L+N_d+4\) qubits, with \(t^+\le16\bigl((L+1)^2+N_d\bigr)\) terms, such that, with $q$ the polynomial of
\cref{thm:cited-source-reduction}:
\begin{enumerate}[label=(\roman*)]
\item if some witness $\psi$ satisfies
$\Pr[V\text{ accepts }\psi\otimes|b\rangle]=1$, then
$\ker H^+\neq\{0\}$;
\item if every witness $\psi$ satisfies
$\Pr[V\text{ accepts }\psi\otimes|b\rangle]\le\frac13$, then
$H^+\succeq \frac{1}{q(L+N_d+2)}I$.
\end{enumerate}
\end{proposition}

\begin{proof}
\emph{Step 1: moving the output.}  Adjoin a fresh auxiliary qubit $o'$
with prescribed bit $0$, append the instruction
$(\mathrm{CNOT},(o,o'))$, and declare $o'$ the output qubit.  Call the
resulting circuit $\widehat V$.  It has $L+1\ge1$ instructions, all in
$\mathcal G$, and $N_d+1$ data qubits.  Since
$\mathrm{CNOT}(\varphi\otimes|0\rangle_{o'})
=\sum_{c\in\{0,1\}}(|c\rangle\langle c|_o\varphi)\otimes|c\rangle_{o'}$,
measuring $o'$ after $\widehat V$ has the same outcome distribution as
measuring $o$ after $V$.  Hence every witness is accepted by $\widehat V$
with the same probability as by $V$.  Let $B$ be the set of auxiliary
qubits whose prescribed bit is $1$; then $o'\notin B$.

\emph{Step 2: conjugation to zero initialization.}  Let
$X=\begin{psmallmatrix}0&1\\1&0\end{psmallmatrix}$ and
$X_B:=\bigotimes_{i\in B}X_i$ on the data register.  Let $C'$ be the
circuit with instructions $(X_Bg_tX_B,T_t)$, where $(g_t,T_t)$ are the
instructions of $\widehat V$.  Since $X_i$ for $i\notin T_t$ cancels,
$X_Bg_tX_B$ is a unitary on the qubits $T_t$.  The unitary of $C'$ is
$X_BU_{\widehat V}X_B$, so
\[
  U_{C'}\bigl(\psi\otimes|0\rangle\bigr)
  =X_BU_{\widehat V}\bigl(\psi\otimes|b\rangle\bigr).
\]
Because $o'\notin B$, the acceptance projection on $o'$ commutes with
$X_B$.  Hence $C'$, with all auxiliary qubits initialized to $0$, accepts
every witness with the same probability as $\widehat V$ does with the
prescribed bits $b$, and hence as $V$ does.  The gates of $C'$ need not
lie in $\mathcal G$; \cref{thm:cited-source-reduction} allows arbitrary
one- and two-qubit unitaries.

\emph{Step 3: the cited construction.}  Apply
\cref{thm:cited-source-reduction} to $C'$, which has $L+1$ gates and
$N_d+1$ data qubits.

\emph{Step 4: undoing the conjugation.}  Let $\mathcal X:=I\otimes X_B$,
acting trivially on the clock, and put
$H_b:=\mathcal XH(C')\mathcal X$.  This is unitarily equivalent to
$H(C')$, so it has the same spectrum.  Term by term: the clock terms act
only on the clock and are unchanged; $H_{\mathrm{out}}$ acts on the clock
and on $o'\notin B$ and is unchanged; by linearity in the gate and
$X_B^2=I$, the propagation term for $X_Bg_tX_B$ becomes the propagation
term for $g_t\in\mathcal G$; and the input term for qubit $i$ becomes
$|a_1\rangle\langle a_1|_1\otimes|1\oplus b_i\rangle\langle1\oplus b_i|_i$.
Thus $H_b$ equals $H(\widehat V)$ except that the input term of each
$i\in B$ penalizes $|0\rangle$ instead of $|1\rangle$.

\emph{Step 5: the spectator qubit.}  Adjoin one further qubit $s$ and put
$H^+:=H_b\otimes I_s$, on $2(L+1)+(N_d+1)+1=N^+$ qubits.  Every term of
$H_b$ becomes the same term placed on the same tuple, and for $i\in B$ the
input term becomes
$|a_1\rangle\langle a_1|_1\otimes|0\rangle\langle0|_i\otimes I_s$.  If
$\{u_k\}$ is an orthonormal eigenbasis of $H_b$, then
$\{u_k\otimes|c\rangle:c\in\{0,1\}\}$ is an orthonormal eigenbasis of
$H^+$ with the same eigenvalues, each occurring twice.  Hence
$\ker H^+=\ker H_b\otimes\CC^2$, and for every real $g$, $H^+\succeq gI$
if and only if $H_b\succeq gI$.

\emph{Step 6: catalogue presentation.}  By Steps 4 and 5, every term of
$H^+$ is either a summand of $H(\widehat V)$, whose gates lie in
$\mathcal G$, or an operator of the additional form in
\cref{lem:catalogue-closure}.  Replacing each term by its decomposition
from that lemma gives a finite-catalogue source presentation of $H^+$,
exactly.  The bound on $t^+$ is counted in \cref{app:catalogue-closure}.

\emph{Promises.}  If some witness is accepted by $V$ with probability one,
it is accepted by $C'$ with probability one, so $\ker H(C')\neq\{0\}$ and
hence $\ker H^+\neq\{0\}$.  If every witness is accepted with probability
at most $\frac13$, the same holds for $C'$, so
$H(C')\succeq \frac{1}{q(L+N_d+2)}I$; unitary equivalence and Step 5
transfer this bound to $H^+$.  Every step is a polynomial-time
computation on the description of $(V,b,o)$.
\end{proof}

\begin{remark}[Prescribed input bits]
\label{rem:prescribed-bits}
In \cite[Definition~7.2]{KingKohler2026} the verifier receives the input
string $x$ as part of its input, while the history-state construction
used there initializes the input qubits to zero
\cite[Appendix~C]{KingKohler2026}, and the tabulated projections
\cite[Table~1]{KingKohler2026} contain no term enforcing a prescribed bit
equal to one.  Steps 1-5 above supply this case with no new local
projection type.
\end{remark}

\begin{corollary}[Finite-catalogue source theorem]
\label{cor:finite-catalogue-source}
Let \(L=(L_{\mathrm{YES}},L_{\mathrm{NO}}) \in\mathsf{QMA}_1(\mathcal G)\), and fix a polynomial-time uniform verifier family $V=\{V_x\}$ for $L$ as in
\cref{def:qma1}.  There exist a deterministic polynomial-time algorithm
$A_V$ and a polynomial $p_V\ge1$ such that, for every input $x$, the
algorithm $A_V$ produces

\begin{enumerate}[label=(\roman*)]
\item integers $N(x),t(x)\geq1$;

\item catalogue labels \(a_j\in\{1,\ldots,13\}, \, 1\leq j\leq t(x)\);

\item ordered support tuples \(S_j = (s_{j,1},\ldots,s_{j,m_{a_j}})\) of distinct elements of $\{1,\ldots,N(x)\}$;
\end{enumerate}

such that
\[
  H_x
  :=
  \sum_{j=1}^{t(x)}
  \bigl(\Pi_{a_j}\bigr)_{S_j}^{(N(x))}
\]
satisfies
\[
  x\in L_{\mathrm{YES}}
  \quad\Longrightarrow\quad
  \ker H_x\neq\{0\},
\]
and
\[
  x\in L_{\mathrm{NO}}
  \quad\Longrightarrow\quad
  H_x\succeq\frac{1}{p_V(|x|)}I.
\]

Moreover,
\[
  N(x)\leq p_V(|x|),
  \qquad
  t(x)\leq p_V(|x|),
\]
and the complete binary description of the finite-catalogue source
presentation has length at most $p_V(|x|)$.  The resulting rational gap is
$g_x:=p_V(|x|)^{-1}\in(0,1]$.
\end{corollary}

\begin{proof}
By uniformity, compute from $x$ the description of $V_x$, its prescribed
auxiliary bits $b_x$, and its output qubit, in time polynomial in $|x|$.
Let $r_V$ be a polynomial with nonnegative coefficients bounding the number
$L(x)$ of instructions and the number $N_d(x)$ of data qubits of $V_x$.
Apply \cref{prop:source-catalogue-identification} to $(V_x,b_x,o)$.
The resulting $N^+$, $t^+$, and presentation length are bounded by
polynomials in $r_V(|x|)$.  Choose a polynomial $p_V\ge1$ with
nonnegative coefficients dominating these bounds and satisfying
$p_V(|x|)\ge q\bigl(2r_V(|x|)+2\bigr)\ge q\bigl(L(x)+N_d(x)+2\bigr)$ for
every $x$.  The two promises follow from
\cref{prop:source-catalogue-identification}.
\end{proof}

For each fixed verifier family $V$, the algorithm $A_V$ and polynomial
$p_V$ work for all inputs $x$; the dependence of $p_V$ on $V$ is only
through the size bound $r_V$.  The reduction below starts from the
finite-catalogue presentation of \cref{cor:finite-catalogue-source}.

\section{Further topological preliminaries}
\label{sec:further-topological-preliminaries}

We collect the standard constructions needed for the gadget analysis:
augmentation, relative homology, simplicial maps, and joins.  Our conventions
follow the usual simplicial ones; see \cite{Hatcher}.

\subsection{Augmented chains and reduced homology}
\label{subsec:augmented-chains}

\begin{definition}[Augmented simplicial chain complex]
\label{def:augmented-chain-complex}
For a finite complex $K$ containing the empty simplex, set
$\widetilde C_r(K;\CC)=C_r(K;\CC)$ for $r\ge0$ and
$\widetilde C_{-1}(K;\CC)=\CC[\varnothing]$.  The augmented boundary agrees
with \cref{def:boundary-operator} in positive degrees and satisfies
\[
  \widetilde\partial_0[v]=[\varnothing],
  \qquad
  \widetilde\partial_{-1}=0.
\]
Its homology is the reduced homology $\widetilde H_r(K;\CC)$.  For
$r\ge1$ this agrees with ordinary homology; if a nonempty $K$ has $c$
connected components, then
\[
  H_0(K;\CC)\cong\CC^c,
  \qquad
  \widetilde H_0(K;\CC)\cong\CC^{c-1},
  \qquad
  \widetilde H_{-1}(K;\CC)=0.
\]
For $K=\{\varnothing\}$, one has
$\widetilde H_{-1}(K;\CC)=\CC[\varnothing]$ and all other reduced groups
vanish.
\end{definition}

\begin{definition}[Augmented cochains and Hodge Laplacian]
\label{def:augmented-hodge-laplacian}
Let $\widetilde C^r(K;\CC)=\operatorname{Hom}_{\CC}(\widetilde
C_r(K;\CC),\CC)$ and let $\widetilde D$ be dual to
$\widetilde\partial$.  Identifying $\widetilde C^{-1}$ with $\CC$,
\[
  \widetilde D_{-1}(1)=\sum_{v\in K_0}v^*.
\]
Give $[\varnothing]$ and its dual norm one.  The augmented Hodge Laplacian
is
\[
  \widetilde\Delta_r
  :=\widetilde D_{r-1}\widetilde D_{r-1}^*
   +\widetilde D_r^*\widetilde D_r,
  \qquad r\ge-1,
\]
with nonexistent maps interpreted as zero.  In the weighted setting the
empty simplex has weight $1$.
\end{definition}

\begin{convention}[Augmentation in join, prism, cone, and depth calculations]
\label{conv:augmented-use}
For joins, prisms, cones, and the private-depth filtered complexes we use
the augmented chain and cochain complexes above, usually suppressing the
tildes on $C,\partial,D,\Delta$.  Thus in these calculations
$\partial_0[v]=[\varnothing]$ and $D_{-1}(1)=\sum_vv^*$.  Ordinary
homology denoted by $H$ remains ordinary unless a reduced group
$\widetilde H$ is displayed.  Relative pairs below contain the empty
simplex in both terms, so it cancels in the quotient.
\end{convention}

\subsection{Relative homology}
\label{subsec:relative-homology}

For a subcomplex $A\subseteq K$, the relative chain complex is
$C_r(K,A;\CC):=C_r(K;\CC)/C_r(A;\CC)$ with the induced boundary.  Its
homology is denoted by $H_r(K,A;\CC)$.  A chain $c\in C_r(K;\CC)$
represents a relative cycle exactly when $\partial c\in C_{r-1}(A;\CC)$.

\begin{definition}[Relative cochains]
\label{def:relative-cochains}
The relative cochain space is the dual of the relative chain space:
\[
  C^r(K,A;\CC):=\operatorname{Hom}_{\CC}(C_r(K,A;\CC),\CC).
\]
It is identified with the cochains on $K$ that vanish on $C_r(A;\CC)$.  In the oriented
simplex basis it is therefore spanned by the dual basis vectors
corresponding to simplices of $K\setminus A$.  The coboundary preserves
this subspace, and its cohomology is denoted by $H^r(K,A;\CC)$.
\end{definition}

\begin{theorem}[Long exact sequence of a pair {\cite[Section~2.1]{Hatcher}}]
\label{thm:long-exact-sequence-pair}
For every inclusion $\iota:A\hookrightarrow K$ there is an exact sequence
\[
  \cdots\to H_{r+1}(K,A;\CC)
  \xrightarrow{\delta_{r+1}}H_r(A;\CC)
  \xrightarrow{\iota_*}H_r(K;\CC)
  \xrightarrow{\pi_*}H_r(K,A;\CC)
  \xrightarrow{\delta_r}H_{r-1}(A;\CC)\to\cdots .
\]
This is the \emph{long exact sequence of the pair} $(K,A)$.
Here $\iota_*$ is induced by the inclusion $\iota$, $\pi_*$ by the quotient
map $C_*(K;\CC)\to C_*(K,A;\CC)$, and $\delta_r$ is the connecting
homomorphism.  Exactness means that, for every $r$,
\[
  \operatorname{im}\delta_{r+1}=\ker\iota_* ,\qquad
  \operatorname{im}\iota_*=\ker\pi_* ,\qquad
  \operatorname{im}\pi_*=\ker\delta_r.
\]
If $c\in C_{r+1}(K;\CC)$ represents a relative cycle, so that
$\partial c\in C_r(A;\CC)$, then
\[
  \delta_{r+1}([c]_{\mathrm{rel}})=[\partial c]\in H_r(A;\CC).
\]
\end{theorem}

\subsection{Simplicial maps}
\label{subsec:simplicial-maps}

A simplicial map $f:K\to L$ is a vertex map satisfying $f(\sigma)\in L$
for every $\sigma\in K$.

\begin{definition}[Image complex]
\label{def:image-complex}
The image complex of a simplicial map $f:K\to L$ is
\[
  f(K):=\{f(\sigma):\sigma\in K\},
\]
a simplicial subcomplex of $L$.
\end{definition}

\begin{definition}[Normalized induced chain map]
\label{def:normalized-chain-map}
For a simplicial map $f:K\to L$, define
\[
  f_\#[v_0,\ldots,v_r]
  :=
  \begin{cases}
  [f(v_0),\ldots,f(v_r)],&f(v_0),\ldots,f(v_r)\text{ distinct},\\
  0,&\text{otherwise},
  \end{cases}
\]
with the usual orientation sign.  Then
\[
  \partial f_\#=f_\#\partial,
\]
so $f_\#$ induces $f_*:H_r(K;\CC)\to H_r(L;\CC)$.  Dually,
$(f^*\alpha)(c):=\alpha(f_\#c)$ defines the induced pullback on cochains
and cohomology.
\end{definition}

A simplicial map that is injective on vertices is a simplicial embedding;
a bijective simplicial map with simplicial inverse is a simplicial
isomorphism.

\subsection{Graph joins and simplicial joins}
\label{subsec:joins}

For finite simple graphs $G$ and $L$ on disjoint vertex sets, their graph
join $G*L$ is obtained from their disjoint union by adding every edge
between $V(G)$ and $V(L)$.  For simplicial complexes $K,L$ on disjoint
vertex sets, their simplicial join is
\[
  K*L:=\{\sigma\cup\tau:\sigma\in K,\ \tau\in L\},
\]
where the empty simplex is included in each factor as in
\cref{conv:augmented-use}.

\begin{lemma}[Clique complexes and joins]
\label{lem:clique-complex-join}
For finite simple graphs $G$ and $L$ on disjoint vertex sets,
\[
  \Cl(G*L)=\Cl(G)*\Cl(L).
\]
\end{lemma}

\begin{proof}
A set $\omega\subseteq V(G)\sqcup V(L)$ is a clique in $G*L$ exactly when
$\omega\cap V(G)$ is a clique in $G$ and $\omega\cap V(L)$ is a clique in
$L$.
\end{proof}

Fix an ordering in which every vertex of $K$ precedes every vertex of $L$.
For oriented simplices $x\in\widetilde C_a(K;\CC)$ and
$y\in\widetilde C_b(L;\CC)$, write $x*y$ for their concatenated oriented
join.  Then
\[
  \widetilde\partial(x*y)
  =(\widetilde\partial x)*y+(-1)^{a+1}x*(\widetilde\partial y),
\]
and hence
\[
  \widetilde C_r(K*L;\CC)
  \cong\bigoplus_{a+b=r-1}
  \widetilde C_a(K;\CC)\otimes\widetilde C_b(L;\CC).
\]

\begin{theorem}[Join K\"unneth formula over a field {\cite[Theorem~3B.5]{Hatcher}}]
\label{thm:join-kunneth}
For finite simplicial complexes $K$ and $L$,
\[
  \widetilde H_r(K*L;\CC)
  \cong
  \bigoplus_{a+b=r-1}
  \widetilde H_a(K;\CC)\otimes\widetilde H_b(L;\CC).
\]
\end{theorem}

\begin{proposition}[Weighted Hodge Laplacian of a join]
\label{prop:weighted-join-laplacian}
Give $K*L$ the product simplex weights
$w(\sigma\cup\tau)=w_K(\sigma)w_L(\tau)$.  Under the orthogonal
decomposition
\[
  \widetilde C^r(K*L;\CC)
  \cong\bigoplus_{a+b=r-1}
  \widetilde C^a(K;\CC)\otimes\widetilde C^b(L;\CC),
\]
each summand is preserved by the augmented weighted Hodge Laplacian, and
on the summand with $a+b+1=r$,
\[
  \widetilde\Delta_{K*L,r}
  =\widetilde\Delta_{K,a}\otimes I
   +I\otimes\widetilde\Delta_{L,b}.
\]
The factor Laplacians remain augmented when $a$ or $b$ is $-1$ or $0$.
\end{proposition}

\begin{proof}
The product rule for weights makes the tensor-product identification
unitary.  Let $S_Ku=(-1)^{a+1}u$ on
$\widetilde C^a(K;\CC)$.  Then
\[
  \widetilde D_{K*L}
  =\widetilde D_K\otimes I+S_K\otimes\widetilde D_L.
\]
Since $S_K$ anticommutes with both $\widetilde D_K$ and
$\widetilde D_K^*$, the mixed terms cancel in
$\widetilde D_{K*L}\widetilde D_{K*L}^*+
 \widetilde D_{K*L}^*\widetilde D_{K*L}$, leaving the stated tensor sum.
\end{proof}

\section{The bow-tie encoding}
\label{sec:bowtie}

The $N$-qubit encoding begins with a seven-vertex graph that has two
independent one-dimensional homology classes representing the computational
basis states $|0\rangle$ and $|1\rangle$.  Taking graph joins then turns
these two classes into the $2^N$ basis classes of an $N$-qubit space.
This encoding follows the bow-tie construction of
King and Kohler~\cite{KingKohler2026}; we include its homological and harmonic
properties explicitly.

\subsection{The single-qubit bow-tie}
\label{subsec:single-bowtie}

\begin{definition}[Single-qubit bow-tie]
\label{def:single-bowtie}
Let $B$ be the finite simple graph with vertex set \(V(B):=\{1,2,3,4,5,6,7\}\) and edge set
\[
  E(B)
  :=
  \bigl\{
    \{1,7\},\{7,3\},\{3,6\},\{6,1\},
    \{1,4\},\{4,2\},\{2,5\},\{5,1\}
  \bigr\}.
\]
The graph consists of two four-cycles sharing the single vertex $1$.
Its clique complex is denoted \(X_1:=\Cl(B)\).

Orient the two four-cycles as
\[
  z_0
  :=
  [1,7]+[7,3]+[3,6]+[6,1]
\]
and
\[
  z_1
  :=
  [1,4]+[4,2]+[2,5]+[5,1].
\]
\end{definition}

\begin{figure}[htbp]
\centering
\begin{tikzpicture}[
  scale=1.05,
  every node/.style={
    circle,
    draw=black,
    fill=white,
    inner sep=1.5pt,
    minimum size=6.5mm,
    font=\small
  },
  leftedge/.style={draw=blue!65!black,line width=0.9pt},
  rightedge/.style={draw=orange!80!black,line width=0.9pt}
]
\node (v1) at (0,0) {1};
\node (v7) at (-1.25,1.05) {7};
\node (v3) at (-2.5,0) {3};
\node (v6) at (-1.25,-1.05) {6};
\node (v4) at (1.25,1.05) {4};
\node (v2) at (2.5,0) {2};
\node (v5) at (1.25,-1.05) {5};
\draw[leftedge]  (v1)--(v7)--(v3)--(v6)--(v1);
\draw[rightedge] (v1)--(v4)--(v2)--(v5)--(v1);
\end{tikzpicture}
\caption{The single-qubit bow-tie graph.  The left four-cycle represents
$z_0$ and the right four-cycle represents $z_1$.  The cycles share the
vertex $1$ but have no edge in common.}
\label{fig:bow-tie}
\end{figure}
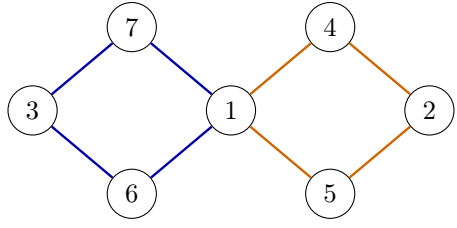

\begin{proposition}[Homology of one bow-tie]
\label{prop:single-bowtie-homology}
The chains $z_0$ and $z_1$ are cycles, and their homology classes form a
basis of the first homology of $X_1$:
\[
  H_1(X_1;\CC)
  =
  \CC\langle[z_0],[z_1]\rangle
  \cong
  \CC^2.
\]
Moreover,
\[
  \widetilde H_r(X_1;\CC)=0
  \qquad
  (r\neq1).
\]
\end{proposition}

\begin{proof}
The alternating boundary of each four-cycle vanishes, so \(\partial_1z_0=0, \, \partial_1z_1=0\). The graph $B$ is connected and has seven vertices and eight edges.  The
six edges of any spanning tree have linearly independent boundaries,
while every edge boundary has total vertex coefficient zero.  Therefore \(\rank\partial_1=6\) and hence \(\dim\ker\partial_1=8-6=2\). The cycles $z_0$ and $z_1$ are linearly independent because their edge
supports are disjoint.

The graph contains no triangles, so its clique complex has no
two-simplices.  Consequently, \(\im\partial_2=\{0\}\). It follows that the two cycle classes form a basis of
$H_1(X_1;\CC)$.  Connectedness gives
$\widetilde H_0(X_1;\CC)=0$, and there are no simplices in dimensions
greater than one.
\end{proof}

\subsection{The \texorpdfstring{$N$}{N}-qubit base complex}
\label{subsec:n-qubit-base}

\begin{definition}[$N$-qubit bow-tie base]
\label{def:bowtie-base}
Let $N\geq1$, and let \(B_1,\ldots,B_N\) be labelled copies of $B$ on pairwise disjoint vertex sets.  Define the
graph and its clique complex by
\[
  G_N
  :=
  B_1*\cdots*B_N,
  \qquad
  X_N
  :=
  \Cl(G_N).
\]
By \cref{lem:clique-complex-join}, \(X_N = X_1^{(1)}*\cdots*X_1^{(N)}\), where $X_1^{(r)}:=\Cl(B_r)$.

The \emph{encoded degree} is \(q_N:=2N-1\). For $i\in\{0,1\}$, let $z_i^{(r)}$ denote the copy of $z_i$ in the
$r$-th factor.  For a bit string \(b=b_1\cdots b_N\in\{0,1\}^N\), define the join cycle
\[
  z_b
  :=
  z_{b_1}^{(1)}*\cdots*z_{b_N}^{(N)}
  \in
  C_{q_N}(X_N;\CC).
\]
\end{definition}

\begin{proposition}[Homology of the $N$-qubit base]
\label{prop:bowtie-base-homology}
The reduced homology of $X_N$ is concentrated in degree $q_N$:
\[
  \widetilde H_r(X_N;\CC)=0
  \qquad
  (r\neq q_N).
\]
In degree $q_N$,
\[
  H_{q_N}(X_N;\CC)
  =
  \CC\langle[z_b]:b\in\{0,1\}^N\rangle
\]
and therefore \(\dim_{\CC}H_{q_N}(X_N;\CC)=2^N\).
\end{proposition}

\begin{proof}
By \cref{prop:single-bowtie-homology}, every factor $X_1^{(r)}$ has
reduced homology only in degree one, where it is two-dimensional with
basis $[z_0^{(r)}],[z_1^{(r)}]$.  Applying the join K\"unneth formula of
\cref{thm:join-kunneth} repeatedly, the only nonzero reduced-homology
term of the $N$-fold join occurs in degree \(\underbrace{1+\cdots+1}_{N\text{ terms}}+(N-1)=2N-1=q_N\). In that degree,
\[
  \widetilde H_{q_N}(X_N;\CC)
  \cong
  \bigotimes_{r=1}^N
  \widetilde H_1(X_1^{(r)};\CC).
\]
The tensor-product basis corresponds exactly to the join classes \([z_{b_1}^{(1)}*\cdots*z_{b_N}^{(N)}] =[z_b], \, b\in\{0,1\}^N\). Since $q_N\geq1$, reduced and ordinary homology agree in this degree.
\end{proof}

\begin{definition}[Homological encoding]
\label{def:homological-encoding}
The \emph{homological encoding map} is the linear isomorphism
\[
  \mathfrak e_N:
  H_{q_N}(X_N;\CC)
  \longrightarrow
  \Qenc_N=(\CC^2)^{\otimes N}
\]
defined on the basis of \cref{prop:bowtie-base-homology} by
\[
  \mathfrak e_N([z_b])
  :=
  |b\rangle,
  \qquad
  b\in\{0,1\}^N.
\]
\end{definition}

\subsection{The harmonic encoding}
\label{subsec:harmonic-encoding}

For the spectral analysis, normalize the harmonic representatives of the
encoded classes to obtain an isometric map into the cochain Hilbert space.

Using the oriented-simplex basis to identify edge chains with edge
cochains, for $z_0,z_1$ from \cref{def:single-bowtie}, define \(h_0:=\frac12z_0, \, h_1:=\frac12z_1\).

\begin{lemma}[Normalized harmonic basis of one bow-tie]
\label{lem:single-bowtie-harmonic-basis}
The cochains $h_0$ and $h_1$ form an orthonormal basis of \(\mathscr H^1(X_1) = \ker\Delta_1(X_1)\).
\end{lemma}

\begin{proof}
Because $X_1$ has no two-simplices, \(D_1h_i=0\). Under the unweighted chain-cochain identification, \(D_0^*=\partial_1\). Since $z_i$ is a cycle, \(D_0^*h_i = \frac12\partial_1z_i =0\). The Hodge energy identity therefore gives \(\Delta_1h_i=0\).

Each $z_i$ is the sum of four orthonormal oriented-edge vectors, so \(\|z_i\|^2=4 \,\text{and hence}\, \|h_i\|=1\). Their edge supports are disjoint, and therefore \(\langle h_0,h_1\rangle=0\). Finally, the combinatorial Hodge theorem and
\cref{prop:single-bowtie-homology} give \(\dim\mathscr H^1(X_1) = \dim H_1(X_1;\CC) =2\). Thus $h_0,h_1$ form the asserted orthonormal basis.
\end{proof}

For $b=b_1\cdots b_N\in\{0,1\}^N$, define
\[
  h_b
  :=
  h_{b_1}^{(1)}*\cdots*h_{b_N}^{(N)}
  \in
  C^{q_N}(X_N;\CC),
\]
where the join is understood through the chain-cochain simplex-basis
identification.

\begin{proposition}[Encoded harmonic basis]
\label{prop:encoded-harmonic-basis}
The cochains \(\{h_b:b\in\{0,1\}^N\}\) form an orthonormal basis of
\[
  \Eenc_N
  :=
  \mathscr H^{q_N}(X_N)
  =
  \ker\Delta_{q_N}(X_N).
\]
Moreover, under the simplex-basis identification, \(h_b=2^{-N}z_b\).
\end{proposition}

\begin{proof}
By \cref{lem:single-bowtie-harmonic-basis}, each factor $h_{b_r}^{(r)}$
is harmonic.  Repeated application of the join-Laplacian formula in
\cref{prop:weighted-join-laplacian}, with all vertex weights equal to one,
gives \(\Delta_{q_N}(X_N)h_b=0\). The tensor-product inner product gives
\[
  \langle h_b,h_c\rangle
  =
  \prod_{r=1}^N
  \langle h_{b_r}^{(r)},h_{c_r}^{(r)}\rangle
  =
  \delta_{bc}.
\]
There are $2^N$ such cochains.  By the combinatorial Hodge theorem and
\cref{prop:bowtie-base-homology},
\[
  \dim\mathscr H^{q_N}(X_N)
  =
  \dim H_{q_N}(X_N;\CC)
  =2^N.
\]
They therefore form an orthonormal basis.  Finally, multilinearity of the
join and $h_i=z_i/2$ give \(h_b=2^{-N}z_b\).
\end{proof}

\begin{definition}[Encoded harmonic isometry and projection]
\label{def:encoded-harmonic}
Define \(J_N: \Qenc_N \longrightarrow \Eenc_N\) on computational-basis vectors by \(J_N|b\rangle:=h_b\). By \cref{prop:encoded-harmonic-basis},
\[
  \langle J_N\psi,J_N\varphi\rangle
  =
  \langle\psi,\varphi\rangle
  \qquad
  (\psi,\varphi\in\Qenc_N),
\]
and $J_N$ is a unitary isomorphism from $\Qenc_N$ onto $\Eenc_N$.

The \emph{encoded harmonic projection} is
\[
  Q_N
  :=
  J_NJ_N^*
  =
  \Proj_{\Eenc_N}.
\]
It is the orthogonal projection from $C^{q_N}(X_N;\CC)$ onto the encoded
harmonic subspace.
\end{definition}

\subsection{Target states and target cycles}
\label{subsec:target-cycles}

\begin{definition}[Local state with integer coordinates and target cycle]
\label{def:target-cycle}
Let $m\geq1$, and let
\[
  |\phi\rangle
  =
  \frac1{\mathcal Z}
  \sum_{b\in\{0,1\}^m}a_b|b\rangle,
  \qquad
  a_b\in\ZZ,
\]
where the coefficients are not all zero and \(\mathcal Z := \left(\sum_ba_b^2\right)^{1/2}\). The associated \emph{unnormalized target cycle} is
\[
  j_\phi
  :=
  \sum_{b\in\{0,1\}^m}a_bz_b
  \in
  Z_{q_m}(X_m;\CC).
\]

\end{definition}

Under the homological encoding,
\[
  \mathfrak e_m([j_\phi])
  =
  \sum_ba_b|b\rangle
  =
  \mathcal Z|\phi\rangle.
\]
Consequently, the homology line \(\CC[j_\phi] \subseteq H_{q_m}(X_m;\CC)\) corresponds exactly to the quantum-state line \(\CC|\phi\rangle \subseteq \Qenc_m\). Likewise, under the chain-cochain simplex-basis identification,
\[
  J_m|\phi\rangle
  =
  \frac1{\mathcal Z}\sum_ba_bh_b
  =
  \frac{1}{2^m\mathcal Z}j_\phi.
\]
Thus the homological and harmonic encodings represent the same
one-dimensional state space, with different normalizations.

\section{The gap in the hardness proof of King and Kohler}
\label{sec:king-kohler-gap}

The point at issue is the passage from an auxiliary simplicial filling to
the clique complex of the graph produced by the reduction.  The bow-tie
encoding, the use of local fillings and small vertex weights, and the
spectral-sequence strategy are retained here.  The problem is the argument
used to transfer topological information through the noninjective folding
of the local gadgets.

Unless another version is specified, we use the numbering of
\cite{KingKohler2026}.  Lemmas~8.3 and~8.4 provide the preservation
statements used by Lemma~8.5; Theorems~8.7 and~8.8 then apply this framework
to the nontrivially folded entangled gadgets.  Because these conclusions
are the local input for the subsequent spectral and global arguments, the
failure of the preservation statements invalidates the hard-direction proof
as a whole.

\subsection{Folded fillings and the three complexes}

Fix a local projector $P=|\phi\rangle\langle\phi|$ whose encoded range is
represented by a cycle $j$ in the bow-tie base, so that $[j]$ corresponds
to the line $\CC|\phi\rangle$.  The King-Kohler construction begins with
an auxiliary complex $K$, a cycle $c$ on $K$, and a vertex map
$f:V(K)\to V(X)$ whose signed image is intended to be $j$.

There is a minor ambiguity in the original description of this map.  In
Section~8.2 of the arXiv version \cite{KingKohler}, the relation $R$ is
written using vertices of the auxiliary complex $K$, whereas Equation~(9)
defines a map with values in the target cycle $J$.  We interpret the two
vertex sets as separately labelled and $f$ as a map from the former to the
latter.  This resolves the notation but not the preservation issue below.

The auxiliary complex is thickened to a prism and its upper rim is coned
off.  Before folding there is therefore a chain $\Gamma$ with
$\partial\Gamma=c$.  If $q$ is the induced map on the thickened-and-coned
complex and the lower-rim map satisfies the signed identity $f_\#c=j$, then
\[
  \partial(q_\#\Gamma)=q_\#(\partial\Gamma)=f_\#c=j.
\]
Thus the prescribed cycle remains a boundary after folding.

The reduction, however, involves three distinct simplicial complexes: the
unfolded filling, its folded simplicial image, and the clique complex of the
final graph.  The first passage is induced by a generally noninjective
vertex map; the second takes the full clique completion after gluing to the
base.  Neither passage preserves the required topology automatically.

\subsection{Failure of the preservation claims}

A four-vertex path already shows both failures.  Let $P$ have edges
$[v_0v_1],[v_1v_2],[v_2v_3]$ and their faces, and identify the endpoints by
$q(v_0)=q(v_3)=x$.  The image has edges
$[xv_1],[v_1v_2],[v_2x]$, so
\[
  H_1(q(P);\CC)\cong\CC,
\]
even though $P$ has no positive-dimensional homology.  Vertex
identification can therefore create a new cycle.

The same image is not flag.  Its one-skeleton is the triangle on
$\{x,v_1,v_2\}$, but the simplex $[xv_1v_2]$ is absent.  Hence
\[
  q(P)\ne\Cl(\Sk_1q(P)).
\]
This is the obstruction to the general acyclicity and clique-complex
preservation assertions in Lemmas~8.3 and~8.4.

A separate completion issue appears after the gadget is glued to the base.
For graphs $G_1,G_2$ with common vertices one need not have
$\Cl(G_1\cup G_2)=\Cl(G_1)\cup\Cl(G_2)$.  For example, if a private vertex
$p$ is adjacent in the gadget graph to base vertices $b_1,b_2$, while the
edge $[b_1b_2]$ belongs to the base graph, then the clique complex of the
union contains the simplex $[pb_1b_2]$ whether or not it belonged to either
simplicial piece separately.

What does survive folding is the chain-level statement used above: if
$z=\partial u$, then $q_\#z=\partial(q_\#u)$.  Applied to the prescribed
filling, this proves that $[j]$ maps to zero.  It does not determine the
whole kernel \(\ker\bigl(H_*(X)\to H_*(Y)\bigr)\), and in particular does
not prove that this kernel is exactly $\CC[j]$.  The latter statement is
what is needed to identify the surviving encoded subspace.  In the present
proof it follows from the long exact sequence of the relative pair after
the completed local clique complex has been analyzed directly.

\subsection{Signed chain data}

For a nonbasis target state
$|\phi\rangle=\sum_ba_b|b\rangle$, the folding map must realize the signed
chain identity
\[
  f_\#c=\sum_ba_bz_b.
\]
Matching only the unoriented support is insufficient: permutation of the
image vertices changes orientation signs, and a simplex whose image has a
repeated vertex contributes zero to the normalized chain map.  The present
construction therefore records and verifies the complete coefficient-wise
identities for every catalogue gadget, including the orientation reversals
needed for the two Pythagorean gadgets.  The explicit signed calculations
are recorded in Appendix~\ref{app:diagnosis}.

The signed identities below concern the explicit maps of
\cref{constr:local-gadget}; their graph-builder provenance is recorded in
\cref{app:reproducibility}.

\subsection{Dependence of the hard-direction argument}

Lemma~8.5 uses the preservation claims of Lemmas~8.3 and~8.4 to analyze the
glued gadget, and Theorems~8.7 and~8.8 invoke that result for the entangled
gadgets.  Because these conclusions are the local input for the subsequent
spectral and global arguments, their failure invalidates the hard-direction
proof as a whole.

The local topology then supplies the cochain bases and filtered complex
used in the spectral analysis.  If the completed local complex has
additional simplices or homology, the killed base classes, the
fixed-depth kernels and images, the later spectral-sequence pages, the long
differential, and the small Hodge eigenvalues can all change.  These local
spectral statements are subsequently used in padding and in the global
many-gadget estimate.

Consequently, the King-Kohler argument does not establish the hard
direction.  This conclusion concerns the proof, not the truth of the
hardness theorem: the present article proves the hard direction after
replacing the unsupported preservation step by calculations on the output
clique complexes.  The computational-basis gadgets can be checked directly,
and the independent $\mathsf{QMA}$-containment argument is unaffected.

\subsection{Requirements for the replacement proof}
\label{subsec:king-kohler-gap-obligations}

For each local catalogue gadget it remains to establish, on the clique
complex of the graph produced by the reduction:
\begin{enumerate}[label=(\arabic*)]
\item the intended target chain is filled with the correct orientation;
\item the kernel of the map from encoded base homology is exactly the
one-dimensional projector range;
\item the full flag completion contains the displayed base as the required
induced subcomplex, with distinct private occurrences isolated for the
global construction;
\item the filtered cochain complex has the required intermediate vanishings
and a unique surviving private class;
\item the long differential gives the intended rank-one low-energy penalty
and no additional local low-energy modes;
\item these statements persist under relabelling, placement, padding, and
the combination of polynomially many gadget occurrences.
\end{enumerate}

Condition C1 supplies (1); C0-C2 and the long exact sequence give (2)-(3);
C3 and C4 give the finite filtered input for (4); the filtered-Hodge
analysis proves (5); and the placement, join-padding, and global estimates
establish (6).
\section{The local gadget construction}
\label{sec:local-gadget}

Recall that the catalogue consists of the local rank-one projections \(\Pi_a=|\phi^{(a)}\rangle\langle\phi^{(a)}|\) with localities \(m_a\). Under the bow-tie encoding, \(\CC\phi^{(a)}\) is represented in \(X_{m_a}\) by a cycle \(j_a\in Z_{q_{m_a}}(X_{m_a};\ZZ)\).  The local
gadget is required to contain a chain \(\Gamma_a\) with \(\partial\Gamma_a=j_a\); Section~\ref{sec:topological-correctness} shows
that no other encoded line is removed.

We use the thickening, coning, and folding procedure of King and Kohler
\cite[Section~8.2]{KingKohler}, with the source-cycle data and signed
maps specified below.

Start with a finite source complex \(K_a\), a cycle \(c_a\), and a
simplicial map \(f_a:K_a\to X_{m_a}\) with \((f_a)_\#c_a=j_a\).  Join two
copies of \(K_a\) by a triangulated prism, cone off the top copy, and attach
the bottom copy to \(X_{m_a}\).  The definitions below fix the source data,
orientations, folding map, and final clique complex.

\subsection{The source complex and its cycle}
\label{subsec:source-cycle-data}

\begin{definition}[Oriented triangulated sphere]
\label{def:oriented-triangulated-sphere}
Let \(d\ge1\), and let \(K\) be a finite simplicial complex whose realization is homeomorphic to \(S^d\) (as in \cref{def:geometric-realization}).  An orientation of \(K\) is a coherent orientation of its \(d\)-simplices: the two top-dimensional simplices incident to each \((d-1)\)-simplex induce opposite orientations on their common face.  The
oriented sum of all top-dimensional simplices is called the \emph{fundamental
cycle}:
\[
  [K]:=\sum_{\sigma\in K_d}[\sigma]\in C_d(K;\ZZ),
\]
where every top-dimensional simplex carries its chosen orientation.  Then \(\partial[K]=0, \, H_d(K;\ZZ)=\ZZ\langle[K]\rangle\).
\end{definition}

\begin{definition}[Source-cycle data]
\label{def:source-cycle-data}
For a catalogue index \(a\), put \(m:=m_a\) and \(d:=q_m=2m-1\), and let \(j_a\in Z_d(X_m;\ZZ)\) be the target cycle associated with \(\phi^{(a)}\).  The source-cycle data are \((K_a,c_a,f_a,\prec_a)\), where:
\begin{enumerate}[label=(\roman*)]
\item \(K_a\) is a finite simplicial complex;

\item \(\prec_a\) is a total order on \(V(K_a)\), which will be used below
to triangulate the prism over \(K_a\);

\item \(c_a\in C_d(K_a;\QQ)\) is a nonzero rational cycle: \(\partial c_a=0\); and

\item \(f_a:K_a\to X_m\) is a simplicial map satisfying \((f_a)_\#c_a=j_a\).
\end{enumerate}

For each catalogue type \(a\), we fix these finite data and verify their
identities by the certificates in \cref{sec:certificates}.  The map \((f_a)_\#\) uses the orientation and degeneracy convention of
\cref{def:normalized-chain-map}.  The construction applies to arbitrary
finite source complexes, including those whose maximal simplices have
different dimensions; sphere representatives are an optional choice.
\end{definition}

\begin{proposition}[Source-sphere representatives]
\label{prop:source-sphere-existence}
Let \(m\ge1\), put \(d:=q_m=2m-1\), and let \(X_m\) be the \(m\)-fold join of bow-tie complexes.  For every cycle \(j\in Z_d(X_m;\ZZ)\) with integer coefficients,
there exist a finite oriented triangulation \(K\) of \(S^d\) and a
simplicial map \(f:K\to X_m\) with \(f_\#[K]=j\). After \(K\) has been chosen, its vertex set may be equipped with any
total order.  Thus a source sphere is one possible choice of source-cycle
data for every catalogue target \(j_a\); the explicit certified choices
above need not be spheres.
\end{proposition}

\begin{proof}
See \cref{app:source-sphere-existence}.
\end{proof}

\subsection{Prism, cone, and folding}
\label{subsec:prism-cone}

The chain operators in this subsection use
\cref{conv:augmented-use}, including in degrees zero and $-1$.

\begin{definition}[Standard simplicial prism]
\label{def:prism}
Let \(K\) be a finite simplicial complex with a fixed total order on its
vertices.  For each \(v\in V(K)\), introduce copies \(v^0:=(v,0)\) and \(v^1:=(v,1)\). The standard simplicial prism \(\Cyl(K)\) is generated by the simplices
\[
  [v_0^0,\ldots,v_i^0,v_i^1,\ldots,v_r^1],
  \qquad 0\leq i\leq r,
\]
for every ordered simplex \([v_0<\cdots<v_r]\) of \(K\), together with all
their faces.  Write \(i_0(v)=v^0, \, i_1(v)=v^1\) for the bottom and top inclusions.  The prism operator is
\[
  \mathsf{Prism}_r[v_0,\ldots,v_r]
  :=
  \sum_{i=0}^{r}(-1)^i
  [v_0^0,\ldots,v_i^0,v_i^1,\ldots,v_r^1].
\]
In degree $-1$, set $\mathsf{Prism}_{-1}[\varnothing]=0$.
\end{definition}
\begin{lemma}[Prism identity]
\label{lem:prism-identity}
For every chain \(c\in C_r(K;\CC)\), \[
  \partial\mathsf{Prism}(c)
  +\mathsf{Prism}(\partial c)
  =(i_1)_\#c-(i_0)_\#c.
\]
\end{lemma}

\begin{proof}
It is enough to consider one oriented simplex \(\sigma=[v_0,\ldots,v_r]\). For \(0\le i\le r\), write \(\tau_i := [v_0^0,\ldots,v_i^0,v_i^1,\ldots,v_r^1]\). Then \(\mathsf{Prism}(\sigma) =\sum_{i=0}^r(-1)^i\tau_i\). Expanding the boundary gives
\[
\begin{aligned}
  \partial\mathsf{Prism}(\sigma)
  &=
  \sum_{i=0}^r(-1)^i\partial\tau_i \\
  &=
  (i_1)_\#\sigma-(i_0)_\#\sigma
  -
  \sum_{j=0}^r(-1)^j
  \mathsf{Prism}
  [v_0,\ldots,\widehat{v_j},\ldots,v_r].
\end{aligned}
\]
Indeed, the top face of \(\tau_0\) gives \((i_1)_\#\sigma\), the bottom
face of \(\tau_r\) gives \(-(i_0)_\#\sigma\), and the remaining faces
group into the final sum.  Since \(\partial\sigma = \sum_{j=0}^r(-1)^j [v_0,\ldots,\widehat{v_j},\ldots,v_r]\), we obtain
\[
  \partial\mathsf{Prism}(\sigma)
  =
  (i_1)_\#\sigma-(i_0)_\#\sigma
  -\mathsf{Prism}(\partial\sigma).
\]
Rearranging proves the identity.
\end{proof}

\begin{definition}[Simplicial cone]
\label{def:cone}
For a new vertex \(v_*\), define
\[
  \Cone_{v_*}(K):=K*\{\varnothing,\{v_*\}\}.
\]
The cone operator \(\mathsf C_{v_*}:C_r(K;\CC) \longrightarrow C_{r+1}(\Cone_{v_*}(K);\CC)\) is given by
\[
  \mathsf C_{v_*}[w_0,\ldots,w_r]
  :=[v_*,w_0,\ldots,w_r].
\]
With the augmented convention \(\mathsf C_{v_*}[\varnothing]=[v_*]\).
\end{definition}

\begin{lemma}[Cone identity]
\label{lem:cone-identity}
For every chain \(c\in C_r(K;\CC)\), \[
  \partial\mathsf C_{v_*}(c)
  =c-\mathsf C_{v_*}(\partial c).
\]
\end{lemma}

\begin{proof}
For $r=-1$, the identity is
$\partial[v_*]=[\varnothing]$, which is the augmented convention.
For $r\ge0$, it is enough to consider one oriented simplex \(\sigma=[w_0,\ldots,w_r]\).  Then
\[
\begin{aligned}
  \partial\mathsf C_{v_*}(\sigma)
  &=
  \partial[v_*,w_0,\ldots,w_r] \\
  &=
  [w_0,\ldots,w_r]
  -
  \sum_{j=0}^r(-1)^j
  [v_*,w_0,\ldots,\widehat{w_j},\ldots,w_r] \\
  &=
  \sigma-\mathsf C_{v_*}(\partial\sigma).
\end{aligned}
\]
The result for arbitrary chains follows by linearity.
\end{proof}

\begin{center}
\begin{minipage}{0.98\textwidth}
\centering

\begin{tikzpicture}[
  line cap=round,
  line join=round,
  vertex/.style={circle,fill=black,inner sep=1.5pt},
  baseedge/.style={draw=black,line width=0.9pt},
  prismedge/.style={draw=blue!65!black,line width=0.85pt},
  coneedge/.style={draw=orange!75!black,line width=0.9pt},
  vlabel/.style={font=\scriptsize,fill=white,inner sep=1pt},
  paneltitle/.style={font=\small\bfseries}
]


\coordinate (k0) at (-5.8,0);
\coordinate (k1) at (-3.8,0);
\coordinate (k2) at (-4.8,1.75);

\draw[baseedge] (k0)--(k1)--(k2)--cycle;

\foreach \p in {k0,k1,k2}
  \node[vertex] at (\p) {};

\node[vlabel,anchor=north east] at (k0) {$v_0$};
\node[vlabel,anchor=north west] at (k1) {$v_1$};
\node[vlabel,anchor=south] at (k2) {$v_2$};

\node[paneltitle] at (-4.8,2.35)
  {input cycle $K$};

\node[font=\small] at (-4.8,-0.55)
  {fundamental cycle $[K]$};


\coordinate (p0) at (-1.7,0);
\coordinate (p1) at (1.7,0);
\coordinate (p2) at (0,2.65);

\coordinate (q0) at (-0.68,0.65);
\coordinate (q1) at (0.68,0.65);
\coordinate (q2) at (0,1.75);

\fill[blue!5]  (p0)--(p1)--(q1)--cycle;
\fill[blue!10] (p0)--(q1)--(q0)--cycle;

\fill[blue!5]  (p1)--(p2)--(q2)--cycle;
\fill[blue!10] (p1)--(q2)--(q1)--cycle;

\fill[blue!5]  (p0)--(p2)--(q2)--cycle;
\fill[blue!10] (p0)--(q2)--(q0)--cycle;

\draw[baseedge]  (p0)--(p1)--(p2)--cycle;
\draw[prismedge] (q0)--(q1)--(q2)--cycle;

\draw[prismedge]
  (p0)--(q0)
  (p1)--(q1)
  (p2)--(q2);

\draw[prismedge]
  (p0)--(q1)
  (p1)--(q2)
  (p0)--(q2);

\foreach \p in {p0,p1,p2,q0,q1,q2}
  \node[vertex] at (\p) {};

\node[vlabel,anchor=north east] at (p0) {$v_0^0$};
\node[vlabel,anchor=north west] at (p1) {$v_1^0$};
\node[vlabel,anchor=south] at (p2) {$v_2^0$};

\node[vlabel,anchor=north east] at (q0) {$v_0^1$};
\node[vlabel,anchor=north west] at (q1) {$v_1^1$};
\node[vlabel,anchor=south west] at (q2) {$v_2^1$};

\node[paneltitle] at (0,3.35)
  {prism $\Cyl(K)$};

\node[font=\small] at (0,-0.55)
  {$i_0(K)$ outside; $i_1(K)$ inside};


\coordinate (c0) at (4.0,0);
\coordinate (c1) at (6.0,0);
\coordinate (c2) at (5.0,1.75);
\coordinate (ca) at (5.0,0.66);

\fill[orange!8]  (ca)--(c0)--(c1)--cycle;
\fill[orange!14] (ca)--(c1)--(c2)--cycle;
\fill[orange!5]  (ca)--(c2)--(c0)--cycle;

\draw[baseedge] (c0)--(c1)--(c2)--cycle;

\draw[coneedge]
  (ca)--(c0)
  (ca)--(c1)
  (ca)--(c2);

\foreach \p in {c0,c1,c2,ca}
  \node[vertex] at (\p) {};

\node[vlabel,anchor=north east] at (c0) {$v_0$};
\node[vlabel,anchor=north west] at (c1) {$v_1$};
\node[vlabel,anchor=south] at (c2) {$v_2$};
\node[vlabel,anchor=south west] at (ca) {$v_*$};

\node[paneltitle] at (5.0,2.35)
  {cone $\Cone_{v_*}(K)$};

\node[font=\small] at (5.0,-0.55)
  {$v_*$ fills the cycle $[K]$};

\end{tikzpicture}

\refstepcounter{figure}
\label{fig:prism-and-cone-planar}
\smallskip

{\small\textbf{Figure~\thefigure.}
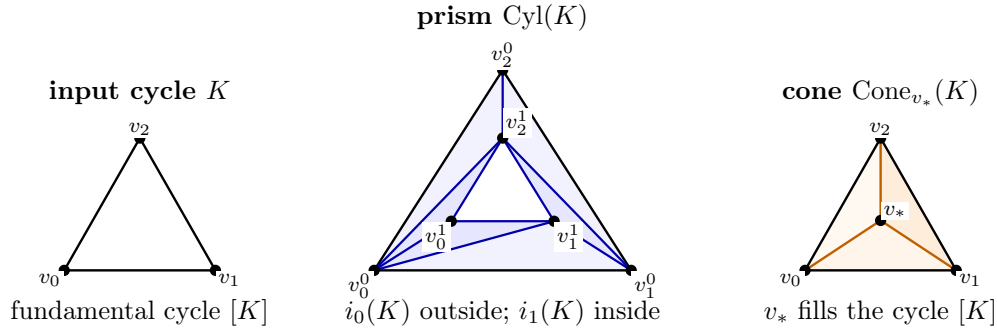
Prism and cone over the triangle boundary, used here as a simplicial
example with source order $v_0<v_1<v_2$.  The planar annulus displays the
prism's six triangles and two boundary copies.  The cone is the filled
triangle with apex at its centre.}

\end{minipage}
\end{center}

For a cycle \(c\), these two pictures have the boundary identities
\[
  \partial\mathsf{Prism}(c)=(i_1)_\#c-(i_0)_\#c,
  \qquad
  \partial\mathsf C_{v_*}(c)=c.
\]
Coning the top copy and subtracting the prism therefore leaves the
bottom copy as boundary.

\begin{construction}[Local thickening-coning-folding gadget]
\label{constr:local-gadget}
Fix a catalogue index \(a\) and the source data of
\cref{def:source-cycle-data}.  Put \(m:=m_a, \, d:=q_m, \, X:=X_m, \, K:=K_a\).

\begin{enumerate}[label=(\roman*)]

\item Form the unfolded prism-cone attachment
\[
  U_a
  :=
  \Cyl(K)\cup_{i_1(K)}
  \Cone_{v_a^*}(i_1(K)),
\]
where \(v_a^*\) is a new apex.  The notation means that the cone is
attached to the prism along their common copy \(i_1(K)\), which is both
the top of the prism and the base of the cone.

\item Fold the bottom copy of \(K\) into \(X\).  More precisely, on the
disjoint union \(V(X)\sqcup V(U_a)\) impose exactly the relations
\[
  i_0(v)\sim f_a(v)
  \qquad
  (v\in V(K)),
\]
and no others.  Let \(q_a\) be the resulting quotient vertex map and let \(\widetilde K_a:=q_a(U_a)\) be the folded simplicial image of \(U_a\), as defined in
\cref{def:image-complex}.

\item Define the local gadget graph and its actual clique complex by
\[
  G_a^{\mathrm{loc}}
  :=
  G_m\cup\Sk_1(\widetilde K_a),
  \qquad
  Y_a^{\mathrm{loc}}
  :=
  \Cl(G_a^{\mathrm{loc}}).
\]
Here \(\Sk_1(\widetilde K_a)\) is the one-skeleton graph defined in
\cref{def:one-skeleton}, and the first formula is a graph union on the
common folded base vertices.  Clique completion may add simplices beyond \(X_m\cup\widetilde K_a\).

\item The prism-cone construction also supplies the filling chain before folding
\[
  \Gamma_a^{\mathrm{raw}}
  :=
  \mathsf C_{v_a^*}\bigl((i_1)_\#c_a\bigr)
  -\mathsf{Prism}(c_a)
  \in C_{d+1}(U_a;\CC).
\]
Its folded image is
\[
  \Gamma_a
  :=
  (q_a)_\#\Gamma_a^{\mathrm{raw}}
  \in C_{d+1}(Y_a^{\mathrm{loc}};\CC).
\]

\end{enumerate}

The private vertex set is \(V_a^{\mathrm{priv}} := V(Y_a^{\mathrm{loc}})\setminus V(X_m)\).

The final object used in the rest of the proof is the pair \(X_m\subseteq Y_a^{\mathrm{loc}} =\Cl(G_a^{\mathrm{loc}})\), together with the filling chain \(\Gamma_a\); the complexes \(U_a\) and \(\widetilde K_a\) and the map \(q_a\) record how this object was
constructed.
\end{construction}


\begin{figure}[htbp]
\centering
\resizebox{\textwidth}{!}{%
\begin{tikzpicture}[
  bv/.style={circle,fill=black,minimum size=3.4pt,inner sep=0pt},
  pv/.style={circle,fill=blue!65!black,minimum size=3.2pt,inner sep=0pt},
  pe/.style={draw=blue!65!black,line width=0.6pt},
  ph/.style={draw=blue!55!black,line width=0.45pt,densely dashed},
  wl/.style={draw=blue!35!black,line width=0.4pt},
  wh/.style={draw=blue!35!black,line width=0.35pt,densely dashed},
  be/.style={draw=black,line width=0.95pt},
  bh/.style={draw=black,line width=0.95pt,densely dashed},
  gl/.style={draw=black!35,line width=0.45pt},
  gd/.style={draw=black!22,line width=0.35pt,densely dotted},
  fl/.style={draw=blue!70!black,line width=0.9pt},
  re/.style={draw=red!80!black,line width=1.3pt},
  ttl/.style={font=\small},
  lab/.style={font=\tiny},
  slab/.style={font=\tiny,text=blue!55!black}
]


\begin{scope}[shift={(0,0)}]
  \node[ttl] at (0,2.45) {1.\ the source data};

  \foreach \k in {0,...,7}{
    \pgfmathsetmacro{\A}{45*\k+22.5}
    \coordinate (s\k) at
      ({-2.35+0.95*cos(\A)},{0.35+0.95*sin(\A)});
  }

  \draw[pe]
    (s0)--(s1)--(s2)--(s3)--(s4)--(s5)--(s6)--(s7)--cycle;

  \foreach \k in {0,...,7}{
    \node[pv] at (s\k){};
  }

  \node[slab,right=1pt]       at (s0) {$\hat x$};
  \node[slab,left=1pt]        at (s4) {$\hat x_1$};
  \node[slab,above right=0pt] at (s1) {$\hat a_3$};
  \node[slab,below right=0pt] at (s7) {$\hat b_3$};

  \node[lab] at (-2.35,-1.05)
    {illustrative source \(K_a\), a triangulated \(S^1\)};

  \coordinate (c)  at ( 1.95, 0.35);
  \coordinate (l1) at ( 1.30, 0.80);
  \coordinate (l2) at ( 0.55, 0.35);
  \coordinate (l3) at ( 1.30,-0.10);
  \coordinate (r1) at ( 2.60,-0.10);
  \coordinate (r2) at ( 3.35, 0.35);
  \coordinate (r3) at ( 2.60, 0.80);

  \draw[re] (c)--(l1)--(l2)--(l3)--(c);
  \draw[re] (c)--(r1)--(r2)--(r3)--(c);

  \foreach \p in {c,l1,l2,l3,r1,r2,r3}{
    \node[bv] at (\p){};
  }

  \node[lab,below=1pt]      at (c)  {$x$};
  \node[lab,above=1pt]      at (l1) {$a_3$};
  \node[lab,below left=0pt] at (l2) {$a_2$};
  \node[lab,below=1pt]      at (l3) {$a_4$};
  \node[lab,below=1pt]      at (r1) {$b_4$};
  \node[lab,right=1pt]      at (r2) {$b_2$};
  \node[lab,above=1pt]      at (r3) {$b_3$};

  \node[lab] at (1.95,-1.05)
    {base \(X_1\), target \(j_a=z_0-z_1\)};

  \draw[-{Stealth[length=4pt]},black]
    (-1.20,0.35)--(0.05,0.35);

  \node[lab,above=1pt] at (-0.48,0.42) {$f_a$};
  \node[lab,below=2pt,align=center] at (-0.48,0.30)
    {$\hat x,\hat x_1\mapsto x$};
\end{scope}

%

\begin{scope}[shift={(8.6,0)}]
  \node[ttl] at (0,2.45) {2.\ the prism};

  \foreach \k in {0,...,7}{
    \pgfmathsetmacro{\A}{45*\k}
    \coordinate (p\k) at ({1.85*cos(\A)},{-0.95+0.52*sin(\A)});
    \coordinate (q\k) at ({1.45*cos(\A)},{ 0.95+0.40*sin(\A)});
  }

  \fill[blue!7]
    (p4)--(p5)--(p6)--(p7)--(p0)--
    (q0)--(q7)--(q6)--(q5)--(q4)--cycle;

  \foreach \k in {0,...,7}{
    \pgfmathsetmacro{\A}{45*\k}
    \ifdim \A pt > 0pt
      \ifdim \A pt < 180pt
        \draw[wh] (p\k)--(q\k);
      \else
        \draw[wl] (p\k)--(q\k);
      \fi
    \else
      \draw[wl] (p\k)--(q\k);
    \fi
  }

  \foreach \a/\b in {0/1,1/2,2/3,4/5,5/6,6/7}{
    \draw[wl] (p\a)--(q\b);
  }
  \draw[wh] (p3)--(q4);
  \draw[wl] (p0)--(q7);

  \draw[be] (q4)--(q5)--(q6)--(q7)--(q0);
  \draw[bh] (q0)--(q1)--(q2)--(q3)--(q4);

  \draw[be] (p4)--(p5)--(p6)--(p7)--(p0);
  \draw[bh] (p0)--(p1)--(p2)--(p3)--(p4);

  \foreach \k in {0,...,7}{
    \node[pv] at (q\k){};
    \node[bv] at (p\k){};
  }

  \node[lab,right=1pt] at (q0) {$i_1(K_a)$};
  \node[lab,right=1pt] at (p0) {$i_0(K_a)$};

  \node[lab,align=center] at (0,-2.05)
  {
    $\partial\mathsf{Prism}(c_a)
      =(i_1)_\#c_a-(i_0)_\#c_a$\\
    both rims are exposed
  };
\end{scope}


\begin{scope}[shift={(0,-6.1)}]
  \node[ttl] at (0,2.75) {3.\ cone on the top copy};

  \foreach \k in {0,...,7}{
    \pgfmathsetmacro{\A}{45*\k}
    \coordinate (p\k) at ({1.85*cos(\A)},{-1.05+0.52*sin(\A)});
    \coordinate (q\k) at ({1.45*cos(\A)},{ 0.75+0.40*sin(\A)});
  }

  \coordinate (ap) at (0,2.15);

  \foreach \a/\b in {0/1,1/2,2/3,3/4,4/5,5/6,6/7,7/0}{
    \fill[orange!13] (q\a)--(q\b)--(ap)--cycle;
  }

  \fill[blue!7]
    (p4)--(p5)--(p6)--(p7)--(p0)--
    (q0)--(q7)--(q6)--(q5)--(q4)--cycle;

  \foreach \k in {0,...,7}{
    \pgfmathsetmacro{\A}{45*\k}
    \ifdim \A pt > 0pt
      \ifdim \A pt < 180pt
        \draw[wh] (p\k)--(q\k);
        \draw[wh] (q\k)--(ap);
      \else
        \draw[wl] (p\k)--(q\k);
        \draw[wl] (q\k)--(ap);
      \fi
    \else
      \draw[wl] (p\k)--(q\k);
      \draw[wl] (q\k)--(ap);
    \fi
  }

  \foreach \a/\b in {0/1,1/2,2/3,4/5,5/6,6/7}{
    \draw[wl] (p\a)--(q\b);
  }
  \draw[wh] (p3)--(q4);
  \draw[wl] (p0)--(q7);

  \draw[pe] (q4)--(q5)--(q6)--(q7)--(q0);
  \draw[ph] (q0)--(q1)--(q2)--(q3)--(q4);

  \draw[be] (p4)--(p5)--(p6)--(p7)--(p0);
  \draw[bh] (p0)--(p1)--(p2)--(p3)--(p4);

  \foreach \k in {0,...,7}{
    \node[pv] at (q\k){};
    \node[bv] at (p\k){};
  }

  \node[pv] at (ap){};
  \node[slab,right=1pt] at (ap) {$v_a^*$};

  \node[lab,align=center] at (0,-2.15)
  {
    $\partial\Gamma_a^{\mathrm{raw}}
      =(i_0)_\#c_a$\\
    only the bottom rim is exposed
  };
\end{scope}


\begin{scope}[shift={(8.6,-6.1)}]
  \node[ttl] at (0,2.75) {4.\ fold the bottom rim onto the base};

  \foreach \k in {0,...,7}{
    \pgfmathsetmacro{\A}{45*\k+22.5}
    \coordinate (q\k) at ({1.55*cos(\A)},{0.80+0.45*sin(\A)});
  }

  \coordinate (ap) at (0,2.15);

  \coordinate (c)  at ( 0.00,-1.45);
  \coordinate (l1) at (-0.72,-1.12);
  \coordinate (l2) at (-1.55,-1.38);
  \coordinate (l3) at (-0.85,-1.78);
  \coordinate (r1) at ( 0.72,-1.78);
  \coordinate (r2) at ( 1.55,-1.38);
  \coordinate (r3) at ( 0.75,-1.12);

  \foreach \a/\b in {0/1,1/2,2/3,3/4,4/5,5/6,6/7,7/0}{
    \fill[orange!13] (q\a)--(q\b)--(ap)--cycle;
  }

  \foreach \k in {0,...,7}{
    \draw[wl] (q\k)--(ap);
  }

  \draw[fl] (q0)--(c);
  \draw[gl] (q1)--(l1);
  \draw[gl] (q2)--(l2);
  \draw[gl] (q3)--(l3);
  \draw[fl] (q4)--(c);
  \draw[gl] (q5)--(r1);
  \draw[gl] (q6)--(r2);
  \draw[gl] (q7)--(r3);

  \draw[gd] (c) --(q1);
  \draw[gd] (l1)--(q2);
  \draw[gd] (l2)--(q3);
  \draw[gd] (l3)--(q4);
  \draw[gd] (c) --(q5);
  \draw[gd] (r1)--(q6);
  \draw[gd] (r2)--(q7);
  \draw[gd] (c) --(q7);

  \draw[pe]
    (q0)--(q1)--(q2)--(q3)--(q4)--(q5)--(q6)--(q7)--cycle;

  \draw[re] (c)--(l1)--(l2)--(l3)--(c);
  \draw[re] (c)--(r1)--(r2)--(r3)--(c);

  \foreach \k in {0,...,7}{
    \node[pv] at (q\k){};
  }

  \node[pv] at (ap){};

  \foreach \p in {c,l1,l2,l3,r1,r2,r3}{
    \node[bv] at (\p){};
  }

  \node[slab,right=1pt]       at (ap) {$v_a^*$};
  \node[slab,above right=0pt] at (q0) {$\hat x$};
  \node[slab,above left=0pt]  at (q4) {$\hat x_1$};
  \node[lab,fill=white,inner sep=0.6pt] at (-0.30,-1.40) {$x$};

  \node[lab,align=center] at (0,-2.55)
  {
    $\partial\Gamma_a=(f_a)_\#c_a=j_a$\\
    $\hat x$ and $\hat x_1$ both land on $x$
  };
\end{scope}

\end{tikzpicture}%
}

\caption{An illustration of the single-qubit example in
\cite[Section~8.3.2]{KingKohler}, drawn in four stages of the
prism-cone-fold construction for the target line \(\CC(|0\rangle-|1\rangle)\).  In the first panel, the two
distinct source vertices \(\hat x\) and \(\hat x_1\) are both sent to the
shared bow-tie vertex \(x\).  The image of the oriented source cycle is
therefore \(z_0-z_1=j_a\).  The second panel joins a bottom and a top copy
of \(K_a\) by a simplicial prism.  In the third panel, the cone fills the
top copy, leaving the bottom copy as the only boundary.  The fourth panel
attaches that bottom boundary to \(X_1\) by \(f_a\).  Consequently, the
resulting chain satisfies \(\partial\Gamma_a=j_a\).  Black vertices belong
to the bow-tie base, blue vertices are private, orange triangles form the
cone, and red edges show the support of the target cycle.}
\label{fig:gadget-steps}
\end{figure}
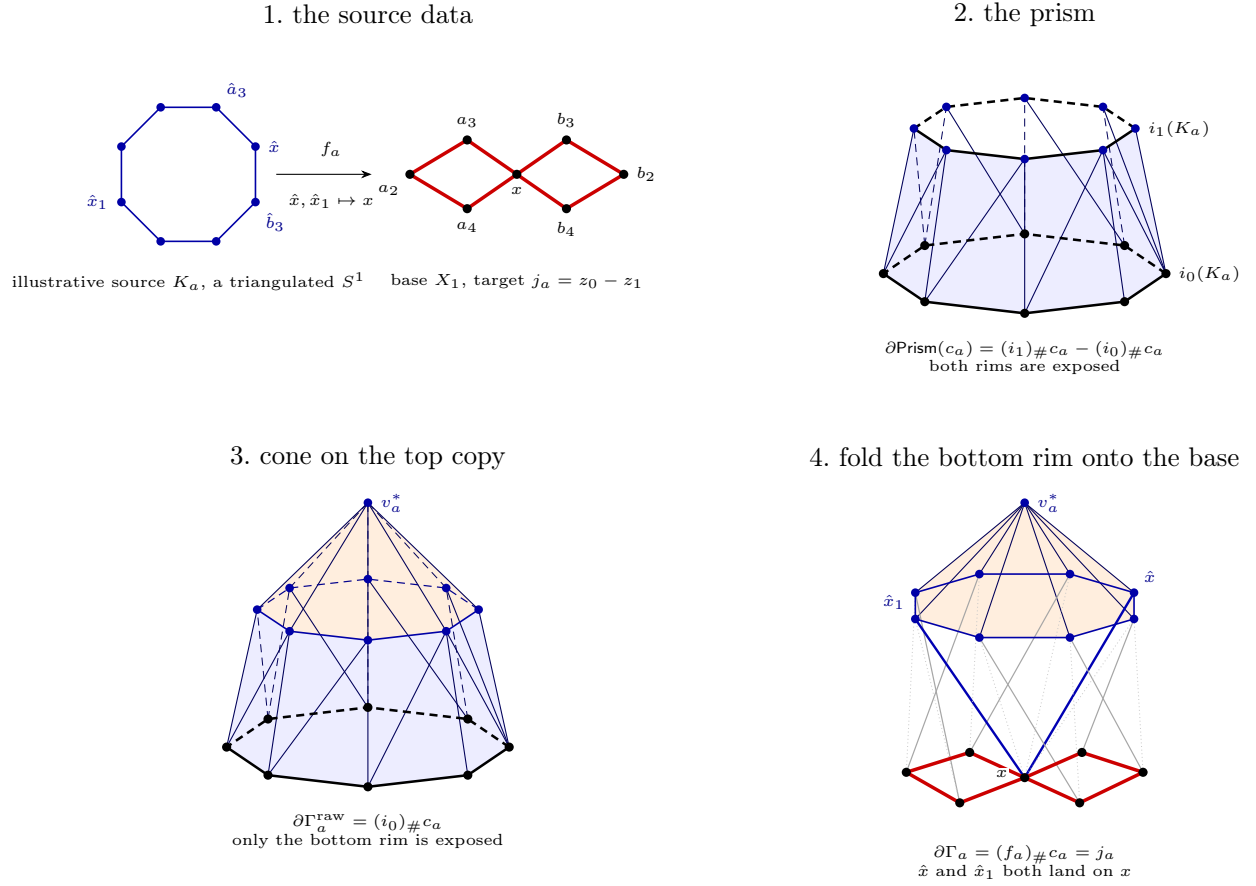

\subsection{The filling identity}
\label{subsec:filling-identity}

\begin{proposition}[Folded filling identity]
\label{prop:folded-filling}
For every catalogue index \(a\), the chain \(\Gamma_a\) defined in
item~\textup{(iv)} of the preceding construction satisfies
\[
  \Gamma_a
  \in C_{q_{m_a}+1}(Y_a^{\mathrm{loc}};\CC),
  \qquad
  \partial\Gamma_a=j_a.
\]
Consequently, \(\Gamma_a\) determines a relative homology class \([\Gamma_a] \in H_{q_{m_a}+1}(Y_a^{\mathrm{loc}},X_{m_a};\CC)\), and the connecting homomorphism satisfies
\[
  \delta[\Gamma_a]=[j_a]
  \in H_{q_{m_a}}(X_{m_a};\CC).
\]
\end{proposition}

\begin{proof}
Write \(c=c_a\). Since \(\partial c=0\), the cone and prism identities from
\cref{lem:cone-identity,lem:prism-identity} give
\[
  \partial\mathsf C_{v_a^*}\bigl((i_1)_\#c\bigr)
  =(i_1)_\#c,
  \qquad
  \partial\mathsf{Prism}(c)
  =(i_1)_\#c-(i_0)_\#c.
\]
With our orientation conventions, the cone face opposite the apex occurs
with sign \(+1\), while the interior prism faces cancel in pairs.
Therefore
\[
  \partial\Gamma_a^{\mathrm{raw}}=(i_0)_\#c.
\]

Every simplex occurring in \(\Gamma_a^{\mathrm{raw}}\) lies in \(U_a\). After applying \(q_a\), each nondegenerate image spans a clique in \(G_a^{\mathrm{loc}}\), and therefore a simplex of \(Y_a^{\mathrm{loc}}=\Cl(G_a^{\mathrm{loc}})\). In normalized simplicial chains, a simplex whose image has repeated vertices
contributes zero. Hence
\[
  \Gamma_a=(q_a)_\#\Gamma_a^{\mathrm{raw}}
  \in C_{q_{m_a}+1}(Y_a^{\mathrm{loc}};\CC).
\]
Since normalized simplicial chain maps commute with the boundary,
\[
  \partial\Gamma_a
  =(q_a)_\#\partial\Gamma_a^{\mathrm{raw}}
  =(q_ai_0)_\#c
  =(f_a)_\#c
  =j_a.
\]

In particular, \(j_a\) is a cycle contained in \(X_{m_a}\).  Since the
boundary of \(\Gamma_a\) lies entirely in the base, \(\Gamma_a\) is a
relative cycle for the pair \((Y_a^{\mathrm{loc}},X_{m_a})\).  The
connecting homomorphism sends the class of such a relative cycle to the
homology class of its boundary in the base, and therefore
\[
  \delta[\Gamma_a]
  =[\partial\Gamma_a]
  =[j_a].
\]

\end{proof}

The full kernel of the base-inclusion map is determined in
\cref{thm:local-topology}.

\section{Placement, padding, and the global graph}
\label{sec:global-graph}

Section~\ref{sec:local-gadget} constructed a gadget for a catalogue
projection \(\Pi_a\) acting on \(m_a\) abstract qubits, starting from the
base graph \(G_{m_a}=B_1*\cdots*B_{m_a}\) with clique complex \(X_{m_a}\). In \cref{constr:local-gadget}, the local graph is \(G_a^{\mathrm{loc}}=G_{m_a}\cup\Sk_1(\widetilde K_a)\).

Now let \(H=\sum_{j=1}^{t}P_j\) act on \(N\) qubits.  Each term \(P_j\) is a copy of a catalogue projection \(\Pi_{a_j}\) placed on an ordered
tuple \(S_j=(s_{j,1},\ldots,s_{j,m_{a_j}})\) of distinct qubits.  We first
relabel its local gadget to match \(S_j\), then join the gadget with the bow-tie factors outside \(S_j\), and finally combine the resulting graphs over one common \(N\)-qubit
base.  We call these steps \emph{placement}, \emph{padding}, and
\emph{assembly}, respectively.

\subsection{Placing a local gadget}
\label{subsec:placed-local-gadgets}

Fix a catalogue type \(a\) and an ordered tuple \(S=(s_1,\ldots,s_{m_a})\) of distinct elements of \(\{1,\ldots,N\}\). The order matters: the \(r\)-th input qubit of \(\Pi_a\) is placed on
qubit \(s_r\).  Recall from \cref{def:single-bowtie,def:bowtie-base} that
each \(B_r\) is a labelled copy of the seven-vertex graph \(B\).  Write \((r,u)\) for the vertex of \(B_r\) corresponding to \(u\in V(B)=\{1,\ldots,7\}\).  Private vertices are assigned labels distinct from every bow-tie vertex
label.

\begin{definition}[\(S\)-placed local gadget]
\label{def:placed-local-gadget}
Define a vertex map \(\rho_{a,S}\) on \(G_a^{\mathrm{loc}}\) by
\[
  \rho_{a,S}(r,u):=(s_r,u)
  \quad (1\leq r\leq m_a,\ u\in V(B)),
  \qquad
  \rho_{a,S}(p):=p
  \quad (p\text{ private}).
\]
The \emph{\(S\)-placed local gadget} \(G_a^{\mathrm{loc}}(S)\) has
vertex set \(\rho_{a,S}(V(G_a^{\mathrm{loc}}))\) and edge set
\[
  E\bigl(G_a^{\mathrm{loc}}(S)\bigr)
  :=\bigl\{\{\rho_{a,S}(u),\rho_{a,S}(v)\}:
       \{u,v\}\in E(G_a^{\mathrm{loc}})\bigr\}.
\]
Its clique complex is \(Y_a^{\mathrm{loc}}(S) :=\Cl\bigl(G_a^{\mathrm{loc}}(S)\bigr)\). Its bow-tie subgraph is \(B_{s_1}*\cdots*B_{s_{m_a}}\).
\end{definition}

For example, if \(m_a=2\) and \(S=(3,7)\), then \(\rho_{a,S}(1,7)=(3,7)\) and \(\rho_{a,S}(2,4)=(7,4)\).  The join edge \(\{(1,7),(2,4)\}\) becomes \(\{(3,7),(7,4)\}\); a private vertex
keeps its label.  For distinct gadget occurrences, we take disjoint copies of their private
vertex sets.

The map \(\rho_{a,S}\) is a bijection onto its image, and the edge
formula shows that placement gives graph and simplicial isomorphisms
\[
  G_a^{\mathrm{loc}}(S)\cong G_a^{\mathrm{loc}},
  \qquad
  Y_a^{\mathrm{loc}}(S)\cong Y_a^{\mathrm{loc}}.
\]

\subsection{Adding the qubits outside the support}
\label{subsec:padded-single-gadgets}

Let \(S^c\) be the complement of \(\{s_1,\ldots,s_{m_a}\}\) in \(\{1,\ldots,N\}\), and put \(G_{S^c}:=*_{r\in S^c}B_r\) for the bow-tie
graph on the untouched qubits.  If \(S^c=\varnothing\), the graph \(G_{S^c}\) has no vertices and its clique complex is \(\{\varnothing\}\).  Joining with this graph leaves the placed gadget
unchanged.

\begin{definition}[Padded gadget]
\label{def:padded-gadget}
The \emph{padded gadget} for \(\Pi_a\) on \(S\) is the graph
\[
  G_{a,S}^{\mathrm{pad}}
  :=G_a^{\mathrm{loc}}(S)*G_{S^c}.
\]
Its clique complex satisfies
\[
  Y_{a,S}^{\mathrm{pad}}
  :=\Cl\bigl(G_{a,S}^{\mathrm{pad}}\bigr)
  =Y_a^{\mathrm{loc}}(S)*\Cl(G_{S^c}).
\]
\end{definition}

The join connects every vertex of an untouched bow-tie factor to every
vertex of the placed gadget, without changing edges within either graph.
Its bow-tie part is the full base graph: \((B_{s_1}*\cdots*B_{s_{m_a}})*G_{S^c}=G_N\).

On the encoded qubit space, the corresponding placement of the
projection is
\[
  P_{a,S}^{(N)}
  :=W_S\bigl(\Pi_a\otimes I_{\Qenc_{N-m_a}}\bigr)W_S^*,
\]
where \(W_S\) is the placement unitary of
\cref{def:tensor-placement}.  Let \(\iota_{a,S}:X_N\hookrightarrow Y_{a,S}^{\mathrm{pad}}\) be the inclusion of the base.  Let \(\mathfrak e_N:H_{q_N}(X_N;\CC)\to\Qenc_N\) be the homological encoding
map of \cref{def:homological-encoding}.  Then
\cref{thm:padded-topology} proves that
\[
  \ker(\iota_{a,S})_*
  =\mathfrak e_N^{-1}\!\left(\operatorname{im}P_{a,S}^{(N)}\right).
\]
Thus precisely the classes in this preimage become boundaries.

\subsection{Combining the gadgets for all Hamiltonian terms}
\label{subsec:global-graph}

For each term \(P_j=(\Pi_{a_j})_{S_j}^{(N)}\), take a copy \(G_j^{\mathrm{pad}}\) of \(G_{a_j,S_j}^{\mathrm{pad}}\).  All these copies are identified on the common base graph \(G_N\), while
their private vertex sets are kept pairwise disjoint.  Write \(V_j^{\mathrm{priv}} :=V(G_j^{\mathrm{pad}})\setminus V(G_N)\). Repeated terms receive separate private vertex sets.

\begin{definition}[Global graph and output complex]
\label{def:global-graph}
\label{def:global-clique-complex}
Define the global graph by
\[
  V(\widehat G(H))
  :=V(G_N)\sqcup\bigsqcup_{j=1}^{t}V_j^{\mathrm{priv}},
  \qquad
  E(\widehat G(H))
  :=E(G_N)\cup\bigcup_{j=1}^{t}E(G_j^{\mathrm{pad}}).
\]
Its associated clique complex is \(\widehat Y(H):=\Cl\bigl(\widehat G(H)\bigr)\).
\end{definition}

In particular, no edge joins private vertices belonging to distinct
occurrences.  Since every pair of vertices in a simplex of a clique
complex must form an edge, no simplex of \(\widehat Y(H)\) contains
private vertices from two different occurrences.  This separation is
used in \cref{sec:topological-correctness,sec:global-spectrum} to
analyze the combined homological and spectral effects of the gadgets.

\section{Exact finite certificates}
\label{sec:certificates}

The fixed catalogue has \(s=13\) local gadget types, one for each catalogue
projection.  For each type we fix one local gadget before assigning it to
particular qubit positions; we call this its \emph{canonical local gadget}.
This section states the exact finite properties verified for these thirteen
canonical gadgets.  These properties are the only machine-certified inputs used in the
remainder of the proof.

\subsection{Reduction to the finite catalogue}
\label{subsec:finite-certification-suffices}

Every source term is a placed copy of one of the thirteen catalogue
projections, and every locality satisfies \(m_a\leq4\).  Changing the source
size does not create a new local gadget: it only relabels its supported
bow-tie factors, joins it with untouched factors, or repeats it over the
shared base.

\begin{proposition}[Scope of the finite certificate list]
\label{prop:certificate-scope}
Let \(H\) have a finite-catalogue source presentation as in
\cref{def:source-family}, and let \(\widehat Y(H)\) be the output complex
of \cref{def:global-graph}.  Assume that, for each catalogue label \(a\), the induced subgraph of \(G_a^{\mathrm{loc}}\) on the bow-tie vertices \(V(G_{m_a})\) is exactly \(G_{m_a}\).  Then every gadget occurrence in \(\widehat Y(H)\) is obtained from its canonical local gadget by:
\begin{enumerate}[label=(\roman*)]
\item a weight- and depth-preserving relabelling of qubit factors;
\item joining the placed gadget with the bow-tie factors on qubits outside
its support; and
\item identifying its full base graph with the one shared copy of \(G_N\) in the global graph, while keeping its private vertices separate.
\end{enumerate}
Therefore no new finite certificate is required when \(N\), \(t\), the
supports \(S_j\), or the source circuit changes.
\end{proposition}

\begin{proof}
Placement and padding are defined in
\cref{def:placed-local-gadget,def:padded-gadget}.  The hypothesis implies
that each padded graph has exactly \(G_N\) as its induced graph on the
base vertices.  In the global edge-set formula of
\cref{def:global-graph}, an occurrence other than \(j\) can therefore
add neither an edge between base vertices nor an edge incident to a
private vertex of occurrence \(j\).  Consequently,
\[
  \widehat G(H)[V(G_N)\cup V_j^{\mathrm{priv}}]
  =G_j^{\mathrm{pad}}.
\]
Here \(G[W]\) denotes the graph induced by the vertex set \(W\). Taking clique complexes shows that the simplices involving the private
vertices of occurrence \(j\), together with their shared base, are
exactly those of its padded gadget.  The later padding and many-gadget
theorems handle the effects of these operations; no instance-dependent
finite certificate is needed.
\end{proof}

\subsection{The finite objects being certified}
\label{subsec:finite-certificate-objects}

Fix a catalogue index \(a\), put \(m=m_a\), and consider the pair \(X_m\subseteq Y_a^{\mathrm{loc}}\) from \cref{constr:local-gadget}.
The private vertex set is \(V_a^{\mathrm{priv}}=V(Y_a^{\mathrm{loc}})\setminus V(X_m)\), as defined
there.

\paragraph{Source-data representation.}
The certificate package stores an integer-coefficient cycle \(\zeta_a\) and a nonzero integer \(\mu_a\) with \(c_a=\zeta_a/\mu_a\). For some gadget types, the initially chosen simplicial map produces the
desired underlying graph but the wrong orientation on one or more encoded
bow-tie cycles.  In those cases we compose it with explicit automorphisms
of the bow-tie base that reverse the required cycle orientations.
The certificate then verifies \((f_a)_\#c_a=j_a\) coefficient by
coefficient.

\begin{definition}[Private depth and split coboundary]
\label{def:base-private-weights}
\label{def:exact-depth-cochains}
Give every base vertex weight \(1\) and every private vertex weight \(0<\lambda<1\).  For a simplex \(\sigma\), let its \emph{private
depth} be the number of its private vertices: \(\dep(\sigma):= |\sigma\cap V_a^{\mathrm{priv}}|\). For every degree \(k\), define the exact-depth space
\[
  C^k_{(\ell)}(Y_a^{\mathrm{loc}})
  :=
  \operatorname{span}_{\CC}\bigl\{
    |\sigma\rangle:
    \sigma\in(Y_a^{\mathrm{loc}})_k,
    \ \dep(\sigma)=\ell
  \bigr\}.
\]
Here \(|\sigma\rangle\) is the simplex-basis cochain from
\cref{subsec:cochains-cohomology}.  The same span is obtained in the weighted
orthonormal basis.
In the weighted orthonormal simplex coordinates,
\[
  C^k(Y_a^{\mathrm{loc}})
  =\bigoplus_{\ell=0}^{k+1}C^k_{(\ell)},
  \qquad
  D_\lambda=\Dbase+\lambda\Dpriv,
\]
where
\[
  \Dbase:C^k_{(\ell)}\to C^{k+1}_{(\ell)},
  \qquad
  \Dpriv:C^k_{(\ell)}\to C^{k+1}_{(\ell+1)}.
\]
The operators \(\Dbase\) and \(\Dpriv\) add a base and a private vertex,
respectively, with their oriented incidence signs.
\end{definition}

\begin{definition}[Private simplices and their attachment to the base]
\label{def:all-private-pair}
Define the induced all-private complex, using the notation \(K[W]\) of
\cref{def:simplicial-complex}, by \(P_a:=Y_a^{\mathrm{loc}}[V_a^{\mathrm{priv}}]\). Thus \(P_a\) consists exactly of the simplices whose vertices are all
private.  Define
\[
  A_a
  :=
  \bigl\{\tau\in P_a:
    \text{some }b\in V(X_m)\text{ satisfies }
    \tau\cup\{b\}\in Y_a^{\mathrm{loc}}
  \bigr\}.
\]
Thus \(A_a\) consists of the private simplices that remain simplices after
adjoining at least one base vertex.  It is a subcomplex:
whenever \(\tau\) admits such an extension, so does every face of \(\tau\).  By the relative-cochain definition in
\cref{def:relative-cochains},
\[
  C^k(P_a,A_a;\CC)
  =
  \operatorname{span}_{\CC}\bigl\{
    |\tau\rangle:
    \tau\in(P_a)_k\setminus(A_a)_k
  \bigr\}.
\]
\end{definition}

\begin{definition}[Private core]
\label{def:private-core-class}
When \(H^{2m}(P_a,A_a;\CC)\) is one dimensional, let \(\kappa_a\in C^{2m}(P_a,A_a;\CC)\) be the unit coordinate cochain supported on the certified private \(2m\)-simplex.  The certificate checks that it is closed in the relative cochain complex
and that its class generates this cohomology group.  Write \(\CC[\kappa_a] := \{c[\kappa_a]:c\in\CC\} \subseteq H^{2m}(P_a,A_a;\CC)\) for the line spanned by its cohomology class.
\end{definition}

\begin{definition}[Private-sector incidence matrix]
\label{def:principal-bulk-matrix}
Put \(d=q_m=2m-1\), and let \(v_a^*\) be the apex.  Define \(\Bbulk_a^{\mathrm{loc}}:=C^d(P_a,A_a;\CC)\), \[
  \Odown_a
  :=
  \operatorname{span}_{\CC}\bigl\{
    |\rho\rangle:
    \rho\in(Y_a^{\mathrm{loc}})_{d-1},\ v_a^*\in\rho
  \bigr\},
\]
and
\[
  \Oup_a
  :=
  \operatorname{span}_{\CC}\bigl\{
    |\omega\rangle:
    \omega\in(Y_a^{\mathrm{loc}})_{d+1},\ v_a^*\in\omega
  \bigr\}.
\]
Let \(Q_a^\downarrow\) and \(Q_a^\uparrow\) be the corresponding coordinate
projections.  The \emph{private-sector incidence matrix} is
\[
  T_{a,0}:\Bbulk_a^{\mathrm{loc}}
  \longrightarrow\Odown_a\oplus\Oup_a,
  \qquad
  T_{a,0}u
  :=
  \bigl(Q_a^\downarrow\partial_du,
        Q_a^\uparrow D_du\bigr).
\]
Here \(D_d\) and \(\partial_d\) are the unit-weight incidence operators
in the ordinary oriented-simplex bases, with relative cochains embedded
by their simplex coordinates.  Thus \(T_{a,0}\) is an integer incidence
matrix independent of \(\lambda\).
\end{definition}

\subsection{The five local conditions}
\label{subsec:conditions-C0-C4}

C0 determines the relative topology of the completed local complex; C1
verifies the oriented filling; C2 checks the induced bow-tie base and
private isolation; C3 eliminates intermediate filtered classes; and C4
identifies the unique surviving private class and controls the remaining private
sector.

\begin{requirement}[C0: relative topology of the actual clique complex]
\label{req:C0}
For every \(a\), \[
  H_r(Y_a^{\mathrm{loc}},X_{m_a};\CC)
  \cong
  \begin{cases}
    \CC,&r=2m_a,\\
    0,&r\neq2m_a.
  \end{cases}
\]
This is the relative homology of the actual pair \(X_{m_a}\subseteq\Cl(G_a^{\mathrm{loc}})\), including any simplices formed when all cliques of the graph are included
in its clique complex, as in \cref{def:clique-complex}.
\end{requirement}

\begin{requirement}[C1: signed source and folded boundary]
\label{req:C1}
For every \(a\), \[
  \partial c_a=0,
  \qquad
  (f_a)_\#c_a=j_a,
  \qquad
  \partial\Gamma_a=j_a,
\]
as exact equalities of oriented chains over \(\QQ\), with \((f_a)_\#\) interpreted as the normalized induced chain map of
\cref{def:normalized-chain-map} (including orientation signs and the
zero image of a simplex made degenerate by the map).
The target cycle \(j_a\) is integral, while \(c_a\) and \(\Gamma_a\) may
have rational coefficients.
\end{requirement}

\begin{requirement}[C2: induced base and private isolation]
\label{req:C2}
For every unplaced local gadget of \cref{constr:local-gadget},
\[
  Y_a^{\mathrm{loc}}[V(X_{m_a})]=X_{m_a}.
\]
In the global graph,
\[
  \{u,v\}\notin E(\widehat G(H))
  \qquad
  \bigl(
    u\in V_i^{\mathrm{priv}},\ 
    v\in V_j^{\mathrm{priv}},\ 
    i\neq j
  \bigr).
\]
Thus forming the clique complex creates neither additional base simplices
locally nor simplices containing private vertices from different gadget
occurrences.
The induced-base equality is part of the finite local certificate.  The
cross-gadget statement follows directly from the definition of the global
edge set and is not recomputed for each instance.  The induced-base
equality supplies the hypothesis of \cref{prop:certificate-scope}.
\end{requirement}

The decreasing filtration associated with private depth is
\[
  U_\ell^k
  :=\bigoplus_{r\geq\ell}C^k_{(r)}.
\]
It is preserved by the unweighted incidence differential \(D=\Dbase+\Dpriv\).  The spectral-sequence conventions for this
filtration are fixed in Appendix~\ref{app:spectral-sequence}; the
fixed-depth exactness requirement is stated next.

\begin{requirement}[C3: fixed-depth exactness]
\label{req:C3}
For every \(a\), every degree \(k\), and every intermediate depth \(1\leq\ell\leq k\), \[
  \ker\!\left(
    \Dbase:C^k_{(\ell)}\to C^{k+1}_{(\ell)}
  \right)
  =
  \operatorname{im}\!\left(
    \Dbase:C^{k-1}_{(\ell)}\to C^k_{(\ell)}
  \right).
\]
Thus, at these depths, every \(\Dbase\)-closed cochain is the \(\Dbase\)-image of a cochain one degree lower.  The corresponding
spectral-sequence interpretation is given in \cref{thm:page-one}.
\end{requirement}

\begin{requirement}[C4: private core and exclusion of additional private low-energy modes]
\label{req:C4}
For every gadget of locality \(m=m_a\), \[
  H^r(P_a,A_a;\CC)
  \cong
  \begin{cases}
    \CC,&r=2m,\\
    0,&r\neq2m,
  \end{cases}
\]
and \(T_{a,0}\) has full column rank.  Consequently,
\[
  a_*:=\min_{1\leq a\leq s}s_{\min}(T_{a,0})>0.
\]
The certificates also verify the following three incidence properties of
the cone apex \(v_a^*\), which are used later:
\begin{enumerate}[label=(\roman*)]
\item
\[
  (P_a)_{q_m}\setminus(A_a)_{q_m}
  =\{\sigma\in(P_a)_{q_m}:v_a^*\in\sigma\};
\]
\item in the local graph \(G_a^{\mathrm{loc}}\), the apex has no edge to
any vertex of the local bow-tie graph \(G_m\);
\item if \(v_a^*\in\sigma\in(P_a)_{q_m}\), then \(\sigma\setminus\{v_a^*\}\in A_a\).
\end{enumerate}
\end{requirement}

\subsection{Certification and relabelling}
\label{subsec:finite-certification-theorem}

\begin{theorem}[Finite local certification theorem]
\label{thm:finite-local-certification}
The thirteen canonical local gadgets satisfy C0-C4.
\end{theorem}

\begin{proof}[Certificate record]
The companion certificate package reconstructs every canonical graph from
its source data, including the ordered prism, cone, folding map, and the
explicit bow-tie automorphisms used to correct the required orientations,
and performs all tests on the resulting clique complex.  It
checks C1 by exact rational chain arithmetic, including all orientation signs
and every simplex made degenerate by the fold.  It checks C0, C3, and the
relative-cohomology part of C4 using exact sparse integer incidence matrices.
Some ranks are computed exactly modulo a prime; the exactness criterion
proved in Appendix~\ref{app:certificate-protocols} shows that the rank and
exactness conclusions used here then hold over \(\QQ\), and hence over \(\CC\).

The induced-base part of C2 and the three apex-incidence properties are
checked by exact graph comparisons and exact comparisons of the enumerated
simplex lists; cross-gadget isolation follows from the global edge-set
definition.  Finally, an exact rank computation over a finite field,
together with the transfer lemma, certifies that every \(T_{a,0}\) has full
column rank, hence is injective over \(\QQ\) and over \(\CC\).  The
complete finite data and verification procedures are included in the
companion package.
Every pass/fail conclusion uses exact arithmetic.
\end{proof}

\begin{proposition}[Relabelling invariance]
\label{prop:certificate-relabel-invariance}
A permutation of the local qubit factors induces a weight- and
depth-preserving simplicial isomorphism of the corresponding local gadgets.
It preserves the pair \((Y_a^{\mathrm{loc}},X_{m_a})\), conditions C0-C4,
the filtered cochain complex, the weighted Laplacian spectrum, and the
singular values of \(T_{a,0}\).  Hence one canonical certificate applies to
every placed occurrence of that catalogue type.
\end{proposition}

\begin{proof}
Relabel the entire gadget datum, including the base, private vertices, apex,
source order, and folding map.  In every degree the induced chain map is a
signed permutation matrix \(U_k\), and boundary naturality gives
\[
  B_k'U_k=U_{k-1}B_k.
\]
Because weights and private depths are preserved, these maps become unitary
after rescaling the simplex basis according to the vertex weights, and they
preserve the filtration, relative pairs, and filtered cohomology.  On the
private-sector coordinates entering \(T_{a,0}\), \[
  T_{a,0}'=U_{\mathrm{out}}T_{a,0}U_{\mathrm{in}}^*,
\]
so the singular values and full column rank are unchanged.  The signed C1
identities are transported with the source chains and folding maps.  Any
additional catalogue relabelling is included in the exact C1 certificate.
\end{proof}

Section~\ref{sec:topological-correctness} uses C0-C2 to determine the exact
homological effect of the attachments.
Section~\ref{sec:single-gadget-spectrum} uses C3-C4 to obtain the spectral
penalty produced by one gadget.  The subsequent global spectral argument then
combines polynomially many overlapping gadget occurrences.

\section{The topological effect of the attachments}
\label{sec:topological-correctness}

Sections~\ref{sec:local-gadget} and~\ref{sec:global-graph} constructed the
local and global clique complexes.  For an inclusion \(\iota:X\hookrightarrow Y\), we say that \([j]\in H_d(X;\CC)\) is
\emph{killed} when \(\iota_*[j]=0\).

\subsection{The effect of one local attachment}

The identity \(\partial\Gamma_a=j_a\) gives \([j_a]=0\) after the local
attachment.  The relative homology of the pair determines the full kernel.

\begin{lemma}[Relative gluing lemma]
\label{lem:relative-gluing}
Let $X\subseteq Y$ be finite simplicial complexes, let $d\ge0$, and suppose
\[
 H_r(Y,X;\CC)\cong
 \begin{cases}
  \CC,&r=d+1,\\
  0,&r\ne d+1.
 \end{cases}
\]
Let $[\Gamma]$ generate $H_{d+1}(Y,X;\CC)$, and suppose that the connecting homomorphism in the long exact sequence of the pair, recalled in
Theorem~\ref{thm:long-exact-sequence-pair}, satisfies
\[
 \delta[\Gamma]=[j]\ne0\in H_d(X;\CC).
\]
Then, for the inclusion $\iota:X\hookrightarrow Y$, \(\iota_*:H_r(X;\CC)\longrightarrow H_r(Y;\CC)\) is an isomorphism for $r\ne d$, while in degree $d$, \(\iota_*:H_d(X;\CC)\longrightarrow H_d(Y;\CC)\) is surjective with
\[
 \ker\iota_*=\CC\langle[j]\rangle.
\]
Equivalently,
\[
 H_d(Y;\CC)\cong H_d(X;\CC)/\CC\langle[j]\rangle.
\]
\end{lemma}

\begin{proof}
The long exact sequence of the pair
(Theorem~\ref{thm:long-exact-sequence-pair}) contains
\[
 0\longrightarrow H_{d+1}(X)
 \longrightarrow H_{d+1}(Y)
 \longrightarrow\CC[\Gamma]
 \xrightarrow{\delta}H_d(X)
 \xrightarrow{\iota_*}H_d(Y)
 \longrightarrow0.
\]

Since \(\delta[\Gamma]=[j]\ne0\), the map \(\delta\) from its
one-dimensional domain is injective.  Exactness therefore implies that \(H_{d+1}(X)\to H_{d+1}(Y)\) is an isomorphism, and that \(\iota_*\) is surjective with
\[
  \ker\iota_*=\CC[j].
\]
For every remaining degree \(r\ne d,d+1\), the hypothesis gives \(H_{r+1}(Y,X;\CC)=0 \,\text{and}\, H_r(Y,X;\CC)=0\). Hence the corresponding part of the long exact sequence is
\[
  0\longrightarrow H_r(X;\CC)
  \longrightarrow H_r(Y;\CC)
  \longrightarrow0,
\]
so the inclusion induces an isomorphism \(H_r(X;\CC)\cong H_r(Y;\CC)\).

The quotient description follows from the first isomorphism theorem.
\end{proof}
\begin{theorem}[One certified gadget kills exactly its target line]
\label{thm:local-topology}
For every catalogue index $a$, put $m=m_a$ and $d=q_m$.  The inclusion \(\iota_a:X_m\hookrightarrow Y_a^{\mathrm{loc}}\) induces isomorphisms in every degree other than $d$, and induces a
surjection in degree $d$ with
\[
 \ker(\iota_a)_*=\CC\langle[j_a]\rangle.
\]
Under the encoding $\mathfrak e_m$, this kernel is precisely
$\CC|\phi^{(a)}\rangle$.  Equivalently,
\[
 \dim H_d(Y_a^{\mathrm{loc}};\CC)=\dim H_d(X_m;\CC)-1.
\]
\end{theorem}

\begin{proof}
By C2, \(X_m\) is the induced base subcomplex of \(Y_a^{\mathrm{loc}}\), and C0 supplies the relative-homology hypothesis of
\cref{lem:relative-gluing}.

By \cref{prop:folded-filling}, \(\Gamma_a\) is a relative cycle whose image under the connecting
homomorphism is \([j_a]\).  The image is nonzero because the base classes \([z_b]\) form a basis and
the target vector defining \(j_a\) is nonzero.  Hence \([\Gamma_a]\neq0\), and the one-dimensional relative group supplied by C0 is generated by \([\Gamma_a]\).

The gluing lemma therefore gives kernel \(\CC[j_a]\) and isomorphisms
in every other degree.

Finally, \(\mathfrak e_m([j_a]) = \mathcal Z_a|\phi^{(a)}\rangle, \, \mathcal Z_a>0\). Hence the homology line \(\CC[j_a]\) corresponds exactly to the target line \(\CC|\phi^{(a)}\rangle\).
\end{proof}

\subsection{The effect of padding by untouched qubits}

\begin{theorem}[Join-padding structure]
\label{thm:join-padding}
For a catalogue type $a$, an ordered support $S$ of size $m_a$, and
$N\ge m_a$, the graph and flag-complex equalities
\[
 G_{a,S}^{\mathrm{pad}}
 =G_a^{\mathrm{loc}}(S)*G_{S^c},
 \qquad
 Y_{a,S}^{\mathrm{pad}}
 \cong Y_a^{\mathrm{loc}}*X_{N-m_a}
\]
hold exactly.
\end{theorem}

\begin{proof}
Relabelling changes only vertex names and hence gives a simplicial
isomorphism.  By definition of the graph join, every vertex of the placed
local gadget is adjacent to every vertex of the untouched graph.  Therefore
a set of vertices is a clique in \(G_{a,S}^{\mathrm{pad}}\) exactly when its
intersection with each factor is a clique.  The clique-complex identity
follows from \cref{lem:clique-complex-join}.

When \(N=m_a\), the untouched factor is \(X_0=\{\varnothing\}\), so the join leaves the local complex unchanged.
\end{proof}

\begin{theorem}[Topological action of one padded projector]
\label{thm:padded-topology}
Let $\iota_{a,S}:X_N\hookrightarrow Y_{a,S}^{\mathrm{pad}}$ be the base
inclusion.  In encoded degree $q_N=2N-1$,
\[
 \ker(\iota_{a,S})_*
 =\mathfrak e_N^{-1}(\operatorname{im} P_{a,S}^{(N)})
 \subseteq H_{q_N}(X_N;\CC),
\]
and
\[
 H_{q_N}(Y_{a,S}^{\mathrm{pad}};\CC)
 \cong
 \Qenc_N/\operatorname{im} P_{a,S}^{(N)}.
\]
The quotient is identified with $\ker P_{a,S}^{(N)}$ by taking the
orthogonal complement of $\operatorname{im}P_{a,S}^{(N)}$ in the encoded
qubit space.  This identifies the surviving homology classes.  Their
weighted harmonic representatives are studied in
Section~\ref{sec:single-gadget-spectrum}.
\end{theorem}

\begin{proof}
Put \(m=m_a\) and \(n=N-m\).  Under \cref{thm:join-padding}, the base inclusion
is the join of \(i_a:X_m\hookrightarrow Y_a^{\mathrm{loc}}\) with the identity of \(X_n\).  The reduced homology of \(X_m\) and \(Y_a^{\mathrm{loc}}\) is concentrated
in degree \(q_m\), while that of \(X_n\) is concentrated in degree \(q_n\).

The join K\"unneth formula of \cref{thm:join-kunneth} reduces here to
the tensor product of the homology groups in degrees \(q_m\) and \(q_n\), because the factor homology vanishes in all other degrees.  Under this
identification, the join map induces the tensor product of the maps on
homology.  Therefore, in degree \(q_N=q_m+q_n+1\), the base inclusion induces
\[
 (i_a)_*\otimes I:
 H_{q_m}(X_m)\otimes\widetilde H_{q_n}(X_n)
 \longrightarrow
 H_{q_m}(Y_a^{\mathrm{loc}})\otimes\widetilde H_{q_n}(X_n).
\]
For \(n=0\), the second factor is interpreted as \(\widetilde H_{-1}(X_0)=\CC\).

By \cref{thm:local-topology}, the first map is surjective with kernel \(\CC[j_a]\).  Tensoring over the field \(\CC\) preserves this exact
sequence.  Hence the displayed map is surjective with kernel \(\CC[j_a]\otimes\widetilde H_{q_n}(X_n)\).

The explicit encoding sends this kernel, after the factor placement \(W_S\) of \cref{def:tensor-placement}, to
\[
 W_S(\CC\phi^{(a)}\otimes\Qenc_n)
 =\operatorname{im}\bigl(W_S(\Pi_a\otimes I)W_S^*\bigr)
 =\operatorname{im} P_{a,S}^{(N)}.
\]
Thus the padded attachment kills exactly the image of the placed projector \(P_{a,S}^{(N)}\).

Each bow-tie cycle has degree one, so swapping two such join factors has sign \((-1)^{(1+1)(1+1)}=1\). Thus their tensor permutation introduces no extra sign in \([z_b]\).

The first isomorphism theorem proves the quotient statement.  The map \([x]\longmapsto(I-P_{a,S}^{(N)})x\) is its well-defined isomorphism with \(\ker P_{a,S}^{(N)}\), since its kernel before quotienting is exactly \(\operatorname{im}P_{a,S}^{(N)}\).
\end{proof}

\subsection{All attachments and the YES case}

The decisive point is that a nonbase simplex belongs to exactly one gadget
occurrence: it cannot contain private vertices from two different
occurrences.  Consequently the relative chain complex of the global pair \((\widehat Y(H),X_N)\) is the direct sum of the relative chain complexes of
the individual occurrences.  The connecting maps for those occurrences
therefore add, which is the mechanism behind the following theorem.

\begin{theorem}[Global topological kernel]
\label{thm:global-topology}
Let \(H=\sum_{j=1}^tP_j, \, P_j:=(\Pi_{a_j})_{S_j}^{(N)}\), and let $\widehat\iota:X_N\hookrightarrow\widehat Y(H)$ be the base
inclusion.  Then
\[
 \ker(\widehat\iota)_*
 =\mathfrak e_N^{-1}\!\left(\sum_{j=1}^t\operatorname{im} P_j\right)
 \subseteq H_{q_N}(X_N;\CC).
\]
Consequently the image of the base homology is naturally isomorphic to \(\Qenc_N\Big/\sum_{j=1}^t\operatorname{im} P_j\). Since the \(P_j\) are orthogonal projections, this quotient is naturally
identified with
\[
  \left(\sum_{j=1}^t\operatorname{im}P_j\right)^\perp
  =\bigcap_{j=1}^t\ker P_j
  =\ker H.
\]
Two vectors $\psi,\varphi\in\Qenc_N$ represent the same coset in this
quotient precisely when \(\psi-\varphi\in\sum_{j=1}^t\operatorname{im} P_j\).
\end{theorem}

\begin{proof}
Write $X=X_N$, $Y=\widehat Y(H)$, $q=q_N$, and let $Y_j$ be the actual
padded clique complex of occurrence $j$, viewed inside $Y$.
We first verify the relevant simplex decomposition from the graph itself.
By condition C2, a clique contains private vertices from at most one
occurrence.
A clique containing none is a simplex of $X$, since every attachment has
the same induced base graph.  If it contains a private vertex of occurrence
$j$, all its edges incident to private vertices come from that occurrence;
its remaining edges are base edges.  It is therefore a simplex of $Y_j$.
Conversely, every simplex of $Y_j$ is a simplex of $Y$, by the definition
of the global edge set.  In particular, the inclusions $Y_j\hookrightarrow Y$
are simplicial, $Y_i\cap Y_j=X$ for $i\ne j$, and each nonbase simplex has
a unique occurrence label.

The nonbase oriented simplices are a basis of relative chains.  The preceding
partition gives, in every degree, an isomorphism
\[
 \Phi:\bigoplus_{j=1}^t C_\bullet(Y_j,X;\CC)
 \xrightarrow{\ \cong\ }C_\bullet(Y,X;\CC).
\]
This is a chain isomorphism: deleting a vertex either retains the same
private occurrence label or produces a base face, which is zero in the
relative quotient.  Thus no relative boundary changes occurrence label.
Taking homology gives the corresponding direct sum in every degree.

Let $\delta_j:H_{q+1}(Y_j,X)\to H_q(X)$ and
$\widehat\delta:H_{q+1}(Y,X)\to H_q(X)$ be the connecting maps.
For relative cycles represented by chains $c_j$ with $\partial c_j\in C_q(X)$,
the global connecting map is explicitly
\[
 \widehat\delta\,\Phi_*(([c_j])_j)
 =\left[\partial\sum_jc_j\right]
 =\sum_j[\partial c_j]
 =\sum_j\delta_j[c_j].
\]
Consequently $\operatorname{im}\widehat\delta=\sum_j\operatorname{im}\delta_j$.
Exactness of the long exact sequence of each pair, followed by
\cref{thm:padded-topology}, gives
\[
 \operatorname{im}\delta_j=\ker(H_q(X)\to H_q(Y_j))
 =\mathfrak e_N^{-1}(\operatorname{im} P_j).
\]
Exactness for $(Y,X)$ now yields
\[
 \ker\widehat\iota_*
 =\operatorname{im}\widehat\delta
 =\mathfrak e_N^{-1}\!\left(\sum_j\operatorname{im} P_j\right).
\]
The first isomorphism theorem identifies the image of $\widehat\iota_*$
with the stated quotient.  Finally, because every \(P_j\) is an orthogonal projection,
\[
  \ker H
  =\ker\!\left(\sum_jP_j\right)
  =\bigcap_j\ker P_j
  =\left(\sum_j\operatorname{im}P_j\right)^\perp.
\]
Orthogonal projection onto this complement gives the claimed identification
of the quotient with \(\ker H\).
\end{proof}

\begin{corollary}[Exact YES implication]
\label{cor:YES}
If $0\ne|\psi\rangle\in\Qenc_N$ satisfies \(P_j|\psi\rangle=0\,(j=1,\ldots,t)\), then
\[
 \widehat\iota_*\bigl(\mathfrak e_N^{-1}(|\psi\rangle)\bigr)
 \ne0\in H_{q_N}(\widehat Y(H);\CC).
\]
Hence
\[
 H_{q_N}(\widehat Y(H);\CC)\ne0,
 \qquad
 \lambda_{\min}(\widehat\Delta_{q_N})=0.
\]
\end{corollary}

\begin{proof}
Each $P_j$ is an orthogonal projection, so
$\ker P_j=(\operatorname{im} P_j)^\perp$.  The hypotheses imply
\[
 \psi\in\bigcap_j\ker P_j
 =\left(\sum_j\operatorname{im} P_j\right)^\perp.
\]
Since $\psi\ne0$ and
$\psi\perp\sum_j\operatorname{im}P_j$, its coset in the quotient of
\cref{thm:global-topology} is nonzero, which gives the claimed nonzero
homology class.

For positive vertex weights, the weighted inner product is nondegenerate
and finite-dimensional Hodge theory identifies the weighted Laplacian
kernel with cohomology.  Over $\CC$, homology and cohomology have the same
dimension.  Hence the nonzero homology just proved gives a nonzero weighted
harmonic vector.  Positivity of the Hodge Laplacian implies that its
smallest eigenvalue is exactly zero.  This holds for every positive private weight.
\end{proof}

\section{The weighted spectrum of one gadget}
\label{sec:single-gadget-spectrum}

Section~\ref{sec:topological-correctness} determined the exact zero-energy
information: attaching the local gadget removes the line \(\CC|\phi^{(a)}\rangle\) from homology.  With private-vertex weight \(0<\lambda<1\), the local coboundary has the form
\[
  D_\lambda=\Dbase+\lambda\Dpriv.
\]
Here \(\Dbase\) preserves private depth and \(\Dpriv\) raises it by one.
The private-depth filtration therefore records the first power of \(\lambda\) at which an encoded cochain meets an obstruction.  An obstruction of order \(\lambda^r\) produces Laplacian energy of order \(\lambda^{2r}\).

This filtration determines a spectral sequence, whose successive pages
record which cohomology classes survive after obstructions of increasing
order have been taken into account.  We write \(E_{r,\ell}^k\) for the
page-\(r\) term in cochain degree \(k\) and private depth \(\ell\).  Its
differential has degree and depth shift
\[
  \delta_r:E_{r,\ell}^k\longrightarrow E_{r,\ell+r}^{k+1}.
\]
The page quotients and the representative formula for \(\delta_r\) are defined in Appendix~\ref{app:spectral-sequence}.

\subsection{The first two pages}

\begin{theorem}[Page-one structure]
\label{thm:page-one}
For every catalogue gadget, page zero is the exact-depth quotient
\[
 E_{0,\ell}^k=U_\ell^k/U_{\ell+1}^k\cong C^k_{(\ell)},
\]
with differential induced by $\Dbase$.  Therefore
\[
 E_{1,\ell}^k
 =\ker(\Dbase:C^k_{(\ell)}\to C^{k+1}_{(\ell)})
 \Big/
 \operatorname{im}(\Dbase:C^{k-1}_{(\ell)}\to C^k_{(\ell)}).
\]
Condition C3 gives
\[
 E_{1,\ell}^k=0\qquad(1\le\ell\le k).
\]
\end{theorem}

\begin{proof}
Taking the exact-depth component induces the isomorphism
$U_\ell^k/U_{\ell+1}^k\to C^k_{(\ell)}$.
The component $\Dpriv$ raises depth, so it is zero in the target quotient;
$\Dbase$ preserves depth and induces $\delta_0$.
Taking its cohomology gives the asserted formula for $E_{1,\ell}^k$.
The intermediate-depth vanishing is precisely condition C3 applied to those
kernel and image spaces.  At depth zero, condition C2 identifies the resulting complex
with the augmented cochain complex of $X_m$, so \(E_{1,0}^k\cong\widetilde H^k(X_m;\CC)\). In the positive degree $d=2m-1$ this is ordinary $H^d(X_m;\CC)$.
\end{proof}

\begin{theorem}[Top-depth relative term]
\label{thm:top-depth-relative-term}
For every degree $k$,
\[
 E_{1,k+1}^k\cong C^k(P_a,A_a;\CC),
\]
and, under this identification, the page-one differential $\delta_1$ is the
relative coboundary of $(P_a,A_a)$.  Hence
\[
 E_{2,k+1}^k\cong H^k(P_a,A_a;\CC).
\]
By condition C4 this term is zero unless $k=2m_a$, while
\[
 E_{2,2m_a+1}^{2m_a}=\CC[\kappa_a].
\]
\end{theorem}

\begin{proof}
Fix $k\ge0$.  A degree-$k$ simplex of depth $k+1$ is entirely private;
hence $C^k_{(k+1)}(Y_a^{\mathrm{loc}})=C^k(P_a)$.
For such a simplex $\tau$, $\Dbase|\tau\rangle$ is the signed sum of
cofaces obtained by adding one base vertex.  It is nonzero exactly when
$\tau\in A_a$.  Each output simplex has a unique base vertex, and removing
it recovers $\tau$.  Distinct private simplices have disjoint output supports, so
\[
 \ker(\Dbase|_{C^k(P_a)})
 =\Span\{|\tau\rangle:\tau\in(P_a)_k\setminus(A_a)_k\}
 =C^k(P_a,A_a).
\]
There is no incoming degree-$(k-1)$ cochain of depth $k+1$, so the incoming
image is zero.

The set $A_a$ is closed under faces: any base vertex that extends a private
simplex also extends its faces.  Conversely, if $\tau\notin A_a$, then
every private coface of $\tau$ is outside $A_a$.  Hence $\Dpriv$ preserves
the displayed relative cochain space.  For a representative $u$ in this
space, $Du=\Dpriv u$ lies at depth $k+2$.  The representative formula for
$\delta_1$ in Appendix~\ref{app:spectral-sequence} is therefore exactly the
relative coboundary, with its original incidence signs.  The only
page-one source that can enter this target under $\delta_1$ is the
corresponding preceding term, so taking cohomology gives the asserted
$E_2$ term.

In augmented degree $-1$, the only coordinate is the empty simplex, which lies in $A_a$ because the base is nonempty;
its $\Dbase$ image is the nonzero sum of base vertices.  Both sides of the
page-one identification are consequently zero there.  Lower degrees
are zero by convention.  Condition C4 now gives the unique nonzero page-two class of this form
in degree $2m_a$.
\end{proof}

After page two, the calculation has only two possible nonzero terms in the
degrees relevant to the encoded space.  Put \(d:=2m_a-1\) and \(R_a:=2m_a+1\); they are
\[
 E_{2,0}^{d}=H^d(X_{m_a};\CC),
 \qquad
 E_{2,R_a}^{d+1}=\CC[\kappa_a].
\]
These are the encoded base space and the private core.  Their degrees and
depths determine the first possible differential between them.

\subsection{The only possible long differential}

\begin{theorem}[No earlier differential reaches an endpoint]
\label{thm:no-earlier-differential}
Fix $a$, write \(m=m_a, \, d=2m-1, \, R_a=2m+1\), and regard nonexistent negative-depth or over-depth terms as zero.  For every
$2\le r<R_a$,
\[
 E_{r,r}^{d+1}=0,
 \qquad
 E_{r,R_a-r}^{d}=0.
\]
Moreover,
\[
 E_{r,-r}^{d-1}=0,
 \qquad
 E_{r,R_a+r}^{d+2}=0
 \quad(r\ge2).
\]
Thus no page-$r$ differential with $2\le r<R_a$ can leave the depth-zero
base term or enter the private core, and no page-$r$ differential can enter
the base term or leave the core.  The terms
\[
 E_{r,0}^{d}\cong H^d(X_m;\CC),
 \qquad
 E_{r,R_a}^{d+1}=\CC[\kappa_a]
\]
therefore persist unchanged to page $R_a$.  The first page on which a map
between them is degree- and depth-compatible is page $R_a$.
\end{theorem}

\begin{proof}
First account for page one.  At depth zero,
$E_{1,0}^k=\widetilde H^k(X_m)$ by \cref{thm:page-one}.
For $k\ge0$ its $\delta_1$ target has degree $k+1$ and depth one, so condition C3
makes that target zero.  No differential enters from negative depth.
The degree-$-1$ base term is zero because $X_m$ is nonempty.
Thus the depth-zero terms persist to page two.  The base calculation in
Section~\ref{sec:bowtie} and \cref{thm:top-depth-relative-term} show that the
entire augmented page two has just the following possible nonzero terms:
\[
 E_{2,0}^{d}=H^d(X_m;\CC),\qquad
 E_{2,R_a}^{d+1}=\CC[\kappa_a].
\]
Every later page is cohomology of its predecessor, so a zero term stays
zero.

More explicitly, for $2\le r<R_a=d+2$, the outgoing base target
$(k,\ell)=(d+1,r)$ is an intermediate-depth position because
$1\le r\le d+1$.  The incoming core source is
$(k,\ell)=(d,R_a-r)$, with $1\le R_a-r\le d$.
Both vanish already on page one by condition C3.  An incoming base source has
negative depth.  An outgoing core target has degree $d+2$ and depth
$R_a+r\ge d+4$, exceeding its maximum depth $d+3$.
These observations prove every displayed vanishing and show that the two
terms persist unchanged to page $R_a$.  At that page, and only at that
page, the degree and depth shift permits a map from the base term to the
core.
\end{proof}

\begin{theorem}[The page-$R_a$ differential]
\label{thm:page-Ra-differential}
Under the harmonic encoding $H^d(X_m;\CC)\cong\Qenc_m$, the page-$R_a$
differential \(\delta_{R_a}:E_{R_a,0}^{d} \longrightarrow E_{R_a,R_a}^{d+1}=\CC[\kappa_a]\) has the form
\[
 \delta_{R_a}(|\psi\rangle)
 =\alpha_a\langle\phi^{(a)}|\psi\rangle[\kappa_a],
 \qquad \alpha_a\in\CC\setminus\{0\}.
\]
In particular,
\[
 \ker\delta_{R_a}=\ker\Pi_a=(\CC|\phi^{(a)}\rangle)^\perp.
\]
If $i_a:X_m\hookrightarrow Y_a^{\mathrm{loc}}$ denotes inclusion, the limiting
depth-zero term is equivalently
\[
 E_{\infty,0}^{d}
 =\operatorname{im}\bigl(i_a^*:H^d(Y_a^{\mathrm{loc}};\CC)
                 \to H^d(X_m;\CC)\bigr)
 =\ker\Pi_a.
\]
\end{theorem}

\begin{proof}
For a finite complex over $\CC$, evaluation gives the natural duality
\[
 H^d(K;\CC)\cong\operatorname{Hom}_{\CC}(H_d(K;\CC),\CC).
\]
Under this duality, $i_a^*$ is dual to $(i_a)_*$.
By \cref{thm:local-topology}, $(i_a)_*$ is the quotient by $\CC[j_a]$.
Dualizing this quotient shows
\[
 \operatorname{im} i_a^*
 =\{[u]\in H^d(X_m;\CC):u(j_a)=0\}.
\]
To express this subspace in the harmonic encoding, write
$|\phi^{(a)}\rangle=\mathcal Z_a^{-1}\sum_b a_b|b\rangle$ with $a_b\in\ZZ$.
Under the explicit harmonic encoding, a vector
$|\psi\rangle=\sum_b\psi_b|b\rangle$ has representative
$J_m\psi=\sum_b\psi_bh_b$, with $h_b=2^{-m}z_b$.
The cycles $z_b$ have disjoint simplex supports and squared norm $4^m$.
Consequently the complex-linear cochain evaluation is
\[
 (J_m\psi)(j_a)
 =2^m\sum_ba_b\psi_b
 =2^m\mathcal Z_a\langle\phi^{(a)}|\psi\rangle.
\]
The last equality uses that the catalogue coefficients $a_b$ are real.
Hence $\operatorname{im} i_a^*=\ker\Pi_a=\phi^{(a)\perp}$ in encoded
coordinates.

By the depth-zero edge-map identification in
Appendix~\ref{app:spectral-sequence}, this is $E_{\infty,0}^d$.  By
\cref{thm:no-earlier-differential}, the base term
reaches page $R_a$ unchanged.  No differential enters it, and after page
$R_a$ any outgoing target has depth greater than the maximum $d+2$ in
degree $d+1$.  Therefore
\[
 \ker\delta_{R_a}=E_{R_a+1,0}^d
 =E_{\infty,0}^d=\phi^{(a)\perp}.
\]
This is a codimension-one kernel, so the differential has rank one and
is nonzero.  Its target is the one-dimensional space $\CC[\kappa_a]$.
The linear functionals with kernel $\phi^{(a)\perp}$ are exactly the
nonzero multiples of $\langle\phi^{(a)}|$.  Hence the differential has the
stated form for some $\alpha_a\ne0$, with the value depending on the chosen
core generator.
\end{proof}

Thus the first nonzero obstruction to the target line occurs on page \(R_a=2m_a+1\), corresponding to the eigenvalue scale \(\lambda^{2R_a}=\lambda^{4m_a+2}\).

\subsection{From page order to eigenvalue order}

The approximate harmonic spaces below describe cancellation of the
coboundary and adjoint-coboundary residuals through order \(\lambda^{r-1}\).

\begin{definition}[Approximate harmonic order]
\label{def:approximate-harmonic}
Let $D_\lambda=D_0+\lambda D_1$ be a finite-dimensional analytic
coboundary, with $D_0=\Dbase$ and $D_1=\Dpriv$.  Fix $r\ge1$.  A truncated
degree-$k$ correction through order $r-1$ is a vector-valued polynomial
\[
 u^{<r}(\lambda):=\sum_{j=0}^{r-1}\lambda^ju_j
 =u_0+\lambda u_1+\cdots+\lambda^{r-1}u_{r-1},
 \qquad u_j\in C^k.
\]
Its constant coefficient $u_0=u^{<r}(0)$ is called its leading coefficient;
$u_1,\ldots,u_{r-1}$ are correction coefficients.  Define $\mathcal V_r^k$
to be the set of leading coefficients $u_0$ for which the corrections can be
chosen so that
\[
 D_{\lambda,k}u^{<r}(\lambda)=O(\lambda^r),
 \qquad
 D_{\lambda,k-1}^*u^{<r}(\lambda)=O(\lambda^r).
\]
Equivalently, both residuals have zero Taylor coefficients in orders
$0,\ldots,r-1$.  Set $\mathcal V_0^k:=C^k$.  For an integer $s\ge0$, a
nonzero analytic eigenvalue branch $\mu(\lambda)$ has exact order $2s$ if
\[
 0<\lim_{\lambda\downarrow0}\frac{\mu(\lambda)}{\lambda^{2s}}<\infty.
\]
Thus exact order zero means that $\mu(\lambda)$ tends to a finite strictly
positive limit.
\end{definition}

The correspondence below is a finite-dimensional form of Forman's
Hodge-theoretic description of spectral sequences~\cite{Forman}; see also
the formulation used by King and Kohler~\cite[Section~9.2]{KingKohler}.
We give a self-contained proof with the depth splitting and page
conventions fixed here.

\begin{theorem}[Finite-dimensional filtered Hodge correspondence]
\label{thm:filtered-hodge}
Use the fixed orthogonal depth decomposition and the differential
$D_\lambda=\Dbase+\lambda\Dpriv$ of Section~\ref{sec:certificates}.
For each degree $k$ and integer $r\ge0$, the approximate harmonic space of
\cref{def:approximate-harmonic} is a Hilbert-space realization of page $r$:
\[
 \mathcal V_r^k\cong\bigoplus_\ell E_{r,\ell}^k.
\]
The Laplacian $\Delta_k(\lambda)$ has, near zero, an orthonormal eigenbasis
$u_i(\lambda)$ with corresponding eigenvalues $\mu_i(\lambda)$, both
depending real-analytically on $\lambda$.  We call each function
$\mu_i(\lambda)$ an eigenvalue branch.  Each eigenvalue branch is either
identically zero or has exact order $2s$ for an integer $s\ge0$, in the
sense of \cref{def:approximate-harmonic}.  Let
$P_{\ge r}^k(\lambda)$ be the orthogonal projection onto the span of the
eigenvectors whose branches have order at least $2r$, including the
identically zero branches.  Then
\[
 P_{\ge r}^k(0)=\Proj_{\mathcal V_r^k},
 \qquad
 \|P_{\ge r}^k(\lambda)-\Proj_{\mathcal V_r^k}\|_{\rm op}
 =O(|\lambda|).
\]
In particular, the number of branches of exact order $2r$ is
\[
 \dim\mathcal V_r^k-\dim\mathcal V_{r+1}^k
 =\sum_\ell\bigl(\dim E_{r,\ell}^k-\dim E_{r+1,\ell}^k\bigr),
\]
and the span of their eigenvectors at $\lambda=0$ is \(\mathcal V_r^k\cap(\mathcal V_{r+1}^k)^\perp\). The orthogonal projection onto the exact zero eigenspace converges, as
$\lambda\to0$, to the orthogonal projection onto \(\mathcal V_\infty^k:=\bigcap_r\mathcal V_r^k\).

More generally, let $J$ be any fixed set of eigenvalue branches that is
separated from the remaining spectrum for sufficiently small $\lambda$.
The orthogonal projection $P_J(\lambda)$ onto the span of the corresponding
eigenvectors depends analytically on $\lambda$ and satisfies
\[
 P_J(\lambda)=P_J(0)+O(|\lambda|),\qquad
 P_J(0)=\sum_{i\in J}|u_i(0)\rangle\langle u_i(0)|.
\]
The range of $P_J(0)$ may be a proper subspace of an eigenspace of
$\Delta_k(0)$.
\end{theorem}
\begin{proof}
Appendix~\ref{app:spectral-sequence} identifies the page quotients with
approximate harmonic leading coefficients and applies finite-dimensional
analytic perturbation theory, giving the stated form of Forman's
correspondence~\cite{Forman}.
\end{proof}

Thus a drop from page \(r\) to page \(r+1\) produces precisely the
eigenvalue branches whose leading size is \(\lambda^{2r}\).  In our local
gadget the only late drop occurs at \(r=R_a=2m_a+1\), so it produces one
eigenvalue of order \(\lambda^{4m_a+2}\).

The constants in $O(|\lambda|)$ may depend on the fixed complex.
Uniformity over the catalogue and under padding is established in
\cref{lem:uniform-intervals,thm:spectral-padding}.

In the next lemma and the padding argument below, a \emph{cluster} means the
eigenvalues contained in one of the separated intervals in the statement.

\begin{lemma}[Uniform spectral intervals for the finite catalogue]
\label{lem:uniform-intervals}
There exist computable constants
\[
 0<\lambda_*<1,
 \qquad
 0<c_{\rm spec}\le C_{\rm spec}<\infty,
 \qquad
 C_{\rm proj}<\infty,
\]
depending only on the finite catalogue, such that for every catalogue type,
each of the three degrees \(k\in\{2m_a-2,\,2m_a-1,\,2m_a\}\), and every $0<\lambda<\lambda_*$:
\begin{enumerate}[label=(\roman*)]
\item every nonzero branch of exact order $2r$ lies in \([c_{\rm spec}\lambda^{2r}, C_{\rm spec}\lambda^{2r}]\);
\item the intervals belonging to distinct orders are disjoint;
\item if $P_I(\lambda)$ is the orthogonal projection onto one such cluster
and $P_I(0)$ is the orthogonal projection onto the span of the
$\lambda\to0$ limits of its analytic eigenvectors, then
\[
 \|P_I(\lambda)-P_I(0)\|_{\rm op}\le C_{\rm proj}\lambda;
\]
\item all nonzero eigenvalues are exhausted by these clusters.
\item In local degree $2m_a-2$, there is no eigenvalue in \([c_{\rm spec}\lambda^{4m_a+2}, C_{\rm spec}\lambda^{4m_a+2}]\).
\end{enumerate}
\end{lemma}

\begin{proof}[Analytic bounds]
Put $d=2m_a-1$ and $R_a=d+2=2m_a+1$.
By \cref{thm:no-earlier-differential,thm:page-Ra-differential}, page two
has only the base term of dimension $2^{m_a}$ in degree $d$ and the
one-dimensional core in degree $d+1$.  These persist to page $R_a$.
The rank-one differential on that page removes one dimension from each;
afterward only the $2^{m_a}-1$ base classes remain.  Thus the only possible
strict page-dimension drops occur from page zero to one, from one to two,
and from $R_a$ to $R_a+1$.  By \cref{thm:filtered-hodge}, every positive
Laplacian branch has order in \(\{0,2,2R_a\}\). The order-$2R_a$ drop has multiplicity one in degrees $d,d+1$ and zero
in degree $d-1$.  This already establishes the order exclusion needed
for (v), independently of \cref{prop:local-clusters}.

For each nonzero analytic eigenvalue branch write \(\mu_i(\lambda)=\lambda^{2r_i}b_i(\lambda), \, b_i(0)>0\). There are finitely many such branches over the finite catalogue and the
three specified degrees.  On a sufficiently small common closed interval,
each $b_i$ is bounded between half its positive value at zero and twice
that value.  Taking the minimum and maximum over all branches gives
positive constants $c_{\rm spec},C_{\rm spec}$ and proves (i) and (iv).
Include the three candidate orders for each type when choosing the
interval, even when a candidate cluster is empty.  Decrease $\lambda_*$
until
\[
 C_{\rm spec}\lambda_*^2<c_{\rm spec}.
\]
For integers $s>r\ge0$, this implies
$C_{\rm spec}\lambda^{2s}<c_{\rm spec}\lambda^{2r}$ whenever
$0<\lambda<\lambda_*$.  It proves (ii) and, together with the order
exclusion above, proves (v).

Each order cluster follows a fixed finite set of analytic branches.
Its orthogonal projection extends analytically to zero by
\cref{thm:filtered-hodge}.  On a smaller common closed interval its
operator-norm derivative is bounded.  Integrating that derivative gives
\[
 \|P_I(\lambda)-P_I(0)\|_{\rm op}
 \le\lambda\sup_{0\le t\le\lambda_*}\|P_I'(t)\|_{\rm op}.
\]
A maximum over the finitely many projections proves (iii).  The same
argument also applies to the exact-zero branch projection and their
finite sums; we include those bounds in $C_{\rm proj}$ for use below.
All choices so far depend only on the catalogue.

This gives catalogue-uniform analytic constants.  Effective rational
choices are supplied by \cref{prop:local-effectivity}.
\end{proof}

\subsection{Complete local spectral statement}

\begin{proposition}[Three-degree spectral decomposition]
\label{prop:local-clusters}
Fix a catalogue type $a$, put $m=m_a$, $d=2m-1$, and let
$0<\lambda<\lambda_*$.  For the weighted local gadget:
\begin{enumerate}[label=(\roman*)]
\item in degree $d$, the zero eigenspace has dimension $2^m-1$, exactly one
eigenvalue lies in
\[
 [c_{\rm spec}\lambda^{4m+2},
  C_{\rm spec}\lambda^{4m+2}],
\]
and every other positive eigenvalue lies in
\[
 [c_{\rm spec}\lambda^2,C_{\rm spec}\lambda^2]
 \cup[c_{\rm spec},C_{\rm spec}];
\]
\item in degree $d-1$, the kernel and the
$\lambda^{4m+2}$ interval are both zero, and every eigenvalue belongs to the
$\lambda^2$ or constant interval above;
\item in degree $d+1$, the kernel is zero, exactly one eigenvalue belongs to
the $\lambda^{4m+2}$ interval, and every other eigenvalue belongs to the
$\lambda^2$ or constant interval.
\end{enumerate}
\end{proposition}

\begin{proof}
For $\lambda>0$, multiplication by $\lambda^\ell$ on the depth-$\ell$
coordinates gives an invertible chain isomorphism from the unweighted
complex to the weighted one, as in Appendix~\ref{app:spectral-sequence}.
Consequently their cohomology dimensions agree.  By
\cref{thm:local-topology} and the base homology calculation, the reduced
cohomology dimensions in degrees $d-1,d,d+1$ are respectively \(0,\, 2^m-1,\, 0\). Here the augmented convention matters if $d-1=0$: the asserted kernel
is for the augmented Laplacian fixed in the preliminaries.
Finite Hodge theory identifies these dimensions with the exact kernel
multiplicities for every $\lambda>0$.

For positive eigenvalues use the multiplicity formula in
\cref{thm:filtered-hodge}.  The long differential has rank one from
degree $d$ to degree $d+1$, so the dimension drop from page $2m+1$ to
page $2m+2$ is one in each of those degrees and zero in degree $d-1$.
All other positive branches come from the page-zero or page-one drops,
as established in the proof of \cref{lem:uniform-intervals}; their orders
are respectively zero and two.  The mutually disjoint intervals and stated
multiplicities now follow from that lemma.
\end{proof}

The following isometry identifies the low-energy space with the encoded
qubit space and realizes its Laplacian as a positive multiple of \(\Pi_a\).

\begin{theorem}[Rank-one action on the local low-energy space]
\label{thm:local-rank-one}
For each catalogue type $a$ and $0<\lambda<\lambda_*$, let
$\mathscr L_a(\lambda)\subseteq C^{q_{m_a}}(Y_a^{\mathrm{loc}})$ be the spectral
subspace consisting of the zero eigenspace together with the unique
one-dimensional eigenspace in the $\lambda^{4m_a+2}$ interval.  Then \(\dim\mathscr L_a(\lambda)=2^{m_a}\). There is an isometry \(W_a(\lambda):\Qenc_{m_a}\longrightarrow\mathscr L_a(\lambda)\) and a scalar $\varepsilon_a(\lambda)>0$ such that
\[
 \|W_a(\lambda)-J_{m_a}\|_{\rm op}
 \le C_{\rm proj}\lambda,
\]
\[
 W_a(\lambda)^*\Delta_{q_{m_a}}(\lambda)W_a(\lambda)
 =\varepsilon_a(\lambda)\Pi_a,
\]
and
\[
 c_{\rm spec}\lambda^{4m_a+2}
 \le\varepsilon_a(\lambda)
 \le C_{\rm spec}\lambda^{4m_a+2}.
\]
In the first norm estimate, $J_{m_a}$ is followed by the coordinate inclusion
$C^{q_{m_a}}(X_{m_a})\hookrightarrow C^{q_{m_a}}(Y_a^{\mathrm{loc}})$.
\end{theorem}

\begin{proof}
Fix a catalogue type and abbreviate $J=J_m$, $\Pi=\Pi_a$, $d=2m-1$,
and $R=2m+1$.  Regard $J$ as taking values in the full local cochain space.
The realization in \cref{thm:filtered-hodge} sends the surviving
depth-zero terms on page two to a subspace of $C^d_{(0)}$ contained in
$\ker\Dbase\cap\ker\Dbase^*$.  At depth zero this is exactly the encoded
base harmonic space $\operatorname{im} J$.  Both spaces have dimension $2^m$, and page
two has no other degree-$d$ term.  Therefore
\[
 \mathcal V_2^d=\operatorname{im} J.
\]
If $v\in\mathcal V_\infty^d$, Step 2 of
Appendix~\ref{app:spectral-sequence} supplies an analytic harmonic germ
$u(\lambda)$ with $u(0)=v$.  Write $u(\lambda)=\sum_{j\ge0}\lambda^ju_j$
and let $\pi_p$ denote exact-depth projection.  The diagonal-index-zero
part of $D_\lambda u=0$ says that
\[
 z:=\sum_{p=0}^{d+1}\pi_pu_p\quad\hbox{satisfies}\quad Dz=0.
\]
Indeed, its depth-$p$ equation is
$\Dbase\pi_pu_p+\Dpriv\pi_{p-1}u_{p-1}=0$, including the final output
depth $d+2$.  Since $v\in\mathcal V_2^d=\operatorname{im} J$, restriction of $z$ to
the base equals $v$.  The edge-map calculation in
\cref{thm:page-Ra-differential} therefore puts $v$ in
$J\phi^{(a)\perp}$.  Both spaces have dimension $2^m-1$, so they agree.
Hence
\[
 \mathcal V_\infty^d=J\phi^{(a)\perp},
 \qquad
 \mathcal V_R^d\cap(\mathcal V_{R+1}^d)^\perp
 =\CC J\phi^{(a)}.
\]

Let $P_h(\lambda)$ and $P_e(\lambda)$ be the exact-zero and unique
order-$2R$ eigenprojections.  Their limits and analytic bounds are
\[
 P_h(0)=J(I-\Pi)J^*,\qquad P_e(0)=J\Pi J^*,
 \qquad \|P_h(\lambda)-P_h(0)\|_{\rm op}
       +\|P_e(\lambda)-P_e(0)\|_{\rm op}=O(\lambda).
\]
Define a map from the encoded space into $\mathscr L_a(\lambda)$ by
\[
 F(\lambda)=P_h(\lambda)J(I-\Pi)+P_e(\lambda)J\Pi,
 \qquad G(\lambda)=F(\lambda)^*F(\lambda).
\]
Then $F(0)=J$, $G(0)=I$, and $F(\lambda)=J+O(\lambda)$.
Since the two output projections are orthogonal, $G$ is block diagonal
for $\phi^{(a)\perp}\oplus\CC\phi^{(a)}$ and commutes with $\Pi$.
For a smaller common $\lambda_*$, $\|G-I\|_{\rm op}<1/2$, so set
\[
 W_a(\lambda):=F(\lambda)G(\lambda)^{-1/2}.
\]
The inverse square root is analytic near $I$ (its convergent binomial
series suffices).  Thus $W_a^*W_a=I$, $W_a(0)=J$, and
$W_a(\lambda)-J=O(\lambda)$.  Its range has dimension $2^m$ and lies in
$\mathscr L_a(\lambda)$, whose dimension is $2^m$ by
\cref{prop:local-clusters}; the range therefore equals that subspace.

Because $G^{-1/2}$ commutes with $\Pi$, this isometry sends
$\phi^{(a)\perp}$ into the exact kernel and the target line into the
unique lifted eigenline.  Writing its positive eigenvalue as
$\varepsilon_a(\lambda)$ gives
\[
 \Delta_d(\lambda)W_a(\lambda)
 =\varepsilon_a(\lambda)W_a(\lambda)\Pi.
\]
Multiplication by $W_a^*$ proves the exact compression identity.
The eigenvalue bounds are those of \cref{prop:local-clusters}.
Finally take a maximum of the finitely many derivative bounds for $W_a$
and enlarge $C_{\rm proj}$ accordingly, decreasing the common
$\lambda_*$ if necessary.  Earlier bounds remain valid, and the stated
norm estimate follows uniformly over catalogue types.  This argument
proves existence of the constants; their effective rational selection
is proved in \cref{prop:local-effectivity}.
\end{proof}

\subsection{Uniform identity padding}

\begin{definition}[Local-to-padded degree dictionary]
Let $a$ act on $m=m_a$ qubits in an $N$-qubit instance.  Tensoring a local
degree-$k$ cochain with a vector in \(\Eenc_{N-m}:=\mathscr H^{q_{N-m}}(X_{N-m})\) places it in join degree \(k+q_{N-m}+1=k+2(N-m)\). This space is the encoded harmonic space of the $N-m$ bow-tie factors on
which the projector acts as the identity.  When $N=m$, use the augmented
convention $X_0=\{\varnothing\}$ and $\Eenc_0=\CC$ in degree $-1$.
In particular,
\[
 2m-2\leftrightarrow2N-2,
 \quad
 2m-1\leftrightarrow2N-1,
 \quad
 2m\leftrightarrow2N.
\]
\end{definition}

\begin{theorem}[Padding preserves the local spectral intervals and penalty]
\label{thm:spectral-padding}
For every $a,S,N$ and $0<\lambda<\lambda_*$, the local clusters of
\cref{prop:local-clusters} tensored with the explicitly defined space
$\Eenc_{N-m_a}$ occur
in padded degrees $2N-2,2N-1,2N$ with the same endpoints and projection
constant on the reducing subspace with harmonic untouched factor.  Decrease $c_{\rm spec}$ once, if necessary, so that it is no
larger than the least positive eigenvalue among the finitely many degree
Laplacians of the single bow-tie complex.  Then every padded mode using a
positive-eigenvalue state of $X_{N-m_a}$ has eigenvalue at least
$c_{\rm spec}$.  The high eigenvalues are sums of factor eigenvalues and
may grow with $N$.

In degree $q_N$, there exist a $2^N$-dimensional low-energy space
$\mathscr L_{a,S}^{(N)}(\lambda)$ and an isometry \(W_{a,S}^{(N)}(\lambda):\Qenc_N \longrightarrow\mathscr L_{a,S}^{(N)}(\lambda)\) such that
\[
 \|W_{a,S}^{(N)}(\lambda)-J_N\|_{\rm op}
 \le C_{\rm proj}\lambda,
\]
and
\[
 (W_{a,S}^{(N)})^*\Delta_{q_N}(Y_{a,S}^{\mathrm{pad}})
 W_{a,S}^{(N)}
 =\varepsilon_a(\lambda)P_{a,S}^{(N)}.
\]
Here $J_N$ is viewed in the padded cochain space through the base-coordinate
inclusion.
All constants are independent of $N$, $S$, and the number of occurrences in
the eventual global instance.
\end{theorem}

\begin{proof}
Put $m=m_a$ and $n=N-m$.  First group the local and untouched factors,
postponing their placement in $S$.  By \cref{thm:join-padding}, the padded
complex is the join of the local and untouched complexes.  By
\cref{prop:weighted-join-laplacian}, tensor products of orthonormal
eigenvectors give an orthonormal eigenbasis in each homogeneous summand, and
the eigenvalues add:
\[
 \Delta_K u=\mu u,\quad\Delta_{X_n}v=\nu v
 \quad\Longrightarrow\quad
 \Delta_{K*X_n}(u*v)=(\mu+\nu)(u*v).
\]
All factor eigenvalues are nonnegative.  Let $\gamma_B>0$ be the least
positive eigenvalue over all augmented degree Laplacians of one bow-tie.
There are finitely many such matrices.  Applying the same tensor formula
to its $n$ copies shows that every positive eigenvalue of $X_n$ is a sum
containing at least one positive bow-tie eigenvalue.  It is therefore at
least $\gamma_B$, independently of $n$.  Moreover its kernel is precisely
$\Eenc_n$, concentrated in degree $q_n$; for $n=0$ the positive sector is
empty and the kernel is the one-dimensional unit factor.

On the reducing subspace with untouched factor in $\Eenc_n$, the
Laplacian is exactly $\Delta_K\otimes I_{\Eenc_n}$.  A local degree $k$
therefore contributes in degree $k+q_n+1=k+2n$, with identical eigenvalues
and multiplicities multiplied by $2^n$.  If $P(\lambda)$ is any local
cluster projection and $Q_n$ projects onto $\Eenc_n$, its projection in
this reducing subspace, extended by zero, is $P(\lambda)\otimes Q_n$.
Since $\|Q_n\|_{\rm op}=1$,
\[
 \|(P(\lambda)-P(0))\otimes Q_n\|_{\rm op}
 =\|P(\lambda)-P(0)\|_{\rm op},
\]
so the projection bound is uniform in $n$.
The local constant-order cluster is selected by this reducing subspace;
its eigenvalues may overlap those of other high modes.

On its orthogonal complement every tensor eigenvector has $\nu>0$, so
its eigenvalue is at least $\gamma_B$.  Replace $c_{\rm spec}$ by
$\min\{c_{\rm spec},\gamma_B\}$ and decrease the common $\lambda_*$,
if necessary, to retain $C_{\rm spec}\lambda_*^2<c_{\rm spec}$.
The bulk and lifted intervals then lie strictly below this high sector.

Let $U_S$ be the unitary cochain identification from the grouped join to
the labelled padded complex, including the orientation signs.  Define
\[
 W_{a,S}^{(N)}
 :=U_S(W_a\otimes J_n)W_S^*,\qquad
 \mathscr L_{a,S}^{(N)}
 :=U_S(\mathscr L_a\otimes\Eenc_n).
\]
Take $J_0=I_\CC$ when $n=0$.  The factor maps are isometries, so this is
an isometry onto the displayed $2^m2^n=2^N$ dimensional space.  By the
local spectral theorem and the high-sector lower bound, this is exactly
the spectral subspace comprising the kernel and the lifted interval in
degree $q_N$.  The encoding definitions and the even shifted degrees of
the bow-tie cycles give
\[
 U_S(J_m\otimes J_n)W_S^*=J_N.
\]
Unitary invariance and the tensor norm identity yield
\[
 \|W_{a,S}^{(N)}-J_N\|_{\rm op}
 =\|(W_a-J_m)\otimes J_n\|_{\rm op}
 =\|W_a-J_m\|_{\rm op}\le C_{\rm proj}\lambda.
\]
Finally $\Delta_{X_n}J_n=0$, so compression of the tensor-sum formula,
using \cref{thm:local-rank-one}, gives
\[
 (W_{a,S}^{(N)})^*\Delta_{q_N}W_{a,S}^{(N)}
 =\varepsilon_a(\lambda)W_S(\Pi_a\otimes I)W_S^*
 =\varepsilon_a(\lambda)P_{a,S}^{(N)}.
\]
Every constant used is local or comes from the single bow-tie, proving
independence of $N$, $S$, and the number of gadget occurrences.
\end{proof}

The preceding results establish constants that depend only on the fixed
catalogue.  For the reduction, we also need rational values that can be
computed once from the finite certificate data and then used in every
instance.  The final proposition records that effective choice.

\begin{proposition}[Effectivity of all local constants]
\label{prop:local-effectivity}
The constants $\lambda_*$, $c_{\rm spec}$, $C_{\rm spec}$, and
$C_{\rm proj}$ from \cref{lem:uniform-intervals}, together with the catalogue-wide
positive singular-value bound \(a_*=\min_{1\le a\le s}s_{\min}(T_{a,0})>0\) from condition C4, can be replaced by explicit rational bounds computed from the
finitely many certified integer matrices.  The single threshold $\lambda_*$
is chosen to satisfy every separation property stated in
\cref{lem:uniform-intervals} uniformly for every source instance.
The reduction may therefore hard-code one rational collection of valid
constants.
\end{proposition}

\begin{proof}
In weighted orthonormal simplex coordinates, every local differential has
the form $D_0+\lambda D_1$ with integer matrices $D_0,D_1$.  Hence every
local Laplacian is a finite real symmetric matrix whose entries lie in
$\ZZ[\lambda]$.  The spaces $\mathcal V_r^k$ of
\cref{def:approximate-harmonic}, their dimensions, and their limiting
orthogonal projections are therefore computable exactly by rational
Gaussian elimination.

For a proposed rational tuple
$(\ell,c,C,K,a)$, the assertions
\[
  0<\lambda<\ell,
  \qquad
  c\lambda^{2r}\le\mu\le C\lambda^{2r},
  \qquad
  \|P_r(\lambda)-P_r(0)\|_{\rm op}\le K\lambda,
\]
together with the required multiplicities, the forbidden-interval
condition, the bow-tie gap, and the bound
$s_{\min}(T_{a,0})\ge a$, can all be expressed by finitely many polynomial
equalities and inequalities with rational coefficients.  Their truth is
decidable by effective quantifier elimination over the real closed field;
see \cite[Section~2.1]{Basu}.

The existence results in
\cref{lem:uniform-intervals,thm:local-rank-one,thm:spectral-padding} and
condition C4 show that some real tuple satisfies all these inequalities with
strict margins.  Consequently a rational tuple does as well.  Enumerating
rational candidates and applying the preceding decision procedure must
therefore terminate.  Since the catalogue is fixed and finite, the resulting
values may be computed once and hard-coded; they do not depend on $N$, on the
supports, or on the number of gadget occurrences.
\end{proof}

\section{The global spectral gap}
\label{sec:global-spectrum}

Section~\ref{sec:single-gadget-spectrum} describes the spectrum of each
padded gadget separately.  The remaining problem is global: all gadget
occurrences share the same base complex, so their down-Laplacian terms need
not be orthogonal.  We must show that private components belonging to
different occurrences cannot cancel through that common base and create a
new low-energy state.

The use of single-gadget restrictions and separate control of up- and
down-boundary contributions follows the strategy of King and Kohler
\cite[Section~10 and Appendix~B]{KingKohler}.  Here the local inputs and
quantitative estimates are established for the complexes defined above.

For a global cochain \(u\), restrict to the base together with one gadget
occurrence and decompose each restriction into the four local spectral
sectors from Section~\ref{sec:single-gadget-spectrum}.  The lifted, bulk,
and high sectors are controlled by the global Hodge energy, while the
harmonic sector is compared with the encoded space on which \(I-P_j\) acts.
Summing over occurrences brings in the source inequality \(\sum_{j=1}^{t}P_j\succeq gI\) to supply the global spectral gap.

\subsection{Global energy and the local views}
\label{subsec:global-energy}

\begin{definition}[Global weighted complex and energy]
\label{def:global-energy}
Fix a source presentation
\[
  H=\sum_{j=1}^{t}P_j,
  \qquad
  P_j=(\Pi_{a_j})_{S_j}^{(N)},
  \qquad
  t\geq1,
  \qquad
  M:=\max_{1\leq j\leq t}m_{a_j}\leq4.
\]
Give every base vertex of \(\widehat G(H)\) weight \(1\) and every private
vertex weight \(0<\lambda<1\).  Put \(q:=q_N=2N-1\) and write
\[
  \widehat\Delta_q(\lambda)
  :=
  \widehat D_{q-1}(\lambda)\widehat D_{q-1}(\lambda)^*
  +
  \widehat D_q(\lambda)^*\widehat D_q(\lambda)
\]
for the weighted degree-\(q\) Laplacian of \(\widehat Y(H)\).

For \(u\in C^q(\widehat Y(H);\CC)\), define its Hodge energy by
\[
  \mathcal E(u)
  :=
  \langle u,\widehat\Delta_q(\lambda)u\rangle
  =
  \|\widehat D_{q-1}(\lambda)^*u\|^2
  +
  \|\widehat D_q(\lambda)u\|^2.
\]
\end{definition}

Thus \(\mathcal E(u)\geq0\), and \(\mathcal E(u)=0 \,\Longleftrightarrow\, u\in\ker\widehat\Delta_q(\lambda)\). The smallest eigenvalue of \(\widehat\Delta_q(\lambda)\) is the minimum of \(\mathcal E(u)\) over unit degree-\(q\) cochains.

\begin{definition}[Base and private coordinate projections]
\label{def:coordinate-projections}
By condition C2, every simplex of \(\widehat Y(H)\) is either a base simplex
or contains private vertices belonging to exactly one gadget occurrence.
In degree \(q\), let \(\mathsf C_0,\mathsf C_1,\ldots,\mathsf C_t\) be the mutually orthogonal coordinate projections onto the base simplices
and the simplices private to occurrences \(1,\ldots,t\).  Hence
\[
  I=\mathsf C_0+\sum_{j=1}^{t}\mathsf C_j.
\]

For a global cochain \(u\), write
\[
  \omega_0:=\mathsf C_0u,
  \qquad
  \omega_j:=\mathsf C_ju,
  \qquad
  u_j:=(\mathsf C_0+\mathsf C_j)u=\omega_0+\omega_j.
\]
The vector \(u_j\) is the part of \(u\) visible in the \(j\)-th padded
single-gadget complex.  Denote the coboundary, weighted boundary, and
degree-\(q\) Laplacian of that padded complex by
\[
  D_j,
  \qquad
  \partial_j:=D_{j,q-1}^*,
  \qquad
  \Delta_j:=\partial_j^*\partial_j+D_{j,q}^*D_{j,q}.
\]
\end{definition}

The base component \(\omega_0\) appears in every \(u_j\); only the
private components \(\omega_j\) are mutually orthogonal.

\begin{definition}[The four local spectral sectors]
\label{def:local-sector-projections}
Fix \(0<\lambda<\lambda_*\).  For occurrence \(j\), let \(I_j^{\mathrm{har}}:=\{0\}\), \[
  I_j^{\mathrm{lift}}(\lambda)
  :=
  [c_{\mathrm{spec}}\lambda^{4m_{a_j}+2},
   C_{\mathrm{spec}}\lambda^{4m_{a_j}+2}],
\]
\[
  I_j^{\mathrm{bulk}}(\lambda)
  :=
  [c_{\mathrm{spec}}\lambda^2,
   C_{\mathrm{spec}}\lambda^2],
  \qquad
  I_j^{\mathrm{high}}
  :=
  [c_{\mathrm{spec}},\infty).
\]
Let \(P_j^{\mathrm{har}}, \, P_j^{\mathrm{lift}}, \, P_j^{\mathrm{bulk}}, \, P_j^{\mathrm{high}}\) be the corresponding spectral projections of \(\Delta_j\), extended by
zero outside the \(j\)-th padded cochain space.

By \cref{thm:spectral-padding}, these four sectors are pairwise orthogonal
and exhaust the degree-\(q\) padded cochain space:
\[
  \mathsf C_0+\mathsf C_j
  =
  P_j^{\mathrm{har}}
  +P_j^{\mathrm{lift}}
  +P_j^{\mathrm{bulk}}
  +P_j^{\mathrm{high}}.
\]
\end{definition}

The harmonic sector is the surviving encoded space.  The lifted sector is
the low-energy penalty corresponding
to \(\operatorname{im}P_j\); before padding, its local factor is one
dimensional.  The bulk sector has energy of order \(\lambda^2\), and the
high sector has energy bounded below by a positive constant.

For later comparison, view the encoded isometry \(J_N\) as taking values in
the base coordinates of the global cochain space.  Define
\[
  \overline P_j^{\mathrm{har}}
  :=
  J_N(I-P_j)J_N^*.
\]
Also define the unperturbed padded bulk space \(\overline{\mathscr B}_j := \Bbulk_{a_j}^{\mathrm{loc}} \otimes \Eenc_{N-m_{a_j}}\) after relabelling the local factor according to \(S_j\) and using the join
cochain identification.  Let \(\overline P_j^{\mathrm{bulk}} := \Proj_{\overline{\mathscr B}_j}\). Both projections are extended by zero away from their indicated subspaces.

\subsection{Isolation of private outputs}
\label{subsec:global-isolation}

The up-coboundary outputs of different occurrences are completely
disjoint.  Down-boundary outputs can meet on the base, so for the bulk
estimate we retain only outputs containing the private apex of the relevant
occurrence.

The following identity is the up-Laplacian decomposition of King and
Kohler~\cite[Claim~10.4]{KingKohler}, written in the present notation.
We include the proof to check it directly from the global graph definition.

\begin{lemma}[No crossing of top cofaces]
\label{lem:no-crossing}
For every \(u\in C^q(\widehat Y(H);\CC)\), \[
  \|\widehat D_qu\|^2
  =
  \sum_{j=1}^{t}\|D_{j,q}u_j\|^2.
\]
In contrast, one need not have
\[
  \|\widehat D_{q-1}^*u\|^2
  =
  \sum_{j=1}^{t}\|D_{j,q-1}^*u_j\|^2.
\]
\end{lemma}

\begin{proof}
The base complex \(X_N\) has dimension \(q\), so it has no \((q+1)\)-simplices.  Every global \((q+1)\)-simplex therefore contains a
private vertex.  Condition C2 implies that all of its private vertices
belong to one occurrence \(j\).  Every degree-\(q\) face of that simplex is
then either a base face or belongs to the same occurrence.  Its coefficient
in \(\widehat D_qu\) is consequently the corresponding coefficient in \(D_{j,q}u_j\), with the same incidence sign and weight.

The \((q+1)\)-simplex supports belonging to different occurrences are
disjoint, so the squared norms add.  A down-boundary can instead delete the
only private vertex and produce a base \((q-1)\)-simplex.  Such outputs from different occurrences may coincide.
\end{proof}

Selecting output simplices containing the cone apex also appears in
\cite[Section~10.2 and Appendix~B.3]{KingKohler}.  Below, the corresponding
lower bound is obtained from the certified private-sector incidence matrix
of \cref{def:principal-bulk-matrix}.

\begin{definition}[Selected outputs]
\label{def:selected-output}
Fix an occurrence \(j\), and put \(m=m_{a_j}\).  Under the join cochain
identification, define
\[
  \mathscr O_{j,N}^{\downarrow}
  :=
  \Odown_{a_j}\otimes\Eenc_{N-m}
  \subseteq
  C^{q-1}(Y_{a_j,S_j}^{\mathrm{pad}}),
\]
\[
  \mathscr O_{j,N}^{\uparrow}
  :=
  \Oup_{a_j}\otimes\Eenc_{N-m}
  \subseteq
  C^{q+1}(Y_{a_j,S_j}^{\mathrm{pad}}).
\]
Let \(S_j^\downarrow\) and \(S_j^\uparrow\) be their orthogonal
projections, extended by zero to the full global output spaces.  Define
\[
  T_j^{\mathrm{sel}}v
  :=
  \bigl(
    S_j^\downarrow\partial_jv,\,
    S_j^\uparrow D_{j,q}v
  \bigr).
\]
Every simplex coordinate supporting either selected output contains the
private apex of occurrence \(j\).
\end{definition}

\begin{lemma}[Selected-output isolation]
\label{lem:selected-isolation}
For every global degree-\(q\) cochain \(u\) and every occurrence \(j\), \[
  S_j^\downarrow\widehat D_{q-1}^*u
  =
  S_j^\downarrow\partial_ju_j,
  \qquad
  S_j^\uparrow\widehat D_qu
  =
  S_j^\uparrow D_{j,q}u_j.
\]
Moreover, selected output spaces belonging to different occurrences are
orthogonal.  Consequently,
\[
  \sum_{j=1}^{t}\|T_j^{\mathrm{sel}}u_j\|^2
  \leq
  \mathcal E(u).
\]
\end{lemma}

\begin{proof}
Every selected output contains the private apex of its occurrence.  By
condition C2, no face or coface contributing to that output can contain a
private vertex from another occurrence.  Its coefficient therefore depends
only on \(u_j\), and the local and global incidence coefficients agree.

Selected spaces belonging to different occurrences have disjoint
apex-labelled coordinate supports and are therefore orthogonal.  Applying
Bessel's inequality separately to the global down- and up-residuals gives
the energy bound.
\end{proof}

\begin{lemma}[Leakage from private coordinates to the base]
\label{lem:boundary-leakage}
There is a constant \(C_{\mathrm{leak}}\), depending only on the finite
catalogue, such that for every occurrence \(j\), every \(0<\lambda<\lambda_*\), and every \(v\in\operatorname{im}\mathsf C_j\), \[
  \|\mathsf C_0^{q-1}\partial_jv\|
  \leq
  C_{\mathrm{leak}}\lambda\|v\|.
\]
Here \(\mathsf C_0^{q-1}\) denotes the coordinate projection onto base
simplices in degree \(q-1\).
\end{lemma}

\begin{proof}
A boundary face is entirely in the base only when the source simplex has
exactly one private vertex and that vertex is deleted.  In the weighted
orthonormal basis, this deletion has coefficient \(\lambda\) times its
orientation sign.

Under join padding, deleting a vertex in an untouched factor leaves the
private local vertex present and is therefore annihilated by the base
projection.  The remaining map is a direct sum of operators of the form \(\lambda L_{a,k}\otimes I\), where \(L_{a,k}\) is a finite local incidence matrix.  Hence its operator
norm is at most \(\lambda\max_{a,k}\|L_{a,k}\|_{\mathrm{op}}\). The maximum is over a fixed finite list and defines \(C_{\mathrm{leak}}\), independently of \(N\), \(t\), and the placement.
\end{proof}

\subsection{Controlling the four local sectors}
\label{subsec:local-sector-control}

\begin{lemma}[Stability of the harmonic and bulk projections]
\label{lem:projection-stability}
For every occurrence \(j\) and \(0<\lambda<\lambda_*\), \[
  \|P_j^{\mathrm{har}}-\overline P_j^{\mathrm{har}}\|_{\mathrm{op}}
  \leq
  C_{\mathrm{proj}}\lambda,
\]
\[
  \|P_j^{\mathrm{bulk}}-\overline P_j^{\mathrm{bulk}}\|_{\mathrm{op}}
  \leq
  C_{\mathrm{proj}}\lambda.
\]
In particular,
\[
  P_j^{\mathrm{har}}
  =
  J_N(I-P_j)J_N^*
  +
  R_j(\lambda),
  \qquad
  \|R_j(\lambda)\|_{\mathrm{op}}
  \leq
  C_{\mathrm{proj}}\lambda.
\]
\end{lemma}

\begin{proof}
In local degree \(d=2m_a-1\), the page calculation of
\cref{thm:top-depth-relative-term}, together with
\cref{thm:filtered-hodge,thm:local-rank-one}, identifies the limiting
harmonic space with the surviving encoded space and the limiting
order-\(\lambda^2\) bulk space with \(C^d(P_a,A_a;\CC)\). The filtered Hodge correspondence and the uniform interval lemma therefore
give the two local projection estimates.  By
\cref{thm:spectral-padding}, padding tensors these projections with the
orthogonal projection onto \(\Eenc_{N-m_a}\).  Tensoring with a norm-one
projection and extending by zero preserve the operator-norm bounds.  The
resulting comparison projections are precisely \(\overline P_j^{\mathrm{har}}\) and \(\overline P_j^{\mathrm{bulk}}\).
\end{proof}

\begin{lemma}[The lifted sector is controlled by the global energy]
\label{lem:lifted-up}
For every unit global cochain \(u\), \[
  \sum_{j=1}^{t}
  \|P_j^{\mathrm{lift}}u_j\|^2
  \leq
  c_{\mathrm{spec}}^{-1}
  \lambda^{-(4M+2)}
  \mathcal E(u).
\]
\end{lemma}

\begin{proof}
The adjacent-degree exclusion in
\cref{lem:uniform-intervals,thm:spectral-padding} says that the padded
degree-\((q-1)\) Laplacian has no spectrum in \(I_j^{\mathrm{lift}}(\lambda)\).  Since the boundary intertwines adjacent
Hodge Laplacians,
\[
  \Delta_{j,q-1}\partial_j
  =
  \partial_j\Delta_{j,q},
\]
the boundary of a lifted eigenvector would be an adjacent-degree
eigenvector with the same eigenvalue.  It must therefore vanish: \(\partial_jP_j^{\mathrm{lift}}=0\). Thus, on the lifted sector,
\[
  D_{j,q}^*D_{j,q}
  \succeq
  c_{\mathrm{spec}}
  \lambda^{4m_{a_j}+2}
  P_j^{\mathrm{lift}}.
\]
Apply this inequality to \(u_j\), sum over \(j\), and use
\cref{lem:no-crossing}.  Since \(m_{a_j}\leq M\) and \(0<\lambda<1\), \(\lambda^{4m_{a_j}+2}\geq\lambda^{4M+2}\), which proves the claim.
\end{proof}

\begin{lemma}[High-sector estimate]
\label{lem:high-sector}
There is a catalogue-dependent constant \(C_{\mathrm{high}}\) such that, if \(u\) is a unit degree-\(q\) cochain and \(\mathcal E(u)\leq\lambda^2\), then
\[
  \sum_{j=1}^{t}\|P_j^{\mathrm{high}}u_j\|^2
  \leq
  C_{\mathrm{high}}\lambda^2t^2.
\]
\end{lemma}

\begin{proof}
Write \(v_j:=\mathsf C_ju, \, u_0:=\mathsf C_0u, \, b_j:=\mathsf C_0^{q-1}\partial_jv_j\). Let \(\partial_X\) denote the weighted boundary on the base complex, and set \(b:=\mathsf C_0^{q-1}\widehat D_{q-1}^*u\). The base part of the \(j\)-th local down-output is
\[
  \mathsf C_0^{q-1}\partial_ju_j
  =
  \partial_Xu_0+b_j
  =
  b-\sum_{i\neq j}b_i.
\]
All nonbase down-outputs carry a unique occurrence label, so their total
squared norm is bounded by \(\mathcal E(u)\).  By Cauchy-Schwarz and
\cref{lem:boundary-leakage},
\[
  \begin{aligned}
  \sum_{j=1}^{t}
  \left\|b-\sum_{i\neq j}b_i\right\|^2
  &\leq
  2t\|b\|^2
  +
  2(t-1)^2\sum_{i=1}^{t}\|b_i\|^2
  \\
  &\leq
  2t\mathcal E(u)
  +
  2C_{\mathrm{leak}}^2\lambda^2(t-1)^2,
  \end{aligned}
\]
because \(\sum_i\|v_i\|^2\leq\|u\|^2=1\).

The up-output sum equals the global up-energy by
\cref{lem:no-crossing}.  Consequently,
\[
  \sum_{j=1}^{t}
  \langle u_j,\Delta_ju_j\rangle
  \leq
  (2t+2)\mathcal E(u)
  +
  2C_{\mathrm{leak}}^2\lambda^2(t-1)^2.
\]
Every high-sector eigenvalue is at least \(c_{\mathrm{spec}}\).  Using \(\mathcal E(u)\leq\lambda^2\) gives the required estimate.
\end{proof}

Here \emph{coercivity} means quantitative control of the bulk norm by the selected incidence outputs together with the already controlled low- and high-energy components.

\begin{lemma}[Bulk coercivity from the selected outputs]
\label{lem:bulk-coercivity}
There are catalogue-dependent constants \(C_{\mathrm{bulk}},C_{\mathrm{sel}}>0\) such that, for every occurrence \(j\), every \(0<\lambda<\lambda_*\), and every padded degree-\(q\) cochain \(v\), \[
  \|T_j^{\mathrm{sel}}\|_{\mathrm{op}}
  \leq
  C_{\mathrm{sel}}\lambda
\]
and
\[
  \|P_j^{\mathrm{bulk}}v\|^2
  \leq
  C_{\mathrm{bulk}}
  \left(
    \lambda^{-2}\|T_j^{\mathrm{sel}}v\|^2
    +
    \|P_j^{\mathrm{high}}v\|^2
    +
    \|P_j^{\mathrm{lift}}v\|^2
    +
    \lambda^2\|v\|^2
  \right).
\]
On the unperturbed padded bulk space,
\[
  T_j^{\mathrm{sel}}
  \big|_{\overline{\mathscr B}_j}
  =
  \lambda
  \bigl(
    T_{a_j,0}\otimes I_{\Eenc_{N-m_{a_j}}}
  \bigr).
\]
\end{lemma}

\begin{proof}
Every selected output has a local factor containing the apex.  The apex has
no neighbour in the supported local base, so that factor is entirely
private; a padded simplex containing it may also contain vertices of the
untouched factors.  Terms acting on an untouched factor vanish after
projection onto its harmonic space: if \(Q_{\mathrm{unt}}\) is the
corresponding harmonic projection, then \(Q_{\mathrm{unt}}D_{\mathrm{unt}}
=Q_{\mathrm{unt}}D_{\mathrm{unt}}^*=0\). Every retained local boundary or coboundary incidence therefore deletes
or adds a private vertex and contributes one factor of \(\lambda\). Hence the selected map on the full padded space is \(\lambda\) times one of finitely many local incidence maps tensored with a
norm-one projection.  This gives
\[
  \|T_j^{\mathrm{sel}}\|_{\mathrm{op}}
  \leq
  C_{\mathrm{sel}}\lambda.
\]
Restricting to \(\overline{\mathscr B}_j\) gives the displayed tensor
identity.  Condition C4 therefore implies
\[
  \|T_j^{\mathrm{sel}}b_0\|
  \geq
  a_*\lambda\|b_0\|
  \qquad
  (b_0\in\overline{\mathscr B}_j).
\]

Put \(P:=P_j^{\mathrm{bulk}}, \, P_0:=\overline P_j^{\mathrm{bulk}}, \, T:=T_j^{\mathrm{sel}}\). If \(b\in\operatorname{im}P\), projection stability gives \(\|b-P_0b\| \leq C_{\mathrm{proj}}\lambda\|b\|\). It follows that
\[
  \|Tb\|
  \geq
  \lambda
  \left[
    a_*-(a_*+C_{\mathrm{sel}})
    C_{\mathrm{proj}}\lambda
  \right]
  \|b\|.
\]
Decrease the catalogue-wide threshold \(\lambda_*\), if necessary, so that
the bracket is at least \(a_*/2\).

For an arbitrary \(v\), decompose it into the four orthogonal sectors.
The harmonic sector is annihilated by both residuals and hence by \(T\). Writing \(b=Pv\), we obtain
\[
  \frac{a_*}{2}\lambda\|b\|
  \leq
  \|Tv\|
  +
  C_{\mathrm{sel}}\lambda
  \left(
    \|P_j^{\mathrm{high}}v\|
    +
    \|P_j^{\mathrm{lift}}v\|
  \right).
\]
Squaring and using \((x+y+z)^2\leq3(x^2+y^2+z^2)\) gives the stated estimate, even without
the additional nonnegative term \(\lambda^2\|v\|^2\).
\end{proof}

\begin{corollary}[Summed bulk estimate]
\label{cor:bulk-sector}
There is a catalogue-dependent constant \(C'_{\mathrm{bulk}}\) such that,
if \(u\) is unit and \(\mathcal E(u)\leq\lambda^2\), then
\[
  \sum_{j=1}^{t}\|P_j^{\mathrm{bulk}}u_j\|^2
  \leq
  C'_{\mathrm{bulk}}
  \left(
    \lambda^{-2}\mathcal E(u)
    +
    \lambda^{-(4M+2)}\mathcal E(u)
    +
    \lambda^2t^2
    +
    \lambda^2t
  \right).
\]
\end{corollary}

\begin{proof}
Apply \cref{lem:bulk-coercivity} to each \(u_j\) and sum.  By
\cref{lem:selected-isolation}, the selected-output contribution is at most \(\lambda^{-2}\mathcal E(u)\).  The high and lifted contributions are
bounded by \cref{lem:high-sector,lem:lifted-up}.  Absorbing the fixed
catalogue constants and using \(\sum_j\|u_j\|^2\leq t\) for the final term proves the result.
\end{proof}

\subsection{The master inequality}
\label{subsec:master-inequality}

The preceding estimates control the three nonharmonic sectors.  On the harmonic sector, the source projectors enter through
\[
  P_j^{\mathrm{har}}
  =
  J_N(I-P_j)J_N^*
  +
  O(\lambda).
\]

\begin{lemma}[Counting the repeated base component]
\label{lem:counting-gap}
For \(u_j=(\mathsf C_0+\mathsf C_j)u\), \[
  \sum_{j=1}^{t}\|u_j\|^2
  =
  \|u\|^2
  +(t-1)\|\mathsf C_0u\|^2.
\]
Let \(Q_N:=J_NJ_N^*\), extended by zero away from the base coordinates.  If \(H=\sum_{j=1}^{t}P_j\succeq gI_{\Qenc_N}\), then
\[
  \sum_{j=1}^{t}
  \langle
    u,J_N(I-P_j)J_N^*u
  \rangle
  \leq
  (t-g)\|Q_Nu\|^2,
\]
and \(\|Q_Nu\|^2 \leq \|\mathsf C_0u\|^2\).
\end{lemma}

\begin{proof}
Orthogonality of the coordinate supports gives
\[
  \|u_j\|^2
  =
  \|\mathsf C_0u\|^2+\|\mathsf C_ju\|^2.
\]
Summing and using \(I=\mathsf C_0+\sum_j\mathsf C_j\) proves the first identity.

Set \(x:=J_N^*u\).  Then
\[
  \sum_{j=1}^{t}
  \langle u,J_N(I-P_j)J_N^*u\rangle
  =
  \langle x,(tI-H)x\rangle
  \leq
  (t-g)\|x\|^2.
\]
Because \(J_N\) is an isometry, \(\|x\|=\|Q_Nu\|\).  Finally, \(\operatorname{im}Q_N\) lies in the base-coordinate space, which proves
the last inequality.
\end{proof}

\begin{proposition}[Master inequality]
\label{prop:master-inequality}
There is a computable constant \(C_*>0\), depending only on the certified
catalogue, with the following property.  Assume
\[
  H=\sum_{j=1}^{t}P_j\succeq gI_{\Qenc_N},
  \qquad
  0<g\leq1,
  \qquad
  t\geq1.
\]
Let \(0<\lambda<\lambda_*\), and let \(u\) be a unit degree-\(q\) cochain
with \(\mathcal E(u)\leq\lambda^2\).  Then
\[
  1
  \leq
  1-g
  +
  C_*
  \left(
    \lambda t
    +
    \lambda^{-(4M+2)}\mathcal E(u)
    +
    \lambda^2t^2
    +
    \lambda^2t
    +
    \lambda^{-2}\mathcal E(u)
  \right).
\]
\end{proposition}

\begin{proof}
For each \(j\), decompose \(u_j\) into its four orthogonal local sectors.
Summing their squared norms gives
\[
  \sum_{j=1}^{t}\|u_j\|^2
  =
  \sum_{j=1}^{t}
  \left(
    \|P_j^{\mathrm{har}}u_j\|^2
    +
    \|P_j^{\mathrm{lift}}u_j\|^2
    +
    \|P_j^{\mathrm{bulk}}u_j\|^2
    +
    \|P_j^{\mathrm{high}}u_j\|^2
  \right).
\]

By \cref{lem:projection-stability,lem:counting-gap}, the harmonic
contribution is at most
\[
  (t-g)\|Q_Nu\|^2
  +
  C_{\mathrm{proj}}\lambda
  \sum_{j=1}^{t}\|u_j\|^2.
\]
Since \(u\) is unit, every \(\|u_j\|\leq1\), and the error is at most \(C_{\mathrm{proj}}\lambda t\).  The lifted, high, and bulk contributions
are bounded by
\cref{lem:lifted-up,lem:high-sector,cor:bulk-sector}.  After enlarging one
catalogue-dependent constant \(C_*\), we obtain
\[
  1+(t-1)a
  \leq
  (t-g)b
  +
  C_*
  \left(
    \lambda t
    +
    \lambda^{-(4M+2)}\mathcal E(u)
    +
    \lambda^2t^2
    +
    \lambda^2t
    +
    \lambda^{-2}\mathcal E(u)
  \right),
\]
where \(a:=\|\mathsf C_0u\|^2, \, b:=\|Q_Nu\|^2\). Here \(0\leq b\leq a\leq1\).  Since \(t\geq1\) and \(g\leq1\), \[
  (t-g)b-(t-1)a
  =
  (1-g)b-(t-1)(a-b)
  \leq
  1-g.
\]
Subtracting the repeated-base term \((t-1)a\) proves the master
inequality.

All constants used above are effective.  The spectral, projection, and
singular-value bounds come from
\cref{prop:local-effectivity}; the remaining constants are operator norms
of finitely many certified incidence matrices and may be replaced by
rational Frobenius-norm bounds.  Thus \(C_*\) can be chosen as one rational
catalogue-dependent constant, independently of the source instance.
\end{proof}

\subsection{Choosing the private weight}
\label{subsec:global-calibration}

\begin{proposition}[Uniform rational calibration]
\label{prop:calibration}
Choose a positive rational constant \(c_0\) such that
\[
  c_0
  \leq
  \min
  \left\{
    \frac{\lambda_*}{2},
    1,
    \frac{1}{20C_*}
  \right\}.
\]
For a source instance satisfying \(H\succeq gI\) with \(0<g\le1\), define
\[
  \lambda:=c_0\frac{g}{t},
  \qquad
  E_0:=c_0\lambda^{4M+2}g.
\]
Then every unit degree-\(q\) cochain satisfies \(\mathcal E(u)\geq E_0\). Consequently,
\[
  \lambda_{\min}(\widehat\Delta_q(\lambda))
  \geq
  E_0
  =
  c_0^{4M+3}
  \frac{g^{4M+3}}{t^{4M+2}}.
\]
\end{proposition}

\begin{proof}
The chosen \(\lambda\) is positive and smaller than \(\lambda_*\). Suppose that a unit cochain satisfies \(\mathcal E(u)<E_0\). Since \(c_0g\leq1\), \(0<\lambda<1\), and \(4M+2\geq2\), \(E_0\leq\lambda^2\). The master inequality therefore applies.

The first four error terms obey
\[
  \lambda t=c_0g,
\]
\[
  \lambda^{-(4M+2)}\mathcal E(u)<c_0g,
\]
\[
  \lambda^{-2}\mathcal E(u)<c_0g,
\]
and
\[
  \lambda^2t^2=c_0^2g^2\leq c_0g.
\]
Finally,
\[
  \lambda^2t
  =
  \frac{c_0^2g^2}{t}
  \leq
  c_0g.
\]
The total contribution of the five error terms is at most \(5C_*c_0g\leq\frac{g}{4}\). The master inequality would then give \(1\leq1-\frac{3g}{4}\), a contradiction.  Hence every unit cochain has energy at least \(E_0\). The Rayleigh principle gives the eigenvalue bound, and substitution of \(\lambda=c_0g/t\) gives the stated formula.
\end{proof}

\begin{theorem}[Global spectral gap]
\label{thm:global-gap}
Let \(H=\sum_{j=1}^{t}P_j\) be a finite-catalogue source presentation with rational promise parameter \(0<g\leq1\).  Construct \(\widehat G(H)\), give every base vertex weight \(1\), and give every private vertex the rational weight \(\lambda=c_0\frac{g}{t}\). Then:
\begin{enumerate}[label=(\roman*)]
\item if \(\ker H\neq0\), then
\[
  H_{q_N}(\widehat Y(H);\CC)\neq0
  \qquad\text{and}\qquad
  \lambda_{\min}(\widehat\Delta_{q_N})=0;
\]
\item if \(H\succeq gI\), then
\[
  H_{q_N}(\widehat Y(H);\CC)=0
\]
and
\[
  \lambda_{\min}(\widehat\Delta_{q_N})
  \geq
  c_0^{4M+3}
  \frac{g^{4M+3}}{t^{4M+2}}.
\]
\end{enumerate}
Because \(M\leq4\), \(g\) is inverse-polynomial, and \(t\) is polynomially
bounded, the lower bound in (ii) is inverse-polynomial in the source input
length.
\end{theorem}

\begin{proof}
If \(H\psi=0\), positivity of the projectors gives
\[
  0
  =
  \langle\psi,H\psi\rangle
  =
  \sum_{j=1}^{t}\|P_j\psi\|^2.
\]
Thus \(P_j\psi=0\) for every \(j\).  For nonzero \(\psi\), \cref{cor:YES} gives the nonzero homology in (i), and finite-dimensional
Hodge theory gives the zero eigenvalue.

Under the NO hypothesis, \cref{prop:calibration} gives the positive lower
bound in (ii).  The Laplacian therefore has trivial kernel, so
finite-dimensional Hodge theory gives \(H^{q_N}(\widehat Y(H);\CC)=0\). Over \(\CC\), homology and cohomology have the same dimension.  Therefore \(H_{q_N}(\widehat Y(H);\CC)=0\).
\end{proof}

\section{The exact reduction and proof of the main theorem}
\label{sec:reduction}

The preceding sections prove the topological and spectral behavior of the
construction for a finite-catalogue Hamiltonian \(H=\sum_{j=1}^{t}P_j\) on \(N\) qubits.  To complete the reduction, we verify polynomial size and bit complexity and
express the degree \(q_N=2N-1\) and the lower bound from
\cref{thm:global-gap} using functions of the output vertex count.  Isolated
vertices implement the latter step without changing the relevant homology or
Hodge Laplacian.

\subsection{Effectivity and isolated-vertex padding}

\begin{lemma}[Construction size and bit complexity]
\label{lem:construction-size}
Let \(H=\sum_{j=1}^{t}(\Pi_{a_j})_{S_j}^{(N)}\) be a finite-catalogue source presentation with \(t\geq1\), and let \(0<g\leq1\) be rational.  From the finite list \(\bigl(N,t,(a_j,S_j)_{j=1}^{t},g\bigr)\) one can construct the adjacency list and rational vertex weights of \(\widehat G(H)\) in deterministic time polynomial in \(N+t+\operatorname{bitlength}(g)\).  In particular, this running time is
polynomial in the original source input length for the presentations
produced by \cref{cor:finite-catalogue-source}.  The number of vertices and
edges is polynomial in \(N+t\), and every
weight has bit length polynomial in \(\operatorname{bitlength}(g)+\log(t+1)\).
\end{lemma}

\begin{proof}
Let \(v_{\max} :=\max_{1\leq a\leq13}|V_a^{\mathrm{priv}}|\). This is a fixed constant.  The common bow-tie graph has \(7N\) vertices,
and every occurrence contributes at most \(v_{\max}\) fresh private
vertices.  Therefore
\[
  |V(\widehat G(H))|
  \leq 7N+v_{\max}t.
\]
Construct \(G_N\), and then process the occurrences one at a time.  For the \(j\)-th occurrence, copy the fixed private graph of type \(a_j\), relabel
its supported base factors according to \(S_j\), and join its private
vertices to the untouched bow-tie factors as prescribed in
\cref{def:padded-gadget,def:global-graph}.  This constructs the graph
directly from its edges.
The resulting graph has at most \(\binom{7N+v_{\max}t}{2}\) edges, and its adjacency list is produced by a polynomial number of elementary
operations.

The base vertices receive weight \(1\), while every private vertex receives \(\lambda=c_0\frac{g}{t}\). The constant \(c_0\in\QQ_{>0}\) is fixed once for the catalogue.  Hence \(\lambda\) is computed by exact rational arithmetic and has bit length \(O\bigl(\operatorname{bitlength}(g)+\log(t+1)+1\bigr)\). Because \(g\leq1\) and \(c_0\leq1\), all weights lie in \((0,1]\). The finite catalogue graphs, their placement rules, and \(c_0\) are
hard-coded into the reduction.
\end{proof}

\begin{lemma}[Spectrally inert isolated-vertex padding]
\label{lem:isolated-padding}
Let \(G\) be a finite positively weighted graph, and let \(G^+\) be obtained
from \(G\) by adjoining \(r\geq0\) isolated vertices with arbitrary positive
weights.  For every \(k\geq1\), there are canonical equalities
\[
  C^k(\Cl(G^+);\CC)=C^k(\Cl(G);\CC)
\]
and
\[
  \Delta_{k,w}(\Cl(G^+))
  =\Delta_{k,w}(\Cl(G)).
\]
Consequently the degree-\(k\) homology groups and the degree-\(k\) weighted
Hodge spectra agree exactly.
\end{lemma}

\begin{proof}
An isolated new vertex belongs to no simplex of positive dimension.  Thus all
degree-\(k\) simplices and all of their cofaces are unchanged.  If \(k\geq2\), their faces also lie in unchanged positive-degree spaces, so both adjacent
incidence matrices are unchanged.  If \(k=1\), the edge-to-vertex boundary
matrix acquires only zero rows, because no new vertex belongs to an edge;
its contribution to the Laplacian on the edge space is therefore unchanged.
The weights of all old simplices remain the same, so the degree-\(k\) Laplacians agree.  The same chain-complex observation gives the homology
assertion.
\end{proof}

\subsection{Universal degree and promise-gap functions}

For \(u,v\in\ZZ_{\geq0}\), let
\[
  \langle u,v\rangle
  :=\frac{(u+v)(u+v+1)}{2}+v
\]
be the Cantor pairing function.  It is a bijection \(\ZZ_{\geq0}^2\to\ZZ_{\geq0}\), and its inverse is computable exactly in
time polynomial in the binary length of its argument by an integer-square-root
calculation.

Put
\[
  M_{\mathrm{cat}}:=\max_{1\leq a\leq13}m_a=4,
  \qquad
  A:=4M_{\mathrm{cat}}+3=19,
  \qquad
  B:=8M_{\mathrm{cat}}+5=37.
\]
For every \(n\geq2\), let \((u(n),v(n))\) be the unique pair satisfying \(n-2=\langle u(n),v(n)\rangle\). Define
\[
  k(n)
  :=
  \begin{cases}
    2u(n)-1,
      &u(n)\geq1\text{ and }2u(n)-1<n,\\
    1,&\text{otherwise},
  \end{cases}
\]
and
\[
  \gamma(n)
  :=
  \frac{c_0^{A}}{(v(n)+1)^{B}}.
\]

\begin{lemma}[Validity of the universal parameters]
\label{lem:universal-parameters}
The functions \(k\) and \(\gamma\) defined above satisfy all the conditions
in \cref{def:weighted-gapped-clique-homology}.  In particular, they are fixed
independently of the source promise problem and are exactly computable in
deterministic polynomial time.
\end{lemma}

\begin{proof}
The definition gives \(1\leq k(n)<n\) for every \(n\geq2\).  Inverting the
Cantor pairing function, evaluating fixed powers, and performing rational
arithmetic require a number of bit operations polynomial in \(\log(n+1)\), and therefore certainly polynomial in \(n\).  The resulting
rational numbers have polynomial binary length.

Moreover, \(0<c_0\leq1\) and \(v(n)+1\geq1\), so \(0<\gamma(n)\leq1\). Since \(\langle u,v\rangle\geq v\), we have \(v(n)+1\leq n-1<n\).  Choose
an integer \(C\geq c_0^{-A}\).  Then
\[
  \gamma(n)
  \geq \frac{1}{C n^{B}},
\]
which is the required inverse-polynomial lower bound.
\end{proof}

\subsection{The input-to-graph reduction}

\begin{theorem}[Polynomial-time promise-preserving reduction]
\label{thm:final-reduction}
Let \(L=(L_{\mathrm{YES}},L_{\mathrm{NO}}) \in\mathsf{QMA}_1(\mathcal G)\). For the single fixed pair \((k,\gamma)\) of
\cref{lem:universal-parameters}, there is a deterministic polynomial-time
map \(x\longmapsto (G_x,w_x)\) such that
\[
  x\in L_{\mathrm{YES}}
  \quad\Longrightarrow\quad
  H_{k(|V(G_x)|)}(\Cl(G_x);\CC)\neq0,
\]
whereas
\[
  x\in L_{\mathrm{NO}}
  \quad\Longrightarrow\quad
  H_{k(|V(G_x)|)}(\Cl(G_x);\CC)=0
\]
and
\[
  \lambda_{\min}
  \!\left(\Delta_{k(|V(G_x)|),w_x}\right)
  \geq\gamma(|V(G_x)|).
\]
\end{theorem}

\begin{proof}
Fix a polynomial-time uniform verifier family for \(L\).  Apply
\cref{cor:finite-catalogue-source} to \(x\).  It produces, in polynomial
time, a finite-catalogue Hamiltonian \(H_x=\sum_{j=1}^{t_x}P_j\) on \(N_x\) qubits and a rational number \(g_x\in(0,1]\) such that
\[
  x\in L_{\mathrm{YES}}
  \Longrightarrow
  \ker H_x\neq0,
\]
and
\[
  x\in L_{\mathrm{NO}}
  \Longrightarrow
  H_x\succeq g_x I.
\]
The integers \(N_x,t_x\), the reciprocal \(g_x^{-1}\), and their binary
descriptions are polynomially bounded in \(|x|\).

Construct the weighted graph \(\widehat G(H_x)\) as in
\cref{def:global-graph}, using private weight \(\lambda_x=c_0\frac{g_x}{t_x}\), and write \(n_x^0:=|V(\widehat G(H_x))|\). Choose the integer
\[
  R_x
  :=\max
  \left\{
    n_x^0,
    t_x,
    \left\lceil g_x^{-1}\right\rceil,
    2N_x,
    1
  \right\},
\]
and set \(n_x:=2+\langle N_x,R_x\rangle\). Because \(\langle N_x,R_x\rangle\geq R_x\geq n_x^0\), we may adjoin
exactly \(n_x-n_x^0\) isolated vertices, each of weight \(1\).  Call the
resulting weighted graph \((G_x,w_x)\).  All the displayed quantities have
polynomial magnitude and polynomial bit length, so
\cref{lem:construction-size} shows that the complete map is polynomial-time.

The inverse pairing associated with \(|V(G_x)|=n_x\) returns \((u(n_x),v(n_x))=(N_x,R_x)\). Since \(R_x\geq2N_x\), the valid branch in the definition of \(k\) applies,
and therefore
\[
  k(|V(G_x)|)=2N_x-1=q_{N_x}\geq1.
\]

If \(x\in L_{\mathrm{YES}}\), then \cref{thm:global-gap}(i) gives
nonzero degree-\(q_{N_x}\) homology for \(\widehat G(H_x)\).  The isolated
vertices preserve it by \cref{lem:isolated-padding}.

Suppose instead that \(x\in L_{\mathrm{NO}}\).  By
\cref{thm:global-gap}(ii), with \(M=\max_jm_{a_j}\leq M_{\mathrm{cat}}\), the unpadded output has zero homology in degree \(q_{N_x}\) and
\[
  \lambda_{\min}(\widehat\Delta_{q_{N_x}})
  \geq
  c_0^{4M+3}\frac{g_x^{4M+3}}{t_x^{4M+2}}
  \geq
  c_0^{A}\frac{g_x^{A}}{t_x^{4M_{\mathrm{cat}}+2}},
\]
where the last inequality uses \(0<c_0,g_x\leq1\) and \(t_x\geq1\). The choice of \(R_x\) gives
\[
  g_x\geq R_x^{-1},
  \qquad
  t_x\leq R_x.
\]
Consequently
\[
  \lambda_{\min}(\widehat\Delta_{q_{N_x}})
  \geq
  c_0^{A}R_x^{-B}
  \geq
  \frac{c_0^{A}}{(R_x+1)^B}
  =\gamma(n_x).
\]
Again \cref{lem:isolated-padding} preserves both the homology and the
Laplacian in this positive degree.
\end{proof}

\begin{proof}[Proof of \cref{thm:main-hardness}]
The functions \(k\) and \(\gamma\) are the universal functions of
\cref{lem:universal-parameters}.  For every promise problem in \(\mathsf{QMA}_1(\mathcal G)\), \cref{thm:final-reduction} supplies the
required polynomial-time many-one reduction.  Hence weighted gapped clique
homology for this fixed pair is \(\mathsf{QMA}_1(\mathcal G)\)-hard.
\end{proof}

\appendix

\section {AI-assisted tools}

ChatGPT (OpenAI), using the GPT-5.6 Sol model, was used for LaTeX editing,
expository and phrasing suggestions, auxiliary checking of the arguments, and for writing
most of the Python code used to carry out the finite homology computations
and related certificate checks.  Claude Opus 5 (Anthropic) was used to
assist with the preparation and typesetting of the figures.

The mathematical results, proof strategy, certificate conditions, and
arguments were developed by the author.  The author reviewed the resulting
code and certificate outputs and takes full responsibility for the
mathematical content of the paper and its accompanying materials.

\section{Spectral-sequence conventions and edge maps}
\label{app:spectral-sequence}

\subsection{The depth filtration and its pages}

Fix one local gadget and write \(C^k=\bigoplus_{\ell=0}^{k+1}C^k_{(\ell)}\) for its orthogonal decomposition by exact private depth.  Throughout this
appendix, \(D:=\Dbase+\Dpriv\) denotes the unweighted incidence coboundary, whereas \(D_t:=\Dbase+t\Dpriv\) is its weighted rescaling on the fixed orthonormal simplex-coordinate
spaces.  The inner product stays fixed and nondegenerate at \(t=0\); all adjoints, limits, and norms below use it.

\begin{definition}[Depth-filtered cochain complex]
\label{def:depth-filtered-complex}
For every integer \(\ell\), set
\[
  U_\ell^k
  :=\bigoplus_{p\geq\ell}C^k_{(p)}.
\]
We use the endpoint conventions \(U_\ell^k=C^k\,(\ell\leq0), \, U_\ell^k=0\,(\ell>k+1)\). The full coboundary preserves this decreasing filtration: \(D(U_\ell^k)\subseteq U_\ell^{k+1}\), because adding a vertex never decreases private depth.
\end{definition}

\begin{definition}[Spectral-sequence pages]
\label{def:pages}
Page zero is
\[
  E_{0,\ell}^k
  :=U_\ell^k/U_{\ell+1}^k
  \cong C^k_{(\ell)},
\]
with differential induced by \(\Dbase\).

For \(r\geq1\), define
\[
  Z_{r,\ell}^k
  :=U_\ell^k\cap D^{-1}(U_{\ell+r}^{k+1}),
\]
\[
  B_{r-1,\ell}^k
  :=U_\ell^k\cap D(U_{\ell-r+1}^{k-1}),
\]
and put \(Z_{0,\ell}^k:=U_\ell^k\).  Then
\[
  E_{r,\ell}^k
  :=
  Z_{r,\ell}^k
  \Big/
  \bigl(
    Z_{r-1,\ell+1}^k+B_{r-1,\ell}^k
  \bigr).
\]
The page differential is induced by \(D\) and has bidegree \(\delta_r: E_{r,\ell}^k\longrightarrow E_{r,\ell+r}^{k+1}\).
\end{definition}

\begin{proposition}[Representative formula for the page differential]
\label{prop:page-representative}
Let $[\alpha]\in E_{r,\ell}^k$ be represented by
$\alpha\in Z_{r,\ell}^k$.  Then \(D\alpha\in U_{\ell+r}^{k+1}\), and the page differential is
\[
 \delta_r[\alpha]=[D\alpha]
 \in E_{r,\ell+r}^{k+1}.
\]
This formula is independent of the representative.  Its kernel modulo its
image is naturally the next page:
\[
 H(E_r,\delta_r)\cong E_{r+1}.
\]
\end{proposition}

\begin{proof}
The case $r=0$ is the associated-graded differential.  Fix $r\ge1$.
If $\alpha\in Z_{r,\ell}^k$, then $D\alpha\in U_{\ell+r}^{k+1}$
and $D(D\alpha)=0$, so $D\alpha$ represents a target page class.
If $\alpha$ is changed by $z\in Z_{r-1,\ell+1}^k$, then
$Dz\in U_{\ell+r}^{k+1}\cap D(U_{\ell+1}^k)
=B_{r-1,\ell+r}^{k+1}$, a target denominator term.
Changing $\alpha$ by a boundary does not change $D\alpha$.
Thus the differential is well defined, and its square is zero.

For the next-page assertion, suppose $[\alpha]$ is in its kernel.  The
vanishing target class means
\[
 D\alpha=a+D\beta,
 \quad a\in Z_{r-1,\ell+r+1}^{k+1},\quad
 \beta\in U_{\ell+1}^k,\quad D\beta\in U_{\ell+r}^{k+1}.
\]
Thus $\beta\in Z_{r-1,\ell+1}^k$, and replacing $\alpha$ by
$\alpha-\beta$ leaves its page class unchanged and puts its differential
in $U_{\ell+r+1}^{k+1}$.  Every cohomology class therefore has a
representative in $Z_{r+1,\ell}^k$.
Such a representative gives zero in page cohomology precisely when it
belongs to \(Z_{r,\ell+1}^k+B_{r,\ell}^k\). If it is a page boundary, write it as
$D\gamma+z+D\eta$, where
$\gamma\in Z_{r,\ell-r}^{k-1}$,
$z\in Z_{r-1,\ell+1}^k$, and
$\eta\in U_{\ell-r+1}^{k-1}$ with $D\eta\in U_\ell^k$.
Then $D(\gamma+\eta)\in B_{r,\ell}^k$; the remaining $z$ has
$Dz\in U_{\ell+r+1}^{k+1}$ and hence lies in $Z_{r,\ell+1}^k$.
Conversely, the first displayed denominator summand is already zero on
page $r$, and every element $D\gamma\in B_{r,\ell}^k$ is the image
of $[\gamma]\in E_{r,\ell-r}^{k-1}$.
Taking the quotient is exactly the definition of $E_{r+1,\ell}^k$.
\end{proof}

\begin{proposition}[Finite convergence and the depth-zero edge map]
\label{prop:edge-map}
For the private-depth filtration in degree $k$, one has
$U_{k+2}^k=0$; hence the spectral sequence stabilizes after finitely many
pages.  If $k\ge1$ and $i:X_m\hookrightarrow Y_a^{\rm loc}$ is the base inclusion, the
limiting depth-zero term is canonically
\[
 E_{\infty,0}^k
 \cong\im\bigl(i^*:H^k(Y_a^{\rm loc};\CC)
                    \longrightarrow H^k(X_m;\CC)\bigr).
\]
Thus this term is the image of restriction to base cohomology.
\end{proposition}

\begin{proof}
The complex is finite, so its degree and depth indices range over finite
sets.  For each fixed position, once $r$ is sufficiently large the
numerator and both denominator terms in \cref{def:pages} stop changing.
The page differential then has zero source or target outside the finite
depth range, so the pages stabilize.
For the edge-map assertion, fix $k\ge1$ and take $r\ge k+3$.  The formula becomes
\[
 E_{r,0}^k
 =\frac{\ker D_k}
 {(U_1^k\cap\ker D_k)+\im D_{k-1}}.
\]
C2 identifies the quotient complex $C^\bullet(Y_a^{\rm loc})/U_1^\bullet$
with the base cochain complex.  Restriction $i^*$ therefore sends a full
cocycle to a base cocycle and defines a map from the displayed quotient
onto $\im i^*$.
If the restricted cocycle is a base coboundary $D_Xb$, extend $b$ by zero
on private coordinates.  Subtracting its full coboundary produces a
cohomologous full cocycle whose restriction is zero, that is, an element
of $U_1^k\cap\ker D_k$.  Conversely, each denominator term has zero restricted cohomology class.
Thus the displayed quotient is canonically identified with $\im i^*$ as a
subspace of base cohomology.
\end{proof}

\begin{remark}[Index dictionary]
The superscript $k$ is cochain degree, the second subscript $\ell$ is private
depth, and the first subscript $r$ is the page.  With the convention of
\cref{def:pages}, \(\delta_r:E_{r,\ell}^k\to E_{r,\ell+r}^{k+1}\). Here $D_0=\Dbase$.
\end{remark}

\subsection{Proof of the finite-dimensional filtered Hodge correspondence}
\label{app:filtered-hodge-proof}

We use Rellich's theorem in the following finite-dimensional form; see
\cite[Theorem~1.1]{RellichStatement}.

\begin{proposition}[Finite-dimensional Rellich theorem]
\label{prop:rellich}
Let \(A(t)\) be a real-analytic family of finite-dimensional Hermitian
matrices for real \(t\) near \(0\).  After restricting the parameter
interval, there are real-analytic eigenvalue functions and a
real-analytic orthonormal eigenbasis for \(A(t)\).
\end{proposition}

For the proof of \cref{thm:filtered-hodge}, fix all vector spaces and inner
products and write $t$ for the real parameter.  A vector-valued analytic
germ is a convergent power series near $t=0$; it is $O(t^r)$ when its
coefficients in degrees $0,\ldots,r-1$ vanish.  The analytic decomposition
uses \cref{prop:rellich}.

\paragraph{Step 1: the rescaled complex and its adjoint.}
On $C^k_{(\ell)}$ define $R_t^k=t^\ell I$.  For $t\ne0$,
\[
 D_{t,k}=R_t^{k+1}D_k(R_t^k)^{-1}=D_{0,k}+tD_{1,k},
 \qquad D_0=\Dbase,\quad D_1=\Dpriv.
\]
Indeed, the two summands increase depth by zero and one, respectively.
Thus $D_t^2=0$; alternatively, separating depths in $D^2=0$ gives
$D_0^2=D_1^2=D_0D_1+D_1D_0=0$.
Taking adjoints in the fixed inner product gives, for real $t\ne0$,
\[
 D_{t,k}^*=(R_t^k)^{-1}D_k^*R_t^{k+1}
          =D_{0,k}^*+tD_{1,k}^*.
\]
The polynomial formulas define both operators at $t=0$.  In degree $k$
the two residuals are $D_{t,k}u$ and $D_{t,k-1}^*u$.

\paragraph{Step 2: an analytic orthonormal decomposition into differential pairs.}
Apply Rellich's theorem to $D_{t,k}^*D_{t,k}$ in every degree $k$.
Let $e_i(t)$ be an analytic unit eigenvector with an eigenvalue
$\eta_i(t)$ that is not identically zero.  Positivity for positive and
negative real $t$ implies
\[
 \eta_i(t)=t^{2s_i}a_i(t)^2,
 \qquad s_i\in\ZZ_{\ge0},\quad a_i(0)>0,
\]
where $a_i$ is real analytic and positive near zero.  Set
\[
 \sigma_i(t):=t^{s_i}a_i(t),\qquad
 f_i(t):=\sigma_i(t)^{-1}D_{t,k}e_i(t)\quad(t\ne0).
\]
This quotient has an analytic extension to zero: the squared norm of its
numerator is $\eta_i(t)$, so every analytic coordinate of that numerator
vanishes to order at least $s_i$.  After shrinking the parameter interval,
all nonzero branches are positive for $t\ne0$, and
\[
 D_te_i=\sigma_i f_i,\qquad D_t^*f_i=\sigma_i e_i,
 \qquad D_t^*e_i=0,\qquad D_tf_i=0.
\]
In these formulas the degree determines the appropriate component of $D_t$
or $D_t^*$.  For example, $D_t^*e_i=0$ follows from
$e_i\in\im D_{t,k}^*\subseteq\ker D_{t,k-1}^*$ for $t\ne0$.
All identities extend to zero by analyticity.

The $f_i$ in each target degree are orthonormal, since
\[
 \langle D_te_i,D_te_j\rangle
 =\langle e_i,D_t^*D_te_j\rangle
 =\eta_j\delta_{ij}.
\]
Incoming vectors $f_i$ and outgoing vectors $e_j$ in the same degree are
orthogonal because $D_t^2=0$.  Their orthogonal complement has an analytic
orthonormal frame $h_b(t)$.  To see this directly, let $P(t)$ be the
analytic complementary projection, project an orthonormal basis of $\im P(0)$ by
$P(t)$, and multiply its column matrix by the inverse square root of its
Gram matrix.  That Gram matrix equals the identity at zero, so this
normalization is analytic nearby.  For $t\ne0$ the complement is
$\ker D_t\cap\ker D_t^*$; consequently \(D_th_b(t)=D_t^*h_b(t)=0\) throughout the interval.  We have therefore obtained an analytic
orthonormal basis consisting of elementary pairs
$e_i\xrightarrow{\,\sigma_i\,}f_i$ and harmonic vectors $h_b$.
The corresponding Laplacian eigenvalues are $\sigma_i^2$ at both ends of
each pair, and zero on the harmonic vectors.

\paragraph{Step 3: leading cycle and boundary spaces.}
For $r\ge1$, define two subspaces of the fixed degree-$k$ space:
\[
 \mathcal A_r^k
 :=\{u(0):u(t)\text{ analytic},\ D_tu(t)=O(t^r)\},
\]
\[
 \mathcal B_r^k
 :=\left\{\left.(t^{-(r-1)}D_tw(t))\right|_{t=0}:
 \begin{array}{l}
 w(t)\text{ analytic of degree }k-1,\\
 D_tw(t)=O(t^{r-1})
 \end{array}\right\}.
\]
The divisibility condition makes the displayed evaluation well defined.
Since $D_t^2=0$, one has $\mathcal B_r^k\subseteq\mathcal A_r^k$.
Only finitely many Taylor coefficients are needed for either definition.
Likewise, the polynomial corrections in
\cref{def:approximate-harmonic} can be replaced by analytic germs: truncate
any such germ after order $r-1$, changing each residual only by $O(t^r)$.

Expand an analytic germ in the basis from Step 2.  On an outgoing
coordinate of exponent $s_i<r$, the condition $D_tu=O(t^r)$ forces its
coefficient to vanish at zero.  Incoming and harmonic coordinates are
unrestricted.  For the boundary space, precisely the incoming coordinates
with $s_i<r$ can have nonzero constant term: the choice
$w=t^{r-1-s_i}a_i(t)^{-1}e_i(t)$ produces the leading vector $f_i(0)$.
Coordinates with $s_i\ge r$ have zero leading boundary.  Hence
\[
 \begin{aligned}
 \mathcal A_r^k
 &=\Span\{h_b(0),\ f_i(0)\text{ incoming},\
                    e_i(0)\text{ outgoing with }s_i\ge r\},\\
 \mathcal B_r^k
 &=\Span\{f_i(0)\text{ incoming with }s_i<r\},\\
 \mathcal V_r^k
 &=\Span\{h_b(0),\ f_i(0)\text{ incoming with }s_i\ge r,\
                    e_i(0)\text{ outgoing with }s_i\ge r\}.
 \end{aligned}
\]
The last equality uses both residuals: $D_t^*u=O(t^r)$ excludes precisely
the remaining incoming coordinates of exponent less than $r$.
Thus
\[
 \mathcal A_r^k=\mathcal B_r^k\mathbin{\overset{\perp}{\oplus}}\mathcal V_r^k,
 \qquad
 \mathcal A_r^k/\mathcal B_r^k
 \xrightarrow{\ \cong\ }\mathcal V_r^k,
 \quad [u_0]\longmapsto\Proj_{\mathcal V_r^k}u_0.
\]
These spaces and the orthogonal projection are defined independently of
the analytic basis used to prove the identities.

\paragraph{Step 4: identify this quotient with the specified filtered pages.}
Write $\pi_\ell:C^k\to C^k_{(\ell)}$ for the exact-depth projection.  For
$z\in U_\ell^k$ set
\[
 \mathcal R_\ell z(t)
 :=\sum_{j\ge0}t^j\pi_{\ell+j}z=t^{-\ell}R_tz.
\]
The sum is finite.  For every integer $\ell$, including negative boundary
indices, only nonnegative powers occur because $z\in U_\ell^k$.
Depth homogeneity gives \(D_t\mathcal R_\ell z=t^{-\ell}R_tDz\). Consequently $z\in Z_{r,\ell}^k$ precisely when
$D_t\mathcal R_\ell z=O(t^r)$, and
$\mathcal R_\ell z(0)=\pi_\ell z$.
The two denominator terms in \cref{def:pages} have the following meanings:
\[
 \begin{array}{c|c}
 \text{filtered representative}&\text{rescaled representative}\\ \hline
 z\in Z_{r-1,\ell+1}^k
   &\mathcal R_\ell z=t\,\mathcal R_{\ell+1}z\\[1mm]
 D\beta\in U_\ell^k,\quad\beta\in U_{\ell-r+1}^{k-1}
   &\mathcal R_\ell(D\beta)
        =t^{-(r-1)}D_t\mathcal R_{\ell-r+1}\beta.
 \end{array}
\]
The first row has zero constant coefficient; the second gives an element
of $\mathcal B_r^k$ after evaluation at zero.  Thus there is a well-defined
map
\[
 \bigoplus_\ell E_{r,\ell}^k
 \longrightarrow\mathcal A_r^k/\mathcal B_r^k,
 \qquad ([z_\ell])_\ell\longmapsto
       \left[\sum_\ell\pi_\ell z_\ell\right].
\]

For surjectivity, take a Taylor
monomial $t^jv$ with $v\in C^k_{(p)}$ and call $p-j$ its diagonal index.
The operator $D_t=D_0+tD_1$ preserves this index.  Grouping the Taylor
coefficients of $u(t)$ by diagonal index therefore decomposes
$D_tu=O(t^r)$ into independent conditions.  The component of index $\ell$
is exactly $\mathcal R_\ell z_\ell$ for a finite vector
$z_\ell\in Z_{r,\ell}^k$.  Components with zero constant coefficient
contribute nothing to $\mathcal A_r^k$.  Only finitely many diagonals can
have a nonzero constant coefficient, since the depth range is finite.

For injectivity, suppose that the leading coefficient of a numerator
belongs to $\mathcal B_r^k$.  Choose a boundary germ realizing it and split
that germ by the same diagonal indices.  In index $\ell$ the second row of
the table supplies $D\beta_\ell\in B_{r-1,\ell}^k$ with that same leading
coefficient.  Subtract it from $z_\ell$.  The difference still lies in
$Z_{r,\ell}^k$ and has zero depth-$\ell$ coefficient, hence lies in
$Z_{r-1,\ell+1}^k$.  This is exactly the other denominator term, proving
injectivity.  Equivalently, an analytic numerator with zero constant term
is $t$ times a germ whose residual is $O(t^{r-1})$.

The same diagonal decomposition shows that both $\mathcal A_r^k$ and
$\mathcal B_r^k$ split into exact-depth subspaces.  Since those depth
subspaces are orthogonal, so does
$\mathcal V_r^k=\mathcal A_r^k\cap(\mathcal B_r^k)^\perp$.
Combining the quotient isomorphism with Step 3 gives
\[
 [z]\longmapsto\Proj_{\mathcal V_r^k}\pi_\ell z
 \quad\text{from }E_{r,\ell}^k
 \text{ onto }\mathcal V_r^k\cap C^k_{(\ell)}.
\]
At $r=0$ this is the orthogonal depth decomposition of $C^k$.

\paragraph{Step 5: eigenvalue counts and limiting projections.}
Step 2 is an exact analytic orthonormal eigenbasis for $\Delta_k(t)$.
Step 3 identifies $\mathcal V_r^k$ as the span of the values at zero of
exactly those basis vectors whose eigenvalue is $O(t^{2r})$, including
identically zero eigenvalues; hence
\[
 P_{\ge r}^k(0)=\Proj_{\mathcal V_r^k}.
\]
These projections are analytic finite sums of rank-one projections, so
their difference from the value at zero is $O(|t|)$ in operator norm.
Subtracting the projections for $r+1$ from those for $r$ proves both the
multiplicity formula and the exact-order limiting-subspace formula.
For $r$ greater than all exponents, $\mathcal V_r^k$ is the span of the
$h_b(0)$ alone, which gives the limiting zero eigenspace.  The same
finite-sum argument applies to any fixed analytic branch set $J$.

For the gadgets, the representative map above identifies the encoded and
bulk limiting subspaces in \cref{thm:local-rank-one,lem:projection-stability}.
Uniform estimates for padded families are supplied by
\cref{thm:spectral-padding}.

\section{Exact finite-certificate protocols}
\label{app:certificate-protocols}

\subsection{Matrix encodings}

\begin{definition}[Ordered simplex list and boundary matrix]
For a finite complex $K$, an \emph{ordered simplex list} (called a manifest in the companion files) is a list of every
oriented simplex in each degree, using increasing vertex order.  The integer
boundary matrix \(B_r(K)\in\ZZ^{|K_{r-1}|\times|K_r|}\) has entry $(-1)^i$ when the row simplex is obtained from the column simplex
by deleting its $i$-th vertex, and entry $0$ otherwise.  For a pair
$L\subseteq K$, the relative matrix \(B_r(K,L):C_r(K,L;\ZZ)\to C_{r-1}(K,L;\ZZ)\) is obtained by retaining only rows and columns indexed by simplices outside
$L$ and applying the quotient convention.
\end{definition}

\begin{definition}[Exact Betti-number certificate]
Let $n_r(K,L):=\dim_{\CC}C_r(K,L;\CC)$.  Once the exact ranks of the adjacent
relative boundary matrices are known, the relative Betti number is
\[
 \beta_r(K,L)
 =n_r(K,L)-\rank B_r(K,L)-\rank B_{r+1}(K,L).
\]
A Betti-number certificate consists of the simplex manifest, the two sparse
integer matrices, and exact rank witnesses for every degree in the claimed
range.
\end{definition}

\begin{lemma}[Finite-field exactness transfer]
\label{lem:rank-transfer}
Let \(\ZZ^a\xrightarrow{A}\ZZ^n\xrightarrow{B}\ZZ^b\) be integer matrices satisfying $BA=0$.  If for some prime $p$
\[
 \rank_{\FF_p}(A\bmod p)+\rank_{\FF_p}(B\bmod p)=n,
\]
then the rational complex is exact at $\QQ^n$:
\[
 \im(A\otimes\QQ)=\ker(B\otimes\QQ).
\]
In particular,
\[
 \rank_{\QQ}A=\rank_{\FF_p}(A\bmod p),
 \qquad
 \rank_{\QQ}B=\rank_{\FF_p}(B\bmod p).
\]
For a rectangular matrix $T$, a nonzero exact minor of size equal to its
number of columns certifies full column rank over $\QQ$ and hence over
$\CC$.
\end{lemma}

\begin{proof}
A nonzero minor modulo $p$ is an integer determinant not divisible by $p$,
hence is nonzero over $\QQ$.  Therefore each rational rank is at least its
rank modulo $p$.  On the other hand $BA=0$ implies
$\im A\subseteq\ker B$ over $\QQ$, so the rational ranks sum to at most $n$.
The assumed modular sum equals $n$, forcing equality in both rank lower
bounds and in the rational rank sum.  The inclusion is therefore an equality
by dimension.  Integer matrix ranks are unchanged when extending scalars
from $\QQ$ to $\CC$, since the same minors determine them.
The identical nonzero-minor argument proves the final assertion about $T$.
\end{proof}

\subsection{Certificate interfaces C0-C4}

\begin{proposition}[C0 verification interface]
\label{prop:C0-interface}
For each catalogue type $a$, construct the actual clique complex
$Y_a^{\rm loc}=\Cl(G_a^{\rm loc})$, enumerate all simplices of
$Y_a^{\rm loc}\setminus X_{m_a}$, and form every matrix \(B_r(Y_a^{\rm loc},X_{m_a})\). Exact rank certificates establish
\[
 \beta_r(Y_a^{\rm loc},X_{m_a})
 =\begin{cases}1,&r=2m_a,\\0,&r\ne2m_a.
 \end{cases}
\]
The simplex list must be generated after clique completion, so mixed cliques are
included automatically.
\end{proposition}

\begin{proof}[Justification of the verification criterion]
Relative chains have exactly the nonbase simplex basis stated above.
Deleting a vertex and discarding base faces is their quotient boundary;
its square is zero because it is induced from the simplicial boundary.
Consequently rank-nullity gives
$\dim H_r=n_r-\rank B_r-\rank B_{r+1}$.
The ranks are computed on this enumerated clique complex.  When they are
obtained modulo a prime, \cref{lem:rank-transfer} transfers vanishing to
characteristic zero.  Once all but one degree vanish, the integer Euler
characteristic $\sum_r(-1)^rn_r$ determines the remaining Betti number.
A nonzero modular Betti number alone gives only an upper bound on its
characteristic-zero counterpart.
The computation covering the fixed catalogue is recorded in
\cref{thm:finite-local-certification}.
\end{proof}

\begin{proposition}[C1 verification interface]
\label{prop:C1-interface}
For each type $a$, encode $c_a$, $j_a$, and $\Gamma_a$ as rational column
vectors in their ordered simplex manifests, encode $(f_a)_\#$ and
$(q_a)_\#$ as normalized signed integer chain matrices, and verify the
following equalities by exact rational arithmetic (or after clearing
denominators):
\[
 B_{q_{m_a}}(K_a)c_a=0,
 \qquad
 (f_a)_\#c_a=j_a,
 \qquad
 B_{q_{m_a}+1}(Y_a^{\rm loc})\Gamma_a=j_a.
\]
Every orientation sign and every degenerate folded simplex is thereby
decided explicitly.
\end{proposition}

\begin{proof}[Justification of the verification criterion]
The ordered simplex bases give injective coordinate identifications of
chains with rational vectors.  The normalized folding matrix assigns to a
simplex its signed image, or zero if vertices repeat.  Matrix multiplication
therefore computes exactly its chain pushforward.  The three displayed
zero differences are equivalent to the three signed chain identities.
Checking support in the source and actual output simplex lists ensures
that all vectors are chains of the claimed complexes.  Exact rational
arithmetic, or integer arithmetic after multiplying by a common denominator,
proves the equalities over $\QQ$ and hence over $\CC$.
\end{proof}

\begin{proposition}[C2 verification interface]
\label{prop:C2-interface}
For each type $a$, compare the following induced subgraph with the base graph,
edge by edge: \(G_a^{\rm loc}[V(G_{m_a})] \,\text{and}\, G_{m_a}\). Equality certifies
\[
 Y_a^{\rm loc}[V(X_{m_a})]=X_{m_a}.
\]
For every global instance, cross-gadget isolation follows from the
edge-set formula in \cref{def:global-graph}.
\end{proposition}

\begin{proof}
A subset of base vertices is a simplex of the induced flag complex exactly
when all its pairs are edges of the induced graph.  Edge equality therefore
implies equality of the induced complexes.  In the global edge formula no
summand contains vertices from two private sets, so no such edge is added.
A clique containing vertices from both sets would require such an edge,
which proves cross-gadget isolation.
\end{proof}

\begin{proposition}[C3 verification interface]
\label{prop:C3-interface}
For each type $a$, degree $k$, and intermediate depth
$1\le\ell\le k$, form the fixed-depth integer matrices
\[
 B_{k-1,\ell}^{\rm base}:
 C^{k-1}_{(\ell)}\to C^k_{(\ell)},
 \qquad
 B_{k,\ell}^{\rm base}:
 C^k_{(\ell)}\to C^{k+1}_{(\ell)}
\]
representing $\Dbase$.  Verify \(B_{k,\ell}^{\rm base}B_{k-1,\ell}^{\rm base}=0\) and certify
\[
 \rank B_{k-1,\ell}^{\rm base}
 +\rank B_{k,\ell}^{\rm base}
 =\dim C^k_{(\ell)}.
\]
By 
\cref{lem:rank-transfer}, these identities establish
$E_{1,\ell}^k=0$ exactly over $\CC$.
\end{proposition}

\begin{proof}[Justification of the verification criterion]
Since $\Dbase$ preserves depth and squares to zero, its incoming image lies
in its outgoing kernel.  The asserted adjacent rank sum says these spaces
have equal dimension, so they coincide.  Their quotient, the page-one term,
is zero.  If the test is performed modulo a prime, the exact integer
composition identity and \cref{lem:rank-transfer} prove the same statement
over $\QQ$, then $\CC$.  There are only finitely many degrees and depths
because each local complex is finite.
\end{proof}

\begin{proposition}[C4 verification interface]
\label{prop:C4-interface}
For each type $a$:
\begin{enumerate}[label=(\roman*)]
\item enumerate the pair $(P_a,A_a)$ and certify its relative Betti numbers
from the integer relative boundary matrices;
\item list an exact relative cocycle vector representing $[\kappa_a]$ and
verify that it generates the unique nonzero relative cohomology group;
\item construct the signed incidence matrix $T_{a,0}$ from the ordered
simplex lists and certify full column rank by an exact nonzero minor or the
finite-field transfer criterion;
\item verify the three apex-incidence assertions stated in
C4 directly from the simplex manifest.
\end{enumerate}
Full column rank and rational root-separation bounds give an explicit
positive rational lower bound on $s_{\min}(T_{a,0})$.
\end{proposition}

\begin{proof}[Justification of the verification criterion]
A private simplex lies in $A_a$ precisely when the common intersection of
its base-neighbour sets is nonempty.  This test is closed under taking
faces, so $A_a$ is a subcomplex.  Retaining private simplices outside $A_a$
and discarding faces in $A_a$ gives exactly the relative boundary matrices.
The Betti formula, modular vanishing transfer, and Euler argument of
\cref{prop:C0-interface} therefore establish the claimed relative homology.
Transposition and duality over $\CC$ give the cohomology statement.
An explicit rational generator can be obtained by row reduction: choose a
vector in $\ker D_{2m}$ outside $\im D_{2m-1}$, and check that adjoining
it increases the latter matrix's rank by one.  Such a choice exists once
the computed Betti number is one.  The companion package instead gives a
direct pairing certificate: it lists a relative cycle $w_a$ and a unit
coordinate cochain $\kappa_a$ supported on a private simplex with no private
coface.  It checks
\[
 \partial_{\rm rel}w_a=0,\qquad
 D_{\rm rel}\kappa_a=0,\qquad \kappa_a(w_a)\ne0.
\]
If $\kappa_a=D_{\rm rel}\eta$, then
$\kappa_a(w_a)=\eta(\partial_{\rm rel}w_a)=0$, a contradiction.
Thus $[\kappa_a]$ is nonzero and generates the already verified
one-dimensional group.  All cycle coefficients and pairings are rational
and supplied explicitly in \path{source_core_witnesses.json}.

The selected row sets are exactly the apex-containing output simplices.
The checked equality of the bulk column sets with the relative degree-$d$
basis identifies the computed incidence matrix with $T_{a,0}$.
Full column rank modulo a prime, or a nonzero maximal minor, implies full
column rank over $\CC$ by \cref{lem:rank-transfer}.  Thus $T_{a,0}^*T_{a,0}$
is positive definite.  A rational lower singular-value bound can, for
example, be found by testing $T_{a,0}^TT_{a,0}-a^2I\succ0$ for
$a=2^{-j}$ with exact rational principal-minor tests, until one succeeds.
Positivity guarantees termination.  The three apex assertions are checked
by finite set and adjacency comparisons.
\end{proof}

\subsection{Certified catalogue record}
\label{app:certified-catalogue-record}

The companion archive described in \cref{app:reproducibility} contains the
complete finite data and deterministic verification suite.
Its thirteen labelled targets are
\[
\begin{gathered}
  00,\ 11,\ 011,\ 1100,\ 0111,\ 0001,\\
  \text{five four-qubit difference vectors},\qquad
  \text{two Pythagorean vectors}.
\end{gathered}
\]
For each target, \path{source_core_witnesses.json} gives the integer
source-cycle coordinates, their nonzero rational normalization denominator,
the total source-vertex order, the top and folded-bottom vertex maps,
the recorded base arm swaps, and an exact private cycle-cocycle pairing.
The graph-generation scripts and verification code determine all graph edges, base
vertices, private vertices and apex labels.

Instructions for running the verification suite are given in
\cref{app:reproducibility}.
The thirteen C1 records verify
\[
  \partial c_a=0,\qquad
  (f_a)_\#c_a=j_a,\qquad
  \partial\Gamma_a=j_a
\]
by exact rational arithmetic, including every orientation sign and every
simplex made degenerate by folding.

The column-reduction routines in \path{gf2persist.py} and
\path{gfp_persist.py} use the clearing optimization of Chen and
Kerber~\cite{ChenKerber2011}: columns paired in one dimension are cleared
before the next lower dimension is processed.

All rank computations used for a mathematical conclusion are exact.  The
bulk matrices have the following certified sizes and ranks over \(\FF_p\), with \(p=2147483647\): \[
\begin{array}{lrrr}
\toprule
\text{private-graph type}
  &\text{rows}&\text{columns}&\rank_{\FF_p}T_{a,0}\\
\midrule
\text{computational basis, }m=2 &40&32&32\\
\text{computational basis, }m=3 &304&192&192\\
\text{computational basis, }m=4 &2048&1024&1024\\
\text{Pythagorean}              &14153&9538&9538\\
\text{four-qubit difference}    &12960&6656&6656\\
\bottomrule
\end{array}
\]
Thus each \(T_{a,0}\) has full column rank over \(\QQ\), and hence over \(\CC\), by \cref{lem:rank-transfer}.  The suite also verifies the exact
relative Betti numbers in C0 and C4, induced-base equality in C2,
fixed-depth exactness in C3, every apex-incidence assertion, and an explicit
rational cycle-cocycle pairing for each of the thirteen private core
classes.

Verification uses exact arithmetic; printed floating-point singular
values serve as diagnostics.  Positive rational bounds for the reduction
are computed from the certified integer matrices by
\cref{prop:local-effectivity}.

The rank certificates are reproduced by exact recomputation.
Characteristic-zero conclusions from modular ranks use
\cref{lem:rank-transfer}.
Simplex lists and incidence matrices are generated deterministically from
the graph-generation scripts.  The source-cycle verification script reconstructs each
graph from its source, prism and cone; the independent connecting-map
verifier constructs a folded filling and checks its boundary and simplex
membership in the actual clique complex.

\section{Signed-chain details for the folded gadgets}
\label{app:diagnosis}

The calculations below use the folding maps of \cref{constr:local-gadget}.
The corresponding entangled gadget families are those of
\cite[Theorems~8.7 and~8.8]{KingKohler2026}; graph-builder provenance is
recorded in \cref{app:reproducibility}.

\subsection{Pythagorean orientation corrections}
\label{subsec:diagnosis-signs}

For a rank-one projector onto \(|\phi\rangle=\sum_b a_b|b\rangle\), it is not enough that the folded facets have the same unoriented support as
the cycles appearing in the sum.  The induced chain map must satisfy the
signed equality
\[
  f_\#c=\sum_ba_bz_b.
\]
The sign of each image simplex depends on its orientation after the vertex
map, and a simplex with repeated image vertices contributes zero in the
normalized chain complex.

With the fixed bow-tie orientation convention, the two Pythagorean
folding maps before the orientation corrections have coefficient patterns
\[
  -5|011\rangle+4|100\rangle-3|101\rangle
\]
and
\[
  +5|010\rangle-4|101\rangle+3|100\rangle,
\]
whereas the required targets are
\[
  -5|011\rangle+4|100\rangle+3|101\rangle
\]
and
\[
  -5|010\rangle-4|101\rangle+3|100\rangle.
\]
These discrepancies are not common overall signs.  Moreover, the two
Pythagorean patterns cannot both be brought to the required orientation by one shared reorientation
of the encoded basis cycles on a common qubit triple.

Each bow-tie has automorphisms that exchange the two arms of either loop
and reverse that loop's orientation.  Postcomposing the folding maps with
the appropriate arm-swap automorphisms gives the required signed targets
while leaving the source Hamiltonian fixed.  The resulting local graph is isomorphic to the graph before the
orientation correction, so the relative homology and spectral certificates
are unchanged.  The certificate package checks the signed identities for
both Pythagorean gadgets and all five four-qubit difference gadgets.

\section{Source-sphere representatives}
\label{app:source-sphere-existence}

This appendix proves \cref{prop:source-sphere-existence}.  Write $B$ for
the bow-tie complex, so \(X_m=B*\cdots*B\) ($m$ factors) and $d=2m-1$.

\begin{proof}[Proof of \cref{prop:source-sphere-existence}]
Since $\dim B=1$ and $\dim(A*A')=\dim A+\dim A'+1$, one has
$\dim X_m=d$.  Thus $C_{d+1}(X_m;\ZZ)=0$ and
$B_d(X_m;\ZZ)=0$, so the map
$Z_d(X_m;\ZZ)\to H_d(X_m;\ZZ)$ is an isomorphism.

For $m\ge2$, repeated use of the connectivity formula for joins shows that
$X_m$ is $(2m-2)$-connected; see
\cite[Section~1.2, Exercise~21, and Theorems~3B.5 and~4.32]{Hatcher}.
Hurewicz therefore gives \(\pi_d(|X_m|)\cong H_d(X_m;\ZZ)\).  Hence the class of $j$ is represented
by a map $g:S^d\to|X_m|$ with $g_*[S^d]=[j]$.  For $m=1$, the same
conclusion follows from the surjection
$\pi_1(|X_1|)\to H_1(X_1;\ZZ)$.

By simplicial approximation
\cite[Theorem~2C.1]{Hatcher}, after finitely many barycentric subdivisions
there are a finite triangulation $K$ of $S^d$ and a simplicial map
$f:K\to X_m$ homotopic to $g$.  Orienting $K$ gives
$[f_\#[K]]=[j]$ in $H_d(X_m;\ZZ)$.  The difference $f_\#[K]-j$ is
therefore a $d$-boundary.  Since $C_{d+1}(X_m;\ZZ)=0$, there are no nonzero
$d$-boundaries, and hence $f_\#[K]=j$ as chains.  This passage from
homologous cycles to equal chains uses the top-dimensionality of $X_m$.
Any total order on the finite set $V(K)$ completes the required data.
\end{proof}

\begin{remark}[Existence versus explicit data]
\label{rem:source-sphere-effectivity}
The proposition gives no useful size bound on the triangulation.  The
thirteen certified gadgets instead use the finite source complexes of
\cref{def:source-cycle-data}, which need not be spheres.  Their certificates
specify $K_a$, $c_a$, a vertex order, and $f_a$, and verify
$\partial c_a=0$ and $(f_a)_\#c_a=j_a$ exactly.
\end{remark}

\section{Source-catalogue verification}
\label{app:catalogue-closure}

This appendix lists the clock terms of \cref{thm:cited-source-reduction},
proves \cref{lem:catalogue-closure}, and counts the terms in
\cref{prop:source-catalogue-identification}.

\subsection{The clock terms}

Following \cite[Appendix~A]{CrichignoKohler}, the clock Hamiltonian of a
circuit with $L$ gates is
$H_{\mathrm{clock}}=\sum_{r=1}^{6}H^{(r)}_{\mathrm{clock}}$, where
\[
\begin{aligned}
  H^{(1)}_{\mathrm{clock}}&=|u\rangle\langle u|_1,
  \qquad
  H^{(2)}_{\mathrm{clock}}=|d\rangle\langle d|_L,\\
  H^{(3)}_{\mathrm{clock}}&=\sum_{1\le j<k\le L}
     \bigl(|a_1\rangle\langle a_1|+|a_2\rangle\langle a_2|\bigr)_j
     \otimes
     \bigl(|a_1\rangle\langle a_1|+|a_2\rangle\langle a_2|\bigr)_k,\\
  H^{(4)}_{\mathrm{clock}}&=\sum_{1\le j<k\le L}
     \bigl(|u\rangle\langle u|+|a_1\rangle\langle a_1|
           +|a_2\rangle\langle a_2|\bigr)_j
     \otimes|d\rangle\langle d|_k,\\
  H^{(5)}_{\mathrm{clock}}&=\sum_{1\le j<k\le L}
     |u\rangle\langle u|_j\otimes
     \bigl(|a_1\rangle\langle a_1|+|a_2\rangle\langle a_2|
           +|d\rangle\langle d|\bigr)_k,\\
  H^{(6)}_{\mathrm{clock}}&=\sum_{j=1}^{L-1}
     |d\rangle\langle d|_j\otimes|u\rangle\langle u|_{j+1}.
\end{aligned}
\]

\begin{remark}[Index ranges]
\label{rem:clock-index-ranges}
The printed sums in \cite[(A.4)-(A.6)]{CrichignoKohler} run over
$1\le j\le k\le L$.  For $H^{(4)}$ and $H^{(5)}$ the diagonal terms
$j=k$ vanish, since they are products of orthogonal projections on the
same particle.  For $H^{(3)}$ the diagonal term would be
$\bigl(|a_1\rangle\langle a_1|+|a_2\rangle\langle a_2|\bigr)_j$, which
penalizes every clock state with an active particle; including it would
give positive energy to every history state and destroy completeness.
The constraint it encodes, that at most one particle is active, concerns
distinct particles, so the sum runs over $j<k$.  Likewise
$H'_{\mathrm{prop},t}$ involves particles $t$ and $t+1$ and is present only
for $t\le L-1$.
\end{remark}

\subsection{Proof of the catalogue-closure lemma}

By the placement rule of \cref{subsec:projector-catalogue}, it suffices to
write each projection as a sum of pairwise orthogonal rank-one projections
onto vectors, each a nonzero multiple of a catalogue vector written onto a
reordering of its qubits.  Throughout, strings are written in the order of
the qubits on which the summand acts.

\emph{Computational-basis terms.}  The clock terms, the input and output
terms, and the additional operator are sums of $|e\rangle\langle e|$ over
distinct bit strings $e$, which are pairwise orthogonal.  In each case $e$
has the same length and Hamming weight as the indicated catalogue basis
vector:
\begin{itemize}
\item $H^{(1)}$, $H^{(2)}$: $00\mapsto|00\rangle$ and $11\mapsto|11\rangle$;
\item $H^{(3)}$: $e\in\{0101,0110,1001,1010\}\mapsto|1100\rangle$;
\item $H^{(4)}$: $e\in\{0111,1011\}\mapsto|0111\rangle$, and
$0011\mapsto|1100\rangle$;
\item $H^{(5)}$: $e\in\{0001,0010\}\mapsto|0001\rangle$, and
$0011\mapsto|1100\rangle$;
\item $H^{(6)}$: $1100\mapsto|1100\rangle$;
\item input term: $011\mapsto|011\rangle$; output term:
$101\mapsto|011\rangle$;
\item additional operator:
$|010\rangle\langle010|\otimes I
=|0100\rangle\langle0100|+|0101\rangle\langle0101|$, with
$0100\mapsto|0001\rangle$ and $0101\mapsto|1100\rangle$.
\end{itemize}

\emph{Propagation terms.}  For a unitary $g$ on $k$ data qubits and
$e\in\{0,1\}^k$, put
\[
  v_e:=|01\rangle|e\rangle-|10\rangle\,g|e\rangle .
\]
Because $|01\rangle\perp|10\rangle$ and $g$ is unitary,
$\langle v_e,v_{e'}\rangle
=\langle e,e'\rangle+\langle ge,ge'\rangle=2\delta_{ee'}$.  Expanding
$\sum_e\frac12|v_e\rangle\langle v_e|$ and using
$\sum_e|e\rangle\langle e|=I$ gives exactly $H_{\mathrm{prop},t}$.  So
$H_{\mathrm{prop},t}$ is the sum of the pairwise orthogonal rank-one
projections onto $\CC v_e$.

For $g=\mathrm{CNOT}$ on $(q_1,q_2,q_3,q_4)$, where $q_1q_2$ is the clock
particle, $q_3$ the control and $q_4$ the target:
\[
\begin{array}{lll}
  v_{00}=|0100\rangle-|1000\rangle
    & =\ |0010\rangle-|0001\rangle & \text{on }(q_3,q_4,q_2,q_1),\\
  v_{01}=|0101\rangle-|1001\rangle
    & =\ |0110\rangle-|0101\rangle & \text{on }(q_3,q_4,q_2,q_1),\\
  v_{10}=|0110\rangle-|1011\rangle
    & =\ -\bigl(|1101\rangle-|1010\rangle\bigr) & \text{on }(q_3,q_1,q_2,q_4),\\
  v_{11}=|0111\rangle-|1010\rangle
    & =\ |1011\rangle-|1100\rangle & \text{on }(q_3,q_1,q_2,q_4).
\end{array}
\]
Each identity is checked by writing the catalogue bits onto the listed
tuple; for example, $|0110\rangle$ written onto $(q_3,q_4,q_2,q_1)$ puts
$1$ on $q_4$ and $q_2$, giving $|0101\rangle$.

For $g=U_{\mathrm{Pyth}}$ on $(q_1,q_2,q_3)$, one has
$U_{\mathrm{Pyth}}|0\rangle=\frac15(3|0\rangle-4|1\rangle)$ and
$U_{\mathrm{Pyth}}|1\rangle=\frac15(4|0\rangle+3|1\rangle)$, hence
\[
  5v_0=-\bigl(-5|010\rangle+3|100\rangle-4|101\rangle\bigr),
  \qquad
  5v_1=-\bigl(-5|011\rangle+4|100\rangle+3|101\rangle\bigr),
\]
both on $(q_1,q_2,q_3)$ itself.

Finally, $H'_{\mathrm{prop},t}$ is the rank-one projection onto
$|a_2u\rangle-|da_1\rangle=|1000\rangle-|1101\rangle$, which equals
$-\bigl(|1011\rangle-|1000\rangle\bigr)$ written onto
$(q_1,q_3,q_2,q_4)$.

In every case at most four rank-one pieces occur.  \qed

\subsection{Counting}

In \cref{prop:source-catalogue-identification} the Hamiltonian has $L+1$
clock particles.  There are at most $\binom{L+1}{2}$ pairs $j<k$, each
contributing at most $4+3+3=10$ pieces to $H^{(3)}$, $H^{(4)}$,
$H^{(5)}$; at most $4(L+1)$ pieces come from propagation terms; at most
$2L$ from $H^{(6)}$ and $H'_{\mathrm{prop}}$; at most $2(N_d+1)$ from
input terms; and $3$ from $H^{(1)}$, $H^{(2)}$, $H_{\mathrm{out}}$.  The
total is at most
\[
  5(L+1)^2+6(L+1)+2N_d+5\le16\bigl((L+1)^2+N_d\bigr).
\]

\section{Dependencies and notation}
\label{app:ledger}

This appendix records the external inputs, finite local data, and notation
used in the hardness proof.

\subsection{External inputs}

The source Hamiltonian is obtained from Bravyi's perfect-completeness
history-state construction, in the form stated in
\cref{thm:cited-source-reduction}.  Standard finite-dimensional topology
supplies the long exact sequence of a pair and the join K\"unneth formula.
Rellich's theorem is the analytic input used in the filtered-Hodge
argument.  Real quantifier elimination is used only to make the fixed
catalogue constants effective.  Sparse Hamiltonian simulation and phase
estimation are used only for the independent $\mathsf{QMA}$-containment
proof.

\subsection{Finite local inputs}

For each of the thirteen canonical gadgets, C0-C4 are the complete finite
inputs used later:
\begin{enumerate}[label=(\roman*)]
\item C0 determines the relative topology of the actual completed local
clique complex;
\item C1 verifies the signed source cycle and the exact filling identity;
\item C2 verifies the induced bow-tie base and the private-vertex isolation
needed for global assembly;
\item C3 eliminates unwanted intermediate classes in the private-depth
filtration;
\item C4 identifies the unique surviving private class and proves injectivity
of the private-sector incidence matrix.
\end{enumerate}
These statements are transported to placed occurrences by
\cref{prop:certificate-relabel-invariance}; join padding is handled
separately by the uniform topological and spectral join theorems.

\subsection{Proof structure}

The finite source catalogue is encoded by the bow-tie complex.  Conditions
C0-C2 and the long exact sequence determine the local and global topological
action of the gadgets.  Conditions C3-C4, together with the filtered-Hodge
theorem, give the single-gadget spectral penalty; join padding and the
global estimate then yield the inverse-polynomial NO-case gap.  The final
isolated-vertex padding converts this statement into the fixed functions
$(k,\gamma)$ of the target problem.

\subsection{Frequently used notation}

The symbols $N$, $t$, and $m_a$ denote the number of encoded qubits, the
number of Hamiltonian-term occurrences, and the locality of catalogue type
$a$, respectively.  The base graph and clique complex are $G_N$ and $X_N$,
and $q_N=2N-1$ is the encoded homological degree.  The cycles $z_b$ encode
computational-basis states, while $j_a$ is the target cycle associated with
catalogue projector $\Pi_a$.  The local output clique complex is
$Y_a^{\mathrm{loc}}$ and its private vertices have weight $\lambda$.
The filtration counts private vertices.  The source promise gap is denoted
$g$, and the global construction uses $\lambda=c_0g/t$.

\section{Reproducing the certificate checks}
\label{app:reproducibility}

The companion archive \texttt{certificates.zip} contains the graph builders,
exact witness data, verification scripts, manifests, and logs used for C0-C4.
The required Python, NumPy, SciPy, and NetworkX versions are recorded in the
manifest.

The graph builders \path{pythag.py}, \path{fourq.py}, and
\path{compbasis.py} are Python ports of the corresponding Mathematica
constructions in King and Kohler's companion files
\cite{KingKohlerNotebooks}.  The archive also contains the orientation
adjustments, witnesses, and verification routines used here; their mapping
to the upstream notebooks is recorded in \path{PROVENANCE.md}.

For each catalogue type the archive records the graph and vertex maps, the
source cycle and folding data in \path{source_core_witnesses.json}, the
private core cycle and cocycle, the completed clique complexes and incidence
matrices, and the exact rank checks reported in
\cref{app:certified-catalogue-record}.  Running \texttt{bash RUN\_ALL.sh}
reconstructs these data, verifies C0-C4, and writes the logs to
\path{logs/}.

\printbibliography

\end{document}